\documentclass[a4paper]{article}
\usepackage[affil-it]{authblk}
\usepackage{amsfonts}     
\usepackage{graphicx}        
\usepackage{enumerate}
\usepackage[utf8]{inputenc}
\usepackage{amsmath}
\usepackage{mathtools}
\usepackage{bm}
\usepackage{comment}
\usepackage{algorithm}
\usepackage{algorithmic}
\usepackage{amssymb}
\usepackage{mathrsfs}
\usepackage{placeins}
\usepackage{float} 
\usepackage{hyperref}
\usepackage{xcolor}
\usepackage{amsthm}
\usepackage{physics}
\usepackage{enumitem}
\usepackage{mdframed}
\usepackage{cancel}
\usepackage{makecell}
\usepackage{tabularx}
\usepackage{subcaption}
\usepackage{dsfont}
\usepackage[numbers,sort&compress]{natbib}

\setlist{nosep}

\newtheorem{theorem}{Theorem}[section]
\newtheorem{lemma}[theorem]{Lemma}
\newtheorem{definition}[theorem]{Definition}
\newtheorem{corollary}[theorem]{Corollary}
\newtheorem{proposition}[theorem]{Proposition}
\newtheorem{example}[theorem]{Example}

\theoremstyle{definition}
\newtheorem*{remark}{Remark}

\newcommand{\vecI}{I}
\newcommand{\vecfactorial}{\vecI !}

\newcommand{\sumI}{\left| \vecI \right|}
\newcommand{\augI}{[\vecI, \vecI]}
\newcommand{\makemat}[1]{#1}
\newcommand{\makevec}[1]{#1}
\newcommand{\matA}{C}
\newcommand{\matB}{B}
\newcommand{\matI}{\mathbb{I}}

\newcommand{\indn}{n} 
\newcommand{\dimn}{N} 

\newcommand{\sumck}{\sum_{\sumI = 2k}}
\newcommand{\sumK}{\sum_{k=0}^K \sum_{\sumI = k}}
\newcommand{\sumInf}{\sum_{k=0}^\infty \sum_{\sumI = k}}

\newcommand{\sumdoublefinite}{\sum_{k_1, k_2 = 0}^{K}  \sum_{\substack{\vert I \vert = 2k_1 \\ \vert J \vert = 2k_2}}}

\newcommand{\bmin}{b_{\text{min}}}
\newcommand{\bmax}{b_{\text{max}}}

\newcommand{\multilog}{\textnormal{Li}}
\newcommand{\hfun}{\textnormal{Hi}}

\DeclareMathOperator{\Haf}{\textnormal{Haf}}
\DeclareMathOperator{\lHaf}{\textnormal{lHaf}}

\DeclareMathOperator{\Per}{\textnormal{Per}}

\begin{document}
	\title{Using Gaussian Boson Samplers to Approximate Gaussian Expectation Problems\thanks{This work is supported by ReNewQuantum from the ERC Synergy program under grant agreement No. 810573, CLUSTEC from the Horizon Europe program No. 101080173,  QCI.DK from the Digital Europe program No. 101091659, TopQC2X from Danmarks Innovationsfond Instrument Grand Solutions Award No. 3200-00030B, and the Simons Foundation collaboration grant on New Structures in Low-Dimensional Topology.}}
	\author{Jørgen Ellegaard Andersen$^{1,2}$ \& Shan Shan$^{1}$ \\
		{\small $^{1}$Center for Quantum Mathematics, University of Southern Denmark \\
		$^{2}$ Danish Institute of Advanced Study, University of Southern Denmark} \\
	}
	\maketitle
		
	\begin{abstract}
		Gaussian Boson Sampling (GBS) have shown advantages over classical methods for performing some specific sampling tasks. To fully harness the computational power of GBS, there has been great interest in identifying their practical applications. In this study, we explore the use of GBS samples for computing a numerical approximation to the Gaussian expectation problem, that is to integrate a multivariate function against a Gaussian distribution. We propose two estimators using GBS samples, and show that they both can bring an exponential speedup over the plain Monte Carlo (MC) estimator. Precisely speaking, the exponential speedup is defined in terms of the guaranteed sample size for these estimators to reach the same level of accuracy $\epsilon$ and the same success probability $\delta$ in the $(\epsilon, \delta)$ multiplicative error approximation scheme. We prove that there is an open and nonempty subset of the Gaussian expectation problem space for such computational advantage.
	\end{abstract}
	
	\allowdisplaybreaks
	\section{Introduction}
Gaussian Boson Samplers (GBS) are special-purpose quantum computers designed to perform specific types of sampling tasks. The photonic implementation of these platforms is most common, where squeezed states of light are injected into a linear optical network and get measured out by photon detectors. The outcome of measurement are photon-count patterns, that are distributed according to the specific probability distribution determined by the network's configuration. This sampling problem has been proven to be intractable for classical computers \cite{aaronson2011computational, hamilton2017gaussian, kruse2019detailed}. Moreover, it has been proven that, up to some mild random matrix theory conjectures, even sampling from an approximation to the GBS distribution is hard for classical computing \cite{aaronson2011computational, deshpande2022quantum}. In parallel to these theoretical developments, there has also been great progress in realizing GBS experimentally. Today, small-scale GBS devices, with up to 216 optical modes, have been realized \cite{zhong2020jiuzhang1, zhong2021phase, madsen2022quantum}, marking one of the earliest breakthroughs in the experimental demonstration of quantum advantages. Recent proposals  \cite{motes2014scalable, he2017time} using time-bin encoding further enhance the possibility of achieving scalable GBS devices in the near term.

The impressive advancements of GBS in both the theories and the experimental realizations have naturally led to questions about their practical utility. Applications of GBS have been proposed towards solving, e.g., vibronic spectroscopy in quantum chemistry \cite{huh2015boson,huh2017vibronic}, molecular docking \cite{banchi2020molecular}, dense subgraph sampling \cite{arrazola2018using}, perfect matchings \cite{bradler2018gaussian}, stochastic optimization \cite{arrazola2018quantum}, graph similarity \cite{bradler2021graph, schuld2019quantum}, and random point processes \cite{jahangiri2020point}.
Small-scale experiments and numerical simulations have shown promising results \cite{sparrow2018simulating, banchi2020molecular, bromley2020applications}. However, 
the theoretical validation for whether these proposals can outperform the best possible classical methods, or even surpass their chosen baseline at scale, is missing. Therefore, it remains as an open question whether GBS, compared to classical methods, can bring computational advantage for problems of practical interests.

In light of this, we propose to use GBS for approximating a certain type of integration problem that occurs in many applications. Precisely, we consider
\begin{align}
	\mu_{\Haf} = \int_{\mathbb{R}^N} f(x) h(x) \dd x,
	\label{eq:I}
\end{align}
where
\begin{equation}
	\begin{aligned}
		&\hspace{7pt} f(\makevec{x})= \sum_{k=0}^K \sum_{\sumI = k} a_{\vecI} \makevec{x}^{\vecI}, ~~~~ a_I \in \mathbb{R}\\
		\hspace{7pt} h(x) = (2&\pi)^{-N/2} (\det B)^{-1/2} \exp(-\frac{1}{2}x^\intercal B^{-1} x).
	\end{aligned}
	\label{eq:fh}
\end{equation}
Here, the function $f$ can be either a multivariate polynomial ($K < \infty$) or a formal power series ($K = \infty$), where
\begin{gather*}
	\vecI = (i_1, \dots, i_\dimn), ~~
	\sumI = i_1 + i_2 \dots + i_N, ~~
	x^I = x_1^{i_1} \dots x_N^{i_N}.
\end{gather*}
We note that \eqref{eq:I} for complex $f$ reduces to real $f$ by splitting the problem into its real and imaginary parts. 
The function $h$ is the Gaussian probability density function with zero mean and covariance matrix $B$ (real, symmetric and positive definite).
From a probability theory perspective, $\mu_{\Haf}$ can be thought of as the Gaussian expectation of $f$. Let $X$ be a random variable that takes values in $\mathbb{R}^N$ with probability density $h$. If $f$ is a measurable function, then
\begin{align*}
	\mu_{\Haf} = \mathbb{E}[f(X)],
\end{align*}
provided that $\vert \mu_{\Haf} \vert < \infty$. 
For notational convenience, we use $\mathcal{I}_{\Haf}$ for the name of the Gaussian expectation problem and use $\mu_{\Haf}$ to denote the actual solution of the problem. 

If we further assume that
\begin{align}
	\int_{\mathbb{R}^N} \sum_{k = 0}^K  \,  \sum_{\vert I \rvert = k}   \vert a_I x^I \vert \, h(x) \, \dd x < \infty, \tag{A1} \label{eq:absconv}
\end{align} 
then the triangle inequality implies that $\vert \mu_{\Haf} \vert < \infty$. 
Moreover, by Fubini's theorem, we are allowed to switch the order of the integral and the sum, and thus in \eqref{eq:I} integrate term by term given by \eqref{eq:fh}. As we shall see in Section \ref{sec:problem}, the assumption \eqref{eq:absconv} is essential for using Wick's theorem (see Theorem \ref{thrm:wick} below) to turn $\mu_{\Haf}$ into a weighted sum of the matrix hafnians; that is
\begin{align}
	\mu_{\Haf} = \sum_{k=0}^K \sum_{\sumI = 2k} a_{\vecI} \Haf(B_I).
	\label{eq:muhaf}
\end{align}
Without loss of generality, we have removed here all the terms that correspond to $\vert I \vert$ being odd, since the Gaussian expectation of an odd function is always 0. 

The new form \eqref{eq:muhaf} is solvable with GBS sampling if one additionally assumes that the eigenvalues of $B$ are strictly below 1. Due to their randomized nature, one can only use the GBS samples to compute a numerical approximation to $\mu_{\Haf}$ in a probabilistic sense.  Specifically, let $0 < \epsilon, \delta < 1$ be given. We seek an estimator $e$ such that
\begin{equation}
	P\left(\lvert \mu_{\Haf} - e \rvert > \epsilon \vert \mu_{\Haf} \vert \right) < \delta,
	\label{eq:mulerror}
\end{equation}
\noindent
provided that $\mu_{\Haf} \neq 0.$
In other words, we want to find $e$ such that the chance of $\lvert \mu_{\Haf} - e \rvert$ exceeding an $\epsilon$ factor of $\vert \mu_{\Haf} \vert$ is less than $\delta$. This is known as the $(\epsilon, \delta)$ multiplicative error approximation problem, and we denote it by $\mathcal{I}^\times_{\Haf}(\epsilon, \delta)$.

We now introduce an important family of special cases, which we denote by $\mathcal{I}_{\Haf^2}$. Consider
\begin{align}
	\mu_{\Haf^2} = \int_{\mathbb{R}^{2N}} {f}(p, q) {h}(p, q) \, \dd p \dd q
	\label{eq:I2}
\end{align}
where
\begin{equation}
	\begin{gathered}
		\hspace{7pt} {f}(p, q) = \sum_{k=0}^K \sum_{\sumI = k} a_{\vecI} p^I q^I, \\
	 {h}(p, q) = (2\pi)^{-N} (\det B)^{-1} \exp(-\frac{1}{2} [p, q]^\intercal (B \oplus B)^{-1} [p, q]).
	\end{gathered}
	\label{eq:fh2}
\end{equation}
Assume that 
\begin{align}
	\int_{\mathbb{R}^{2N}} \sum_{k = 0}^K  \,   \sum_{\vert I \rvert = k}  \vert a_{\vecI} p^I q^I \vert \, h(p, q) \, \dd p \dd q < \infty. \tag{A1'} \label{eq:absconvsq}
\end{align}
As before, $\vert \mu_{\Haf^2} \vert < \infty$ is implied by the triangle inequality. By Wick's theorem, one obtains
\begin{align}
	\label{eq:muhafsq}
	\mu_{\Haf^2} = \sum_{k=0}^K \sum_{\sumI = 2k} a_{\vecI} \Haf(B_I)^2.
\end{align}
The $(\epsilon, \delta)$ multiplicative error approximation problem $\mathcal{I}^\times_{\Haf^2}(\epsilon, \delta)$ is defined analogously to $\mathcal{I}^\times_{\Haf}(\epsilon, \delta)$,  provided that $\mu_{\Haf^2}  \neq 0$. 

To solve for $\mathcal{I}^\times_{\Haf}(\epsilon, \delta)$ and $\mathcal{I}^\times_{\Haf^2}(\epsilon, \delta)$, we introduce two estimators with GBS samples. 
The first estimator, GBS-I, is based on a well-known statistical method called importance sampling, and it is tailored to $\mathcal{I}^\times_{\Haf^2}(\epsilon, \delta)$. The key insight is to use the GBS sampling distribution as the importance sampling distribution. For any $ I \in \mathbb{N}^N$, where  $\mathbb{N}$ denotes the set of natural numbers $\{0, 1, 2, \dots \}$, the GBS sampling density is given by
\begin{align}
	P(I) = p_I = \frac{d}{I!}  \Haf(B_I)^2,
	\label{eq:gbsfirstappear}
\end{align}
where $I! = i_1! i_2! \dots i_N!$
and the constant $d$ is computed from the eigenvalues of $B$ (see \eqref{eq:ddef}). The GBS-I estimator is defined by
\begin{align*}
	\mathcal{E}^\textnormal{GBS-I}_n = \frac{1}{n} \sum_{i = 1}^n \frac{ I_i!}{d} a_{I_i}.
\end{align*}
Here, $n$ represents the number of GBS samples used for computing the estimator. 
The GBS-I estimator is unbiased, meaning that $\mathbb{E}[\mathcal{E}_n^\textnormal{GBS-I}] = \mu_{\Haf^2}$ for all $n$. By the weak law of large numbers (WLLN), we get
\begin{equation}
	\lim_{n \rightarrow \infty} P( \vert \mathbb{E}[\mathcal{E}^\textnormal{GBS-I}_n] - \mu_{\Haf^2} \vert > \epsilon) = 0.
	\label{eq:asymbehavior} 
\end{equation}
Thus, GBS-I solves $\mathcal{I}^\times_{\Haf^2}(\epsilon, \delta)$ for sufficiently large $n$. If the variance of $\mathcal{E}^\textnormal{GBS-I}_n$ is finite, we can further use Chebyshev's inequality to derive a \textit{guaranteed sample size} $n^{\textnormal{GBS-I}}_{\Haf^2}$, such that if $n \geq n^{\textnormal{GBS-I}}_{\Haf^2}$, then GBS-I is guaranteed to solve $\mathcal{I}^\times_{\Haf^2}(\epsilon, \delta)$ (see Lemma \ref{lem:importbest}). 
\vspace{0.5em}

The second estimator, GBS-P, is a direct probability estimation, and it is applicable to the more general $\mathcal{I}^\times_{\Haf}(\epsilon, \delta)$ problem. The key difference here is that $\mu_{\Haf}$ involves the hafnians instead of the square of hafnians as given in the GBS distribution. To address this, we first approximate $p_J$ for each $J$ by counting the occurrences of $J$ in the observed GBS samples:
\begin{equation*}
	S^{(J)}_n = \sum_{l=1}^n 1_J (I_l), ~~~~~~~~~ 1_J(I) = \begin{cases}
		1 & I = J \\
		0 & I \neq J
	\end{cases}.
\end{equation*}
and then take a weighted average of the square root of ${S^{(J)}_n}/{n}$ to define
\begin{equation*}
	\mathcal{E}^{\textnormal{GBS-P}}_n = \sum_{k = 0}^K \sum_{\vert J \vert = k} a_J \sqrt{\tfrac{I!}{d}} \sqrt{\tfrac{S^{(J)}_n}{n}}.
\end{equation*}
Unlike the importance estimator, the GBS-P estimator is biased, and therefore we cannot use WLLN to show that GBS-P solves $\mathcal{I}^\times_{\Haf}(\epsilon, \delta)$ with a large enough $n$. However, we shall see in Section \ref{sec:gbs-estimators} (assuming all entries of $B$ are positive and all $a_I >0$) that the bias disappears asymptotically. In 
Theorem \ref{thrm:probest}, we compute an upper bound for the squared error, which decreases as $n$ grows. 
From Markov's inequality, a guaranteed sample size $n^{\textnormal{GBS-P}}_{\Haf}$ can be obtained for using GBS-P  to solve $\mathcal{I}^\times_{\Haf}(\epsilon, \delta)$.

\vspace{0.5em}
Our first set of results provide estimates for the guaranteed sample size $n^{\textnormal{GBS-I}}_{\Haf^2}$ and $n^{\textnormal{GBS-P}}_{\Haf}$. To state these results, we need the following notations. Let $\bmax$ denote the maximum of the absolute value of the entries of $B$ and let $\bmin$ denote the minimum of the entries of $B$. 
Recall the polylogarithm $\multilog_{s}$ function of order $s$ is given by
\begin{align}
	\multilog_{s}(z) = \sum_{k = 1}^\infty \frac{z^k}{k^s}, ~~~~ \text{ for }\vert z \vert < 1.
\end{align}
When $s = 0$, $\multilog_{s}(z)$ is the standard geometric series and has a closed-form expression. We introduce the notation $\multilog_{s, K}(z)$ to be the truncation of $\multilog_{s}(z)$ up to the $K$-th term,
\begin{align}
	\multilog_{s, K}(z) = \sum_{k = 1}^K \frac{z^k}{k^s} = z + \frac{z^2}{2^s} + \frac{z^3}{3^s} + \dots + \frac{z^K}{K^s}.
\end{align}
If $K < \infty$, then $\multilog_{s, K}(z)$ is defined for all $z \in \mathbb{R}$ with the special case $\multilog_{s, 0}(z) = 0$. If $K = \infty$, we have that $\multilog_{s, K}(z) = \multilog_{s}(z)$. The next two theorems provide estimates for $n^{\textnormal{GBS-I}}_{\Haf^2}$ and $n^{\textnormal{GBS-P}}_{\Haf}$ using the polylogs. 

\begin{theorem}
	\label{main:pp1}
	Suppose $ \vert \mu_{\Haf^2} \vert > \mu_0$ for some $\mu_0 > 0$. Suppose there exists $c_1, c_2, q_1, q_2, \gamma_1$ and $\gamma_2$ such that for all $1 \leq k \leq K$
	\begin{gather*}
	\sum_{\vert I \vert = 2k}  \vert  a_I \vert \leq c_1 \frac{k^{q_1} \gamma_1^{k}}{(2k)!},
	\end{gather*}
	and 
	$$ \sum_{\vert I \vert = 2k} a^2_I I!\leq c_2 \frac{k^{q_2} \gamma_2^{k}}{(2k)!}.$$ 
	If  $K = \infty$, we also require $\gamma_1 < \frac{1}{\bmax^2}$ and $\gamma_2 < \frac{1}{\bmax^2}$.
	Then, 
	\begin{align*}
	n_{\Haf^2}^{\textnormal{GBS-I}} \leq \frac{1}{d \delta \epsilon^2} \frac{1}{\mu_0^2} \left( a_0^2 + \frac{c_2}{\sqrt{\pi}}  \multilog_{\frac{1}{2}-{q_2}, K}({\gamma_2} \bmax^2) \right) - \frac{1}{\delta \epsilon^2}.
	\end{align*}
\end{theorem}

\begin{theorem}
	\label{main:pp2}
	Suppose all $a_I \geq 0$, $\bmin >0$ and $ \vert \mu_{\Haf} \vert > \mu_0$ for some $\mu_0 > 0$. Suppose there exists $c_1, c_2, q_1, q_2, \gamma_1$ and $\gamma_2$ such that for all $1 \leq k \leq K$
	\begin{gather*}
		\sum_{\vert I \vert = 2k} a_I \leq c_1 \frac{k^{q_1} \gamma_1^{k}}{k!},
	\end{gather*}
	and 
	\begin{align*}
		\sum_{\vert I \vert = 2k} a_I I! \leq c_2 k^{q_2} \gamma_2^k k!.
	\end{align*}
	If $K = \infty$, we further require $\gamma_1 < \frac{1}{2\bmax}$ and $\gamma_2 < {2\bmin}.$
	Then, 
	\begin{align*}
		n_{\Haf}^{\textnormal{GBS-P}} \leq \frac{1}{d \delta \epsilon^2} \frac{1}{\mu_0} \left( a_0 +  c_2 \sqrt{\pi} e^{\frac{1}{6}-\frac{1}{25}} \multilog_{-\frac{1}{2} -q_2, K} \left(\frac{\gamma_2 }{2\bmin} \right)  \right) - \frac{1}{\delta \epsilon^2}.
	\end{align*}
\end{theorem}

In these estimate, we have assumed that some lower bound of $ \vert \mu_{\Haf} \vert$ or $ \vert \mu_{\Haf^2} \vert$ are known. In Corollaries \ref{cor:mainpp1} and \ref{cor:mainpp2}, we show how to approximate them by imposing additional assumptions on the $a_I$'s. The purposes of these estimates are twofold. First, they give insights into the amount of resources required for computing the GBS estimators in practice, as the polylogs are easy to evaluate. Second, the techniques used to prove these results are also crucial for proving the main theorems. 

The main contribution of this paper is to address how efficient are the GBS estimators. We use a plain Monte Carlo (MC) estimator as a baseline for comparison. The plain MC estimator solves $\mathcal{I}^{\times}_{\Haf}(\epsilon, \delta)$ or $\mathcal{I}^{\times}_{\Haf^2}(\epsilon, \delta)$ by taking the average of $f$ evaluated at $n$ random samples from the multivariate Gaussian distribution $h$. Precisely, we define
\begin{align*}
	\mathcal{E}^{\textnormal{MC}}_n = \sum_{i=1}^n \frac{1}{n} f(X_i)
\end{align*}
with $X_1, X_2, \dots, X_n$ being i.i.d samples with probability density $h$. 
Using Chebyshev's inequality, one can also compute the guaranteed sample size $n^{\textnormal{MC}}_{\Haf}$ or $n^{\textnormal{MC}}_{\Haf^2}$. 
Our mains theorems below compare these guaranteed sample size.
\begin{theorem}
	\label{thrm:mani1}
	For any $0 < \epsilon, \delta < 1$, $s_1, s_2>0$,  $N$ and sufficiently large $K$, there exists a non-empty and open subset of the problem space, where
	\begin{align*}
		s_2 \exp( s_1 \, n_{\Haf^2}^{\textnormal{GBS-I}} ) \leq n_{\Haf^2}^{\textnormal{MC}} < \infty.
	\end{align*}
\end{theorem}
\noindent
The precise definition of the problem space and the conditions for $a_I$'s and $B$ are given in Theorem \ref{thrm:mainI1}. Note that Theorem \ref{thrm:mani1} establishes a comparison between GBS-I and MC for fixed values of $\epsilon$ and $\delta$. If instead of exponential advantage, we are only looking for when $n_{\Haf^2}^{\textnormal{GBS-I}}$ is less than $n^{\Haf^2}_{\textnormal{MC}}$, which may be more relevant for practical applications, then we can show that there are a plethora of problems where GBS-I outperforms MC uniformly across $0< \epsilon, \delta < 1$. 
\begin{theorem}
	\label{thrm:mani12}
	For any $c>0$,  $N$ and sufficiently large $K$, there exists a non-empty and open subset of the problem space, where
	\begin{align*}
		n_{\Haf^2}^{\textnormal{GBS-I}} \leq  c \, n_{\Haf^2}^{\textnormal{MC}} < \infty
	\end{align*}
	holds for all $0< \epsilon, \delta < 1$.
\end{theorem}
\noindent
The precise statement is given in Theorem \ref{cor:mainI1}. We can also study the behavior of these guaranteed sample size as $N$ increases. Notably there are problems where $n_{\Haf^2}^{\textnormal{GBS-I}}$ grows only polynomially with $N$ and $n_{\Haf^2}^{\textnormal{MC}}$ grows exponentially with $N$.
\begin{theorem}
	\label{thrm:mani13}
	For any $0 < \epsilon, \delta < 1$, $p >0$,  $N$ sufficiently large, and $K \geq \zeta N^2$ with $\zeta >0$, there exists a non-empty and open subset of the problem space, where
	\begin{align*}
		&n_{\Haf^2}^{\textnormal{GBS-I}} = \frac{1}{\epsilon^2 \delta}O(N^{3+p}), \\
		&n_{\Haf^2}^{\textnormal{MC}} = \frac{1}{\epsilon^2 \delta} \Omega(e^{c_RN^2}),
	\end{align*}
	for some constant $c_R>0$. 
\end{theorem}
\noindent
Here, we have used the standard big $O$ and big $\Omega$ notations. The constants in these notations together with $c_R$ are explicitly provided in Theorem \ref{cor:mainI1N}.

The same results hold when comparing GBS-P with MC, and we give the precise statements in Theorems \ref{thrm:mainI2}, \ref{cor:mainI2} and \ref{cor:mainI2N}. 
\begin{theorem}
	\label{thrm:mani2}
	For any $0 < \epsilon, \delta < 1$, $s_1, s_2>0$, $N$ and sufficiently large $K$, there exists a non-empty and open subset of the problem space, where
	\begin{align*}
		s_2 \exp(s_1 \, n_{\Haf}^{\textnormal{GBS-P}} ) \leq n_{\Haf}^{\textnormal{MC}} < \infty.
	\end{align*}
\end{theorem}

\begin{theorem}
	\label{thrm:mani22}
	For any $c>0$,  $N$ and sufficiently large $K$, there exists a non-empty and open subset of the problem space, where
	\begin{align*}
		n_{\Haf}^{\textnormal{GBS-P}}  \leq c\, n_{\Haf}^{\textnormal{MC}} < \infty
	\end{align*}
	holds for all $0< \epsilon, \delta < 1$.
\end{theorem}

\begin{theorem}
	\label{thrm:mani23}
	For any $0 < \epsilon, \delta < 1$, $p >0$,  $N$ sufficiently large, and $K \geq \zeta N^2$  with $\zeta >0$, there exists a non-empty and open subset of the problem space, where
	\begin{align*}
		&n_{\Haf}^{\textnormal{GBS-P}} = \frac{1}{\epsilon^2 \delta}O(N^{3+p}), \\
		&n_{\Haf}^{\textnormal{MC}} = \frac{1}{\epsilon^2 \delta} \Omega(e^{c_RN^2}),
	\end{align*}
	for some constant $c_R>0$.
\end{theorem}

\noindent
These results establish that the GBS estimators achieve a provable exponential advantage over MC for computing certain Gaussian expectation problems. In Section 5, we provide the precise conditions guaranteeing this advantage. A natural follow-up question arises -- how large is the problem space where the GBS estimators outperform MC?
In \cite{shan2025companion}, we improve the the GBS estimators by optimizing the average photon number in the GBS distribution to align with the specific Gaussian expectation problem. We update the guaranteed sample size estimates for these improved algorithms and quantify the proportion of the problem space where they outperform MC. Numerical results show that this advantage covers a substantial portion of the full problem space and the advantage is significant.

The remainder of this paper is organized as follows. Section 2 describes the general setup of the Gaussian expectation problem, which enables applications of different variations of GBS, such as Boson Sampling (BS). For BS, we discuss a complex version of the integration problem in Appendix \ref{sec:boson}. Appropriate analogous results can be established for BS, and we shall return to this in future work separately. In Section 3, we give the construction of GBS-I and GBS-P, and prove their convergence properties. In Section 4, we give proofs for Theorems \ref{main:pp1} and \ref{main:pp2} and give a detailed account to motivate the conditions described as in the theorems. In Section 5, we prove the main results and give precise statements for Theorems \ref{thrm:mani1} to \ref{thrm:mani23}. In Section 6, we illustrate the performance of the GBS estimators with numerical simulations. Discussions such as the GBS hardware preparation, total computation cost of the GBS estimators, the choice of plain MC estimator as a baseline method, and hardness of the Gaussian expectation problems are made in the Appendices. Throughout, we consider samples from the noiseless GBS distribution; we will present our result on the impact of noise and errors in a forthcoming paper.

	\section{Background}
\label{sec:problem}

\subsection{Gaussian expectation; Wick's theorem} 
We have already defined the Gaussian expectation problem in the introduction. Below we show how it can be transformed into a form that is compatible with GBS sampling. A famous result in quantum field theory due to Gian-Carlo Wick is the following.

\begin{theorem}[Wick \cite{wick1950evaluation}]
	\label{thrm:wick}
Let $x = (x_1, \dots, x_N)$ be standard coordinates on $\mathbb{R}^{N}$ and $\vecI = (i_1, \dots, i_\dimn)$. Let $h$ be the zero-mean Gaussian distribution defined as in \eqref{eq:fh}. Then,
	\begin{equation*}
	 \int_{\mathbb{R}^N} x^I h(x) \, \dd x = \Haf(B_I).	
	\end{equation*}
\end{theorem}
\noindent
In words, the expected value of $x^I = x_1^{i_1} \dots x_N^{i_N}$ against a multivariate Gaussian distribution can be explicitly computed from the entries of $B$ by a matrix function called the \textit{hafnian} (see definition below) and $B_I$ is the corresponding sub- or super-matrix of the covariance matrix $B$ determined by the $N$-tuple $\vecI$. Precisely, for each column $m$, we construct a new vector by repeating the $B_{m, n}$ entry $i_n$ many times and then we repeat this new vector $i_m$ many times. 

The definition of the matrix hafnian function is provided below.

\begin{definition}[Hafnian]
	\label{def:haf}
	Let  $M$ be a matrix of size $2m \times 2m$ for some $m \geq 1$. Then,
	\begin{equation*}
	\Haf(\makemat{M}) = \frac{1}{m!2^m}\sum_{\sigma \in \mathcal{S}_{2m}}\prod_{j=1}^m \makemat{M}_{\sigma(2j-1), \sigma(2j)},
	\end{equation*}
	where $\mathcal{S}_{2m}$ denotes the symmetric group on $2m$ elements. 
	If $M$ is an empty matrix, then $\Haf(M) = 1.$ If the size of the matrix is odd, then $\Haf(M) = 0.$
\end{definition}

\noindent
From Wick's theorem, we can rewrite \eqref{eq:I} as
\begin{align*}
	\mu_{\Haf} & = \int_{\mathbb{R}^N} \sumK a_{\vecI} \makevec{x}^{\vecI} h(x) \, \dd x \\
	& = \sumK a_{\vecI} \int_{\mathbb{R}^N} \makevec{x}^{\vecI} h(x) \, \dd x \\
	& =  \sumK a_{\vecI}  \Haf\left( \matB_{\vecI} \right) \\
	& =  \sum_{k=0}^K \sum_{\sumI = 2k} a_{\vecI}  \Haf\left( \matB_{\vecI} \right).
\end{align*}
Interchanging the sum and the integral is allowed by \eqref{eq:absconv} and Fubini's theorem.
Since $\Haf(B_I) = 0$ when $\vert I \vert$ is odd, we can remove those corresponding terms in the sum. Therefore, we can compute $\mu_{\Haf}$ by \eqref{eq:muhaf}. 

Similarly, $\mu_{\Haf^2}$ can be computed using \eqref{eq:muhafsq}. Since $h(p, q)$ is the standard zero mean Gaussian distribution with covariance function $B \oplus B$. Using Wick's theorem, we get
\begin{equation*}
	\int_{\mathbb{R}^{2N}} p^I q^I h(p,q) \, \dd p \, \dd q = \Haf \left( \begin{bmatrix}
		B_I & 0 \\
		0 & B_I
	\end{bmatrix}\right) = \Haf(B_I)^2.
\end{equation*}
From \eqref{eq:absconvsq} and Fubini's theorem, we get $\mu_{\Haf^2}$ is computed by \eqref{eq:muhafsq}.

\subsection{Sampling distribution of GBS}
\label{subsec:gbsdist}
In this section, we give a short review about the GBS sampling distribution. Mathematically speaking,  the sampling distribution of GBS is obtained by measuring an arbitrary Gaussian states against multi-mode number operators. Gaussian states are those quantum states whose Wigner functions are Gaussian. Any $N$-mode Gaussian state $\hat{\rho}$ can be uniquely characterized by its first moment ${\xi} \in \mathbb{C}^{2N}$ and second moment $\Sigma \in \mathbb{C}^{2N \times 2N}$ of annihilation and creation operators. The outcome space upon measurement is $\mathbb{N}^N$. For convenience, let us first introduce the notations 
\begin{align}
	\makemat{\Sigma}_Q = \makemat{\Sigma} + \frac{1}{2}\matI_{2N}, ~~~~ C = 
	\begin{bmatrix}
		0 & \matI_N \\
		\matI_N & 0
	\end{bmatrix} \left[ \matI_{2N} - \makemat{\Sigma}_Q^{-1} \right].
	\label{eq:Cmat}
\end{align}
In the simplest setup, if $\hat{\rho}$ is a pure Gaussian state with zero displacement $\xi = 0$,  then the matrix $C$ has a block diagonal form $C = M \oplus M^*$ with $M \in \mathbb{C}^{N \times N}$ and $-1 < \textnormal{spec}(M) < 1$ (see Appendix \ref{thrm:covform}). In this case, the probability distribution of obtaining $I= (i_1, \dots, i_N)$ is described as in \cite{hamilton2017gaussian} (see also \eqref{eq:gbsfirstappear}), where 
\begin{equation}
	\begin{gathered}
		P(I) = p_I =  \frac{d}{\vecI!} \vert \Haf(M_I) \vert^2, ~~~~d = 1/\sqrt{\det \Sigma_Q}, ~~~~ I! = i_1! \dots i_N!
	\end{gathered}
	\label{eq:gbshaf1}
\end{equation}

\noindent
The constant $d$ is the normalization factor such that the probabilities sum to 1. When $M$ is real, $d$ can also be expressed in terms of the eigenvalues of $M$. Let $\lambda_1, \dots, \lambda_N$ denote the eigenvalues of $M$. Then,
\begin{align}
	d = \prod_{j=1}^\dimn \sqrt{1 - \lambda^2_j}.
	\label{eq:ddef}
\end{align}
Formula \eqref{eq:ddef} is proved in Appendix \ref{lem:dform}.

More generally, if $\hat{\rho}$ is an arbitrary pure Gaussian state with nonzero displacement $\xi \neq 0$, then the sampling distribution is provided in \cite{quesada2019simulating}, where
\begin{equation}
	\begin{gathered}
		P(I) = p_I =  \frac{d}{\vecI!} \lvert \lHaf(\mathcal{M}_I) \rvert^2, ~~~~
		{d} = \frac{\exp{-\frac{1}{2} {\xi}^\dagger \Sigma_Q^{-1} {\xi}}}{\sqrt{\det \Sigma_Q}}
	\end{gathered}
	\label{eq:gbshaf2}
\end{equation}
Here, $\lHaf$ denotes the loop hafnian function and
\begin{align*}
\mathcal{M}_{i,j} = \begin{cases}
		M_{i,j} & \text{if } i \neq j \\
		\gamma_i & \text{if } i = j
	\end{cases}, ~~~~
	\gamma =  \xi^\dagger \makemat{\Sigma}_Q^{-1}.
\end{align*}

Other GBS variations, such as using threshold detectors for measurement, may produce photon patterns consisting of only binary strings. Specifically, the sampling space for an $N$-mode Gaussian state becomes $\{0,1\}^N$, and the probability distribution is written in terms of a new marix function called the Torontoian. See \cite{quesada2018gaussian} for details. Furthermore, Boson Sampling as proposed in \cite{aaronson2011computational} uses single photons as inputs, generating photon patterns that sum to the input photon number. The probability distribution is expressed in terms of the matrix permanent.

In general, we write the abstract sampling distribution to be
\begin{align}
	P(I) = p_I = w_I \phi(M, I), ~~~~~I \in \mathcal{S}
	\label{eq:Pimp}
\end{align}
where $\mathcal{S}$ denotes an arbitrary countable sampling space, $w_I$ is some non-negative scalar value depending on $I$, and $\phi$ is some non-negative real-valued matrix function such that
\begin{align*}
	\sum_{I\in \mathcal{S}} w_I \phi(w, I) = 1.
\end{align*}
For a function $a : \mathcal{S} \rightarrow \mathbb{C}$ such that
\begin{align*}
	\sum_{I\in \mathcal{S}} \vert a_I \vert w_I \phi(M, I) < \infty,
\end{align*}
we define the expectation value to be
\begin{align*}
	\mathbb{E}(a) = \sum_{I\in \mathcal{S}} a_I w_I \phi(M, I).
\end{align*}

	\section{Two estimators using GBS samples}
\label{sec:gbs-estimators}

In this section, we describe two estimators using GBS samples. Since most of the results are true when using the abstract sampling distribution given as in \eqref{eq:Pimp} with an arbitrary matrix function, we will introduce the estimators in this generality. First, we need a definition for the generalized Gaussian expectation problem.

\begin{definition}
	\label{def:frak}
	Let $\{a_I\}_{I \in \mathcal{S}}$ be a collection of real coefficients with index set $\mathcal{S}$ and let $M \in \mathbb{C}^{N \times N}$. For a complex-valued matrix function $\psi$,
	we assume 
	\begin{align*}
		\sum_{I\in \mathcal{S}} \vert a_I \vert \vert  \psi(M, I) \vert < \infty,
	\end{align*}
	and we define $\mu_\psi$ to be
	\begin{equation}
		\label{eq:defgenerat}
		\mu_\psi = \sum_{I \in \mathcal{S}} a_{\vecI} \psi(M, I).
	\end{equation}
\end{definition}
\noindent
Note that if the $a_I$'s in \eqref{eq:defgenerat} are complex, we can reduce the computation of $\mu_\psi$ to the computation of its real and imaginary part. 
We define the approximation problem $\mathcal{G}^\times_\psi(\epsilon, \delta)$ analogously to \eqref{eq:mulerror}.

\begin{definition}
	\label{def:mulerrorapprox}
	Assuming $\mu_\psi  \neq 0$, for any given $0 < \epsilon, \delta < 1$ we seek an estimated value $e$ such that
	\begin{equation*}
		P\left(\lvert \mu_\psi - e \rvert > \epsilon \vert \mu_\psi  \vert \right) < \delta.
	\end{equation*}
\end{definition}

\subsection{The importance estimator}
\label{subsec:imsamp}
We first introduce an estimator based on the well-known {\it importance sampling} method, and hence the name {\it importance estimator}. If $\phi$ from the abstract sampling distribution \eqref{eq:Pimp} and $\psi$ from the generalized Gaussian expectation problem are the same, i.e., $\phi = \psi$, then, 
\begin{equation*}
	\mu_\psi = \mu_\phi = \sum_{I \in \mathcal{S}} a_I \phi(M, I) =  \sum_{I \in \mathcal{S}} \frac{a_I}{w_I} p_I.
\end{equation*}
We now define the importance estimator as follows. 

\begin{definition}
	Let $\vecI_1, \vecI_2, \dots$ be i.i.d with probability density $P$ defined as in \eqref{eq:Pimp} and let
	\begin{equation*}
		Z_i = \frac{a_{\vecI_i}}{w_{I_i}}.
	\end{equation*} 
	The \textit{importance estimator} is defined as
	\begin{equation*}
		\mathcal{E}^{\textnormal{GBS-I}}_n = \frac{1}{n} \sum_{i=1}^n Z_i \label{eq:gbsest1}.
	\end{equation*}
\end{definition}

It is straightforward that $\mathcal{E}^{\textnormal{GBS-I}}_n $ is unbiased, since
\begin{equation*}
	\mathbb{E} \left[ \mathcal{E}^{\textnormal{GBS-I}}_n \right] = \frac{1}{n} \sum_{i=1}^n  \mathbb{E}[Z_i] = 
	\sum_{I \in \mathcal{S}} \frac{a_{\vecI}}{w_I} p_I = \mu_\phi. 
\end{equation*} 
Then, using WLLN, one gets that
for any $\epsilon >0$, $$\displaystyle \lim_{n \rightarrow \infty} P(\vert \mathcal{E}^{\textnormal{GBS-I}}_n - \mu_\phi \vert \geq \epsilon ) = 0.$$
If one further assumes that  $\text{Var}(\mathcal{E}^{\textnormal{GBS-I}}_n) < \infty$, which is equivalent to assuming
\begin{align*}
	\sum_{I \in \mathcal{S}} \frac{a^2_{\vecI}}{w_I} \phi(M, I) < \infty,
\end{align*}
then a lower bound for $n$ can be derived using Chebyshev's inequality. Precisely,

\begin{lemma}
	\label{lem:importbest}
	The importance estimator $\mathcal{E}^{\textnormal{GBS-I}}_n$ solves $\mathcal{G}_\phi^\times(\epsilon, \delta)$ with
	\begin{equation*} 
		n \geq  \frac{Q^{\textnormal{GBS-I}}_{\phi} - \mu_\phi^2}{\delta \epsilon^2 \mu_\phi^2} \equiv n_{\phi}^{\textnormal{GBS-I}},
	\end{equation*}
	provided that 
	\begin{align}
		Q^{\textnormal{GBS-I}}_{\phi} < \infty, \label{eq:QIfinite} \tag{A2'}
	\end{align}
	where
	\begin{equation*}
		Q^{\textnormal{GBS-I}}_{\phi} = \sum_{I \in \mathcal{S}}  \frac{a^2_{\vecI}}{w_I} \phi(M, I).
	\end{equation*}
\end{lemma}

\begin{proof}
	The proof is a straightforward application of Chebyshev's inequality, and we just need to compute the variance of $\mathcal{E}^{\textnormal{GBS-I}}_n$. 
	\begin{align*}
		\textnormal{Var}(\mathcal{E}^{\textnormal{GBS-I}}_n) & = \frac{1}{n^2} \sum_{i = 1}^n \textnormal{Var}(Z_i) \\
		& = \frac{1}{n} \textnormal{Var}(Z_1) \\
		& = \frac{1}{n} \mathbb{E}[\left(  Z_1 - \mathbb{E}[Z_1] \right)^2] \\
		& = \frac{1}{n} \left( \mathbb{E}[Z_1^2] - \mu_\phi^2 \right) \\
		& = \frac{1}{n} \left( \left( \sum_{I \in \mathcal{S}} \frac{a^2_{\vecI}}{w_I}\phi(M, I)\right) - \mu_\phi^2 \right)\\
		& = \frac{1}{n} \left( Q^{\textnormal{GBS-I}}_{\phi} - \mu_\phi^2 \right) < \infty.
	\end{align*}
	Using Chebyshev's inequality, we get
	\begin{equation*}
		P\left( \vert  \mathcal{E}^{\textnormal{GBS-I}}_n - \mu_\phi  \vert > \epsilon \vert \mu_\phi \vert \right) < \frac{ Q^{\Haf}_\text{MC} - \mu_\phi^2 }{n \epsilon^2 \mu_\phi^2}.
	\end{equation*}
	Setting the right-hand side to be less than $\delta$, we obtain the desired lower bound for $n$.
\end{proof}

\subsection{The probability estimator}
\label{sec:probestimator}
In this section, we describe a second technique, based on a straightforward probability estimation, and hence the name {\it probability estimator}.  If $\phi$ from the abstract sampling distribution \eqref{eq:Pimp} and $\psi$ from the generalized Gaussian expectation problem satisfy $\psi(M, I) = \sqrt{\phi(M, I)}$ and $\phi(M, I) \geq 0$ for all $I$. Then, 
\begin{equation*}
	\mu_\psi = \mu_{\sqrt{\phi}} = \sum_{I \in \mathcal{S}} a_I \sqrt{\phi(M, I)} = \sum_{I \in \mathcal{S}} a_J \sqrt{\frac{p_I}{w_J}}
\end{equation*}

\begin{definition}
	Let $\vecI_1, \vecI_2, \dots $ be i.i.d. with probability distribution $P$ defined as in \eqref{eq:Pimp}. 
	For each $J \in \mathbb{N}^N$, we define the random variable $S^{(J)}_n$ which counts the number of occurrences of $J$ in $n$ draws; that is
	\begin{equation*}
		S^{(J)}_n = \sum_{l=1}^n 1_J (I_l), ~~~~~~~~~ 1_J(I) = \begin{cases}
			1 & I = J \\
			0 & I \neq J
		\end{cases}.
	\end{equation*}
	We define the probability estimator to be
	\begin{equation}
		\mathcal{E}^{\textnormal{GBS-P}}_n = \sum_{J \in \mathcal{S}} \alpha_J \sqrt{{S^{(J)}_n}/{n}}, ~~~~\alpha_J = a_J \sqrt{\frac{1}{w_J}}.
		\label{eq:estsqrt}
	\end{equation}
\end{definition}

\begin{remark}
	It is worthwhile to note that $\mathcal{E}^{\textnormal{GBS-P}}_n$ is always finite regardless of $K$ being finite or infinite, since we can always use a dynamically sized lookup table for tracking $S^{(J)}_n$. Upon each draw, this table is updated by the following rule. If a draw yields $J$ and there is no existing entry for $J$ in the table, we add the key $J$ with a value of 1. If $J$ is already a key, its value is incremented by 1. Consequently, the size of the lookup table is bounded by $ n_\text{GBS-P}$, and computing the sum in $\mathcal{E}^{\textnormal{GBS-P}}_n$ involves at most $ n_\text{GBS-P}$ terms. 
\end{remark}

\noindent
Unlike the importance estimator, $\mathcal{E}^{\textnormal{GBS-P}}_n$ is biased, meaning
\begin{equation*}
	\mathbb{E}[\mathcal{E}^{\textnormal{GBS-P}}_n] = \sum_{J \in \mathcal{S}} \alpha_J \mathbb{E}[\sqrt{{S^{(J)}_n}/{n}}] \neq \sum_{J \in \mathcal{S}} \alpha_J \sqrt{\mathbb{E}[{S^{(J)}_n}/{n}]} = \mu_{\sqrt{\phi}}
\end{equation*}
for some $n$.
The bias is due to the non-linearity of the square root function, and we can no longer use the WLLN. 

The next result computes the expectation of the squared error. We use it to show that $\mathcal{E}^{\textnormal{GBS-P}}_n$ converges to $\mu_\phi$ in probability as $n$ grows, and in turn it gives a lower bound for $n$ that is conveniently expressed in the same form as in Lemma \ref{lem:importbest}. 

\begin{theorem}
	\label{thrm:probest}
	If all $a_I \geq 0$, and
	\begin{gather*}
		\sum_{I \in \mathcal{S}} \frac{a_I}{w_I} \frac{1}{\sqrt{ \phi(M, I)}} < \infty,
	\end{gather*}
	then
	\begin{align}
		\label{eq:squareerr}
		\mathbb{E}\left[\left( \mathcal{E}^{\textnormal{GBS-P}}_n - \mu_{\sqrt{\phi}} \right)^2 \right] < \frac{1}{n} \left( Q_{\sqrt{\phi}}^{\textnormal{GBS-P}} - \mu_{\sqrt{\phi}}^2 \right)
	\end{align}
	where
	\begin{equation}
		\label{eq:squareerr1}
		Q_{\sqrt{\phi}}^{\textnormal{GBS-P}} = 
		 \sum_{I, J \in \mathcal{S}} a_I a_J \frac{1}{w_J} \frac{\sqrt{\phi(M, I)}}{\sqrt{\phi(M, J)}}.
	\end{equation}
\end{theorem}

\begin{proof}
	We introduce the following shorthand notations for convenience. Let $ \mathcal{E}_n = \mathcal{E}^{\textnormal{GBS-P}}_n $, $\mu = \mu_{\sqrt{\phi}}$. Then,
	\begin{align}
		\mathbb{E}\left[ \left( \mathcal{E}_n - \mu \right)^2 \right] = \mathbb{E}\left[\mathcal{E}_n^2 \right] +  \mu^2 - 2 \mu \mathbb{E}\left[\mathcal{E}_n \right].
		\label{eq:lemuse0}
	\end{align}
	We first show that
	\begin{align}
		\mathbb{E}\left[\mathcal{E}_n^2 \right] = \mathbb{E}\left[  \sum_{J, J' \in \mathcal{S}} \alpha_J {\alpha}_{J'}  \sqrt{{S^{(J)}_n}/{n}}  \sqrt{{S^{(J')}_n}/{n}} \right] \leq \sum_{J, J' \in \mathcal{S}} \alpha_J {\alpha}_{J'} \sqrt{p_J p_{J'}}. 
		\label{eq:lemuse1}
	\end{align}
	This is because 
	$$ \mathbb{E}\left[  \sum_{J, J' \in \mathcal{S}} \alpha_J {\alpha}_{J'}  \sqrt{{S^{(J)}_n}/{n}}  \sqrt{{S^{(J')}_n}/{n}} \right] =  \sum_{J, J' \in \mathcal{S}} \alpha_J {\alpha}_{J'} \mathbb{E}\left[  \sqrt{{S^{(J)}_n}/{n}}  \sqrt{{S^{(J')}_n}/{n}} \right].
	$$ 
	We can swap the order of the sum and the expectation because for a fixed $n$ the sum is taken over only finitely many terms. 
	Further,
	$$ \mathbb{E}\left[  \sqrt{{S^{(J)}_n}/{n}}  \sqrt{{S^{(J')}_n}/{n}} \right] \leq \sqrt{p_J p_{J'}}.$$
	If $J = J'$, the result clearly holds. If $J \neq J'$, then Jensen's inequality on concave functions yields
	\begin{equation*}
		\mathbb{E}\left[  \sqrt{{S^{(J)}_n}}  \sqrt{{S^{(J')}_n}} \right] \leq \mathbb{E}\left[  {{S^{(J)}_n}}  {{S^{(J')}_n}} \right]^{1/2}.
	\end{equation*}
	We further compute that
	\begin{align*}
		\mathbb{E}\left[  {{S^{(J)}_n}}  {{S^{(J')}_n}} \right] & = \mathbb{E}\left[ \sum_{l=1}^n 1_J(I_l) \sum_{l=1}^n 1_{J'}(I_l) \right] \\
		& = \mathbb{E}\left[ \sum_{l=1}^n \sum_{m=1}^n 1_J(I_l)  1_{J'}(I_m) \right] \\
		& = \mathbb{E}\left[ \sum_{l\neq m}^n  1_J(I_l)  1_{J'}(I_m) \right] \\
		& = \sum_{l \neq m} p_J p_{J'} \\
		& = p_J p_{J'} n (n-1),
	\end{align*}
	which gives the desired inequality. In the computation, we have used the fact that if $l = m$ and $J \neq J'$, then $1_J(I_l) 1_{J'}(I_m) = 0$ since $I_l$ cannot be both $J$ and $J'$.
	
	We now show that
	\begin{align}
		\mathbb{E}\left[ \mathcal{E}_n \right] =  \mathbb{E}\left[  \sum_{J \in \mathcal{S} }\alpha_J  \sqrt{{S^{(J)}_n}/{n}} \right]
		\geq  \sum_{J \in \mathcal{S}} \alpha_J \left(\sqrt{p_J} - \frac{1 - p_J}{2n \sqrt{p_J}} \right).
		\label{eq:lemuse2}
	\end{align}
	We first notice that
	$$  \mathbb{E}\left[  \sum_{J \in \mathcal{S}}\alpha_J  \sqrt{{S^{(J)}_n}/{n}} \right] =   \sum_{J \in \mathcal{S}}\alpha_J  \mathbb{E}\left[  \sqrt{{S^{(J)}_n}/{n}} \right] =  \frac{1}{\sqrt{n}}  \sum_{J \in \mathcal{S}}\alpha_J  \mathbb{E}\left[  \sqrt{{S^{(J)}_n}} \right]. $$
	Again, we can swap the order of the sum and the expectation because for a fixed $n$ the sum is taken over only finitely many terms.  Furthermore, let $Y = X/\mathbb{E}[X]$. Then by using 
	\begin{align*}
		\sqrt{x} \geq \frac{3}{2} x - \frac{1}{2} x^2
	\end{align*}
	we get
	\begin{align*}
		\frac{\mathbb{E}[\sqrt{X}]}{\sqrt{\mathbb{E}[X]}} = \mathbb{E}[\sqrt{Y}] & \geq \frac{3}{2} \mathbb{E}[Y] - \frac{1}{2} \mathbb{E}[Y^2] \\
		& = \frac{3}{2} - \frac{1}{2} \frac{\mathbb{E}[X^2]}{\mathbb{E}[X]^2} \\
		& = 1 - \frac{1}{2} \left( \frac{\text{Var}[X]}{\mathbb{E}[X]^2} \right).
	\end{align*}
	By plugging $X = S^{(J)}_n$ into the above expression, we obtain
	\begin{align}
		\mathbb{E}\left[  \sqrt{{S^{(J)}_n}} \right] \geq  \sqrt{\mathbb{E}\left[  {S^{(J)}_n} \right]} \left( 1 - \frac{\text{Var}({S^{(J)}_n})}{2 \mathbb{E}\left[  {S^{(J)}_n} \right]^2 }\right). \label{eq:use22}
	\end{align}
	Further, since ${S^{(J)}_n}$ is binomial and has variance $n p_J (1-p_J)$, the RHS of \eqref{eq:use22} can be computed as 
	\begin{align*}
	 \sqrt{np_J} \left( 1 - \frac{n p_J(1-p_J)}{2n^2 p_J^2} \right)
		= \sqrt{np_J} \left( 1 - \frac{1 - p_J}{2n {p_J}} \right).
	\end{align*}
	Therefore, 
	\begin{align*}
		\mathbb{E}\left[  \sqrt{{S^{(J)}_n}} \right] \geq \sqrt{np_J} - \frac{1 - p_J}{2 \sqrt{np_J}}
	\end{align*}
	and \eqref{eq:lemuse2} is established.
	
	Plugging \eqref{eq:lemuse1}, \eqref{eq:lemuse2} and $\mu = \displaystyle \sum_{I \in \mathcal{S}} \alpha_I \sqrt{p_I}$ into \eqref{eq:lemuse0}, we get
	\begin{align*}
		\mathbb{E} \left[ \left( \mathcal{E}_n - \mu \right)^2 \right] 
		& \leq 2 \sum_{J, J' \in \mathcal{S}} \alpha_J {\alpha}_{J'} \sqrt{p_J p_{J'}} - 2 \sum_{J \in \mathcal{S}} {\alpha}_J \sqrt{p_J} \sum_{J' \in \mathcal{S}} \alpha_{J'} \left(\sqrt{p_{J'}} - \frac{1 - p_{J'}}{2n \sqrt{p_{J'}}} \right) \\
		& = \frac{1}{n} \sum_{J, J' \in \mathcal{S} }  {\alpha}_J {\alpha}_{J'}  (1- p_{J'}) \frac{\sqrt{p_J}}{\sqrt{p_{J'}}} \\
		& = \frac{1}{n} \left( \sum_{J, J' \in \mathcal{S} } {\alpha}_J {\alpha}_{J'}  \frac{\sqrt{p_J}}{\sqrt{p_{J'}}} - \sum_{J, J' \in \mathcal{S} }  {\alpha}_J {\alpha}_{J'}  {\sqrt{p_J p_{J'}}} \right) \\
		& = \frac{1}{n} \sum_{J, J' \in \mathcal{S} } {\alpha}_J {\alpha}_{J'} \frac{\sqrt{p_J}}{\sqrt{p_{J'}}} - \frac{1}{n} \mu^2.
	\end{align*}
	From $p_J = w_J \phi(M, J)$ and $\alpha_J = a_J \sqrt{\frac{1}{w_J}}$, we obtain 
	\begin{align*}
		Q_{\sqrt{\phi}}^{\textnormal{GBS-P}} = \sum_{J, J' \in \mathcal{S} }  {\alpha}_J {\alpha}_{J'}  \frac{\sqrt{p_J}}{\sqrt{p_{J'}}} & = \sum_{J, J' \in \mathcal{S} } \frac{a_J a_{J'}}{\sqrt{w_J {w}_{J'}}} \frac{\sqrt{ w_J \phi(M, J)}}{\sqrt{ w_{J'} \phi(M, {J'})}} \\
		& = \sum_{J, J' \in \mathcal{S} }  \frac{a_J a_{J'} }{{w}_{J'}} \frac{\sqrt{ \phi(M, J)}}{\sqrt{ \phi(M, {J'})}} 
	\end{align*}
	and we thus get the expression given as in \eqref{eq:squareerr} and \eqref{eq:squareerr1}. 
\end{proof}

Since $Q_{\sqrt{\phi}}^{\textnormal{GBS-P}} < \infty$ by assumption, 
a lower bound for $n$ can be obtained via Markov's inequality.

\begin{corollary}
	\label{cor:n-gbsp}
	The probability estimator $\mathcal{E}^{\textnormal{GBS-P}}_n$ solves $\mathcal{G}_{\sqrt{\phi}}^\times(\epsilon, \delta)$ with
	\begin{equation}
		n \geq  \frac{Q_{\sqrt{\phi}}^{\textnormal{GBS-P}} - \mu_{\sqrt{\phi}}^2}{\delta \epsilon^2 \mu_{\sqrt{\phi}}^2} \equiv n_{\sqrt{\phi}}^{\textnormal{GBS-P}},
		\label{eq:nlbgbsp}
	\end{equation}
	provided that
	\begin{align}
		Q_{\sqrt{\phi}}^{\textnormal{GBS-P}}< \infty. \label{eq:QPfinite} \tag{A2}
	\end{align}
\end{corollary}

\begin{proof}
	Using Markov's inequality on $\left( \mathcal{E}^{\textnormal{GBS-P}}_n - \mu_{\sqrt{\phi}} \right)^2$, we get
	\begin{align*}
		P(\vert \mathcal{E}^{\textnormal{GBS-P}}_n - \mu_{\sqrt{\phi}} \vert > \epsilon \vert \mu_{\sqrt{\phi}} \vert) & = P(\left( \mathcal{E}^{\textnormal{GBS-P}}_n - \mu_{\sqrt{\phi}} \right)^2 > \epsilon^2  \mu_{\sqrt{\phi}}^2) \\
		& <  \frac{\mathbb{E}\left[\left( \mathcal{E}^{\textnormal{GBS-P}}_n - \mu_{\sqrt{\phi}} \right)^2 \right]}{\epsilon^2 \mu_{\sqrt{\phi}}^2}.
	\end{align*}
	Setting the right-hand side to be less than $\delta$ and solving for $n$, we get \eqref{eq:nlbgbsp}.
\end{proof}

	\section{Estimates for the guaranteed sample size}
\label{sec:costest}
For the rest of the paper, we return to the original Gaussian expectation problems $\mathcal{I}_{\Haf}$ and $\mathcal{I}_{\Haf^2}$. We discuss how to use GBS-I to solve for $\mathcal{I}^\times_{\Haf^2}(\epsilon, \delta)$ and how to use GBS-P to solve for $\mathcal{I}^\times_{\Haf}(\epsilon, \delta)$. From this point onward, we only concern ourselves with the GBS sampling distribution described as in \eqref{eq:gbsfirstappear} (or \eqref{eq:gbshaf1}), and the matrix $B$ is real, symmetric with eigenvalues bounded strictly between 0 and 1. 
Our main focus is to prove the bounds for the guaranteed sample size $n_{\Haf^2}^{\textnormal{GBS-I}}$ and $n_{\Haf}^{\textnormal{GBS-P}}$ as stated in Theorems \ref{main:pp1} and \ref{main:pp2}. We also discuss how to use the plain Monte Carlo (MC) estimator to solve these problems and provide estimates for their corresponding  guaranteed sample size.

We start off with two useful lemmas, which give upper and lower bounds for the hafnian of a matrix $M \in \mathbb{R}^{2k \times 2k}$ with $k \geq 1$. Let $m_{\textnormal{max}}$ denote the maximum of the absolute value of the entries of $M$ and let $m_{\textnormal{min}}$ denote the minimum entry of $M$. 

\begin{lemma}
	\label{lem:hafapprox0}
	Given notations as above, 
	\begin{align*}
		\Haf(M) \leq   m_{\textnormal{max}}^k \frac{(2k)!}{2^k k!}.
	\end{align*} 
	If $M_{ij} > 0$ for all $i$ and $j$, then
	\begin{equation*}
	 \Haf(M) \geq m_{\textnormal{min}}^k \frac{(2k)!}{2^k k!}.
	\end{equation*}
\end{lemma}

\begin{proof}
	For a $2k \times 2k$ matrix $M$ with completely identical entries $m$ everywhere,
	\begin{align}
		\Haf(M) = m^k (2k-1)!! = m^k \frac{(2k)!}{2^k k!}.
		\label{eq:usefulhaf}
	\end{align}
	Then, to bound $\Haf(M)$, we use the hafnian of the matrix with all entries equal to $m_{\textnormal{min}}$ as the lower bound, and similarly use the hafnian of the matrix with all entries equal to $m_{\textnormal{max}}$ as the upper bound. 
\end{proof}

The next result follows immediately from Lemma \ref{lem:hafapprox0} by using Stirling's formula on both sides. 
\begin{lemma}
	\label{lem:hafapprox}
	Given notations as above,
	\begin{align*}
		\lvert \Haf(M) \rvert  \leq \sqrt{2} \left( \frac{2 k m_{\textnormal{max}} }{e} \right)^k.
	\end{align*}
	If $M_{ij} > 0$ for all $i$ and $j$, then
	\begin{equation*}
		\Haf(M) \geq \sqrt{2} e^{\frac{1}{25}-\frac{1}{12}} \left( \frac{2k m_{\textnormal{min}}}{e} \right)^k.
	\end{equation*}
\end{lemma}

\begin{proof}
	Using Stirling's formula, one can bound $k!$ by
	\begin{align*}
		\sqrt{2\pi k } \left( \frac{k}{e} \right)^k e^{\frac{1}{12k+1}} \leq 
		k! \leq \sqrt{2\pi k } \left( \frac{k}{e} \right)^k e^{\frac{1}{12k}}.
	\end{align*}
	Therefore, we can bound $\lvert \Haf(M) \rvert$ by 
	\begin{align*}
		\lvert \Haf(M) \rvert \leq \sqrt{2} m_{\textnormal{max}}^k \left( \frac{2 k}{e} \right)^k e^{\frac{1}{24k}-\frac{1}{12k+1}} \leq \sqrt{2} \left( \frac{2 k m_{\textnormal{max}}}{e} \right)^k.
	\end{align*}
	If we further assume that all entries of $M$ are positive, then
	\begin{align*}
		\Haf(M) \geq \sqrt{2} m_{\textnormal{min}}^k  \left( \frac{2 k}{e} \right)^k e^{\frac{1}{24k+1}-\frac{1}{12k}} \geq \sqrt{2} \left( \frac{2k m_{\textnormal{min}}}{e} \right)^k  e^{\frac{1}{25}-\frac{1}{12}}.
	\end{align*}
\end{proof}

\subsection{Proof of Theorem \ref{main:pp1}}
\label{subsec:sqest}
We first look at using GBS-I to solve  $\mathcal{I}_{\Haf^2}^\times(\epsilon, \delta)$. Suppose the samples are drawn i.i.d with the probability density function given as in formula \eqref{eq:gbsfirstappear} (or \eqref{eq:gbshaf1}). It is straightforward from the two assumptions \eqref{eq:absconvsq} and \eqref{eq:QIfinite}, and 
from Lemma \ref{lem:importbest} that
$\mathcal{E}^{\textnormal{GBS-I}}_n$ solves $\mathcal{I}_{\Haf^2}^\times(\epsilon, \delta)$ with
\begin{equation} 
	n \geq n_{\Haf^2}^{\textnormal{GBS-I}} = \frac{Q_{\Haf^2}^{\textnormal{GBS-I}} - \mu_{\Haf^2}^2}{\delta \epsilon^2 \mu_{\Haf^2}^2}, \label{eq:formulamain11}
\end{equation}
where
\begin{gather*}
	Q_{\Haf^2}^{\textnormal{GBS-I}} = \frac{1}{d} \sum_{k = 0}^K \sumck a_I^2 I! \Haf(B_I)^2.
\end{gather*}
Our primary goal in this section is to prove the estimates for $n_{\Haf^2}^{\textnormal{GBS-I}}$ stated as in Theorem \ref{main:pp1}.
To motivate the conditions imposed there, we will discuss the two assumptions \eqref{eq:absconvsq} and \eqref{eq:QIfinite} in turn and illustrate that those conditions are near optimal. 

Recall that \eqref{eq:absconvsq} is given by
\begin{align*}
 \int_{\mathbb{R}^{2N}} \sum_{k = 0}^K  \,   \sum_{\vert I \rvert = k}   \vert a_{\vecI} p^I q^I \vert \, h(p, q) \, \dd p \dd q < \infty
\end{align*}
and it implies $\vert \mu_{\Haf^2} \vert < \infty$ so as to make $\mathcal{I}_{\Haf^2}$ well-defined. When $K < \infty$, the assumption is of course vacuous, and we focus on the case when $K = \infty$. 
We begin by looking at a sufficient condition for $\vert \mu_{\Haf^2} \vert < \infty$.  
First, from the triangle inequality,
\begin{align*}
	\vert \mu_{\Haf^2} \vert  \leq \sum_{k = 0}^\infty \sumck \vert  a_I \vert \Haf(B_I)^2.
\end{align*}
Then, it follows from Lemma \ref{lem:hafapprox} that
\begin{align}
	\vert \mu_{\Haf^2} \vert 
	\leq \vert a_0 \vert + 2 \sum_{k = 1}^\infty  \sumck \vert  a_I  \vert \left( \frac{2 k \bmax }{e} \right)^{2k} \label{eq:est-mu-use1},
\end{align}
where $a_0$ denotes the coefficient $a_{(0,0, \dots,0)}$ corresponding to the all-zero tuple. 
The right-hand side of \eqref{eq:est-mu-use1} is a power series. By the root test, the series converges absolutely if
\begin{equation}
	\begin{aligned}
		C & = \limsup_{k \rightarrow \infty} \left( \sumck \vert a_I \vert \right)^{\frac{1}{k}} \left(\frac{2 k \bmax }{e}  \right)^{2} \\
		& = \limsup_{k \rightarrow \infty} \left(  \sumck \vert a_I \vert k^{2k} \right)^{\frac{1}{k}} \left(  \frac{2 \bmax }{e} \right)^2 \\
		& < 1
	\end{aligned}
	\label{eq:est-use2}
\end{equation}
If $C > 1$, then \eqref{eq:est-mu-use1} diverges, and if $C=1$, then the convergence of the series is undetermined. 
 
\begin{proposition}
	\label{prop:hafsq111}
	Suppose for all $k \geq 1$, there exists $\alpha_k$ such that
	\begin{align}
		\label{eq:est-use3}
		 \sumck \vert a_I \vert \leq \alpha_k \left(\frac{e}{2 k  } \right)^{2k} \gamma^{k}
	\end{align}
	where
	\begin{align}
		\label{eq:est-use4}
		\limsup_{k \rightarrow \infty} \alpha_k^{\frac{1}{k}} \leq 1, 
	\end{align}
	and 
	\begin{align}
		\label{eq:est-use51}
		\gamma < \frac{1}{\bmax^2}.
	\end{align}
	Then, $$\vert \mu_{\Haf^2} \vert < \infty.$$ 
\end{proposition}

\begin{proof}
	We begin by noting that for all $k \geq 1$, \eqref{eq:est-use3} is equivalent to 
	\begin{align*}
		\left(  \sumck \vert a_I \vert k^{2k} \right)^{\frac{1}{k}} \leq \alpha_k^{\frac{1}{k}} \left(\frac{e}{2 } \right)^{2} \gamma,
	\end{align*}
	and taking the limsup on both sides, one obtains from \eqref{eq:est-use4} that
	\begin{align*}
		\limsup_{k \rightarrow \infty} \left(  \sumck \vert a_I \vert k^{2k} \right)^{\frac{1}{k}}  & \leq \limsup_{k \rightarrow \infty}  \alpha_k^{\frac{1}{k}} \left( \frac{e }{2 } \right)^{2}  \gamma \leq \left( \frac{e}{2 } \right)^{2} \gamma.
	\end{align*}
	Therefore,
	\begin{align*}
		C = \limsup_{k \rightarrow \infty} \left(  \sumck \vert a_I \vert k^{2k} \right)^{\frac{1}{k}}   \left(  \frac{2 \bmax }{e} \right)^2 \leq \left( \frac{e}{2 } \right)^{2}  \gamma  \left(  \frac{2 \bmax }{e} \right)^2 < 1.
	\end{align*}
	Hence, by the root test analysis, $\vert \mu_{\Haf^2} \vert < \infty.$
\end{proof}

Note that the conditions presented as in Proposition \ref{prop:hafsq111} are close to optimal, since they directly come from the root test. For the convenience of many upcoming computations, we rewrite \eqref{eq:est-use3} using Stirling's approximation. 
\begin{corollary}
	\label{prop:hafsq114}
	Suppose for all $k \geq 1$, there exists $\alpha_k$ such that
	\begin{align}
		\label{eq:est-mu-use2}
		  \sumck \vert a_I \vert \leq \alpha_k \frac{\gamma^k}{(2k)! },
	\end{align}
	where  $\alpha_k$ and $\gamma$ satisfy \eqref{eq:est-use4} and \eqref{eq:est-use51}. Then, $\vert \mu_{\Haf^2} \vert < \infty$.
\end{corollary}

\begin{proof}
	Using the Stirling's approximation on \eqref{eq:est-mu-use2}, we obtain
	\begin{align*}
		 \sumck \vert a_I \vert \leq \frac{\alpha_k}{\sqrt{4 \pi k }}  \left(\frac{e}{2k } \right)^{2k} \gamma^{k}.
	\end{align*}
	It is clear that $\alpha_k /{\sqrt{4 \pi k }}$ also satisfy $ \eqref{eq:est-use4}$. Then from Proposition \ref{prop:hafsq111}, it follows that $\vert \mu_{\Haf^2} \vert < \infty$. 
\end{proof}

In the next lemma, we prove that the conditions in Corollary \ref{prop:hafsq114} imply \eqref{eq:absconvsq}.

\begin{lemma}
	\label{lem:swap1}
	If the conditions in Corollary \ref{prop:hafsq114} are given, then \eqref{eq:absconvsq} holds.
\end{lemma}

\begin{proof}
	Let $p = (p_1, \dots, p_N)$ and $q = (q_1, \dots, q_N)$. We first define $S$ to be the hypercube in $\mathbb{R}^{2N}$, where each variable is bounded between $-1$ and $1$
	\begin{align*}
		S &= \{ (p, q) \in \mathbb{R}^{2N} \mid \vert p_n \vert, \vert q_n \vert \leq 1 \text{ for all } n = 1, \dots, N\}.
	\end{align*}
	Let $S' = \mathbb{R}^{2N} \backslash S $ be its complement. 
	Then, we divide the integral of \eqref{eq:absconvsq} into two parts. Let 
	\begin{align*}
		V_S
		& =	\int_{S} \sum_{k = 0}^\infty  \,   \sum_{\vert I \rvert = 2k}  \vert a_{\vecI} \vert \vert p^I q^I \vert \, h(p, q) \, \dd p \dd q.
	\end{align*}
	and
	\begin{align*}
		V_{S'}
		& = \int_{S'} \sum_{k = 0}^\infty   \,  \sum_{\vert I \rvert = 2k}  \vert a_{\vecI}\vert \vert p^I q^I \vert \, h(p, q) \, \dd p \dd q.
	\end{align*}
	Then,
	\begin{align*}
		\int_{\mathbb{R}^{2N}} & \sum_{k = 0}^\infty  \,   \sum_{\vert I \rvert = 2k}   \vert a_{\vecI} \vert \vert p^I q^I \vert \, h(p, q) \, \dd p \dd q   = V_S + V_{S'}.
	\end{align*}
	
	Since $\vert p^I q^I \vert \leq 1$ for all $(p, q) \in S$ and all $I$,
	one obtains
	\begin{align}
		V_S & \leq \int_{S} \sum_{k = 0}^\infty   \,  \sum_{\vert I \rvert = 2k} \vert a_{\vecI} \vert \, h(p, q) \, \dd p \dd q \nonumber \\
		& = \left( \sum_{k = 0}^\infty   \sum_{\vert I \rvert = 2k} \vert a_{\vecI} \vert  \right) \int_{S} h(p, q) \, \dd p \dd q \nonumber \\
		& \leq \left( \sum_{k = 0}^\infty \sum_{\vert I \rvert = 2k} \vert a_{\vecI} \vert  \right)  \cdot 1 \nonumber \\
		& \leq \vert a_0 \vert + \sum_{k = 1}^\infty   \alpha_k \frac{\gamma^k}{(2k)! } \label{eq:inproof2} \\
		& < \infty \nonumber
	\end{align}
	The finiteness of \eqref{eq:inproof2} follows from a simple root test analysis.
	Note that
	\begin{align*}
		C & = \limsup_{k \rightarrow \infty} \left\vert \alpha_k \frac{\gamma^k}{(2k)!}  \right\vert^{\frac{1}{k}} \\
		& \leq  \left( \limsup_{k \rightarrow \infty} \vert \alpha_k \vert^{\frac{1}{k}} \right) \left( \limsup_{k \rightarrow \infty} \left\vert \frac{\gamma^k}{(2k)!}  \right\vert^{\frac{1}{k}} \right) \\
		& =  \left( \limsup_{k \rightarrow \infty}  \alpha_k^{\frac{1}{k}} \right) \left( \limsup_{k \rightarrow \infty} \left\vert \frac{\gamma^k}{(2k)!}  \right\vert^{\frac{1}{k}} \right) \\
		& \hspace{0.5in} (\text{ since } \alpha_k \text{ is positive for all } k \geq 1) \\
		& \leq  \left( \limsup_{k \rightarrow \infty} \left\vert \frac{\gamma^k}{(2k)!}  \right\vert^{\frac{1}{k}} \right) =0,
	\end{align*}
	which implies that \eqref{eq:inproof2} is finite. 
	
	For $V_{S'}$, we apply the Fubini-Tolleni theorem for non-negative functions to swap the order of the integral and the sum, and we get
	\begin{align*}
		V_{S'} =\sum_{k = 0}^\infty   \,  \sum_{\vert I \rvert = 2k} \vert a_{\vecI}\vert  \int_{S'}  \vert p^I q^I \vert \, h(p, q) \, \dd p \dd q.
	\end{align*}
	We would like to bound $\vert p^I q^I \vert$ with $p^Lq^L$ for some $L$ and invoke Wick's theorem. Towards that, we use Lemma \ref{lem:useful135} to show that for a given $I = (i_1, i_2, \dots, i_N)$ with $\vert I \vert = 2k$ and $k \geq 1$, there exists $L_{(I)} = (l_1, \dots, l_N) \in \mathbb{N}^N$ such that 
	$$\vert L_{(I)} \vert = 2k + 2m$$ 
	where $m = \lceil \frac{N}{2} \rceil$. Furthermore, for any $n \in \{1, 2, \dots, N\}$, 
	$$l_n \geq i_n  ~~~~\text{ and }~~~~ l_n \text{ is even.}$$
	Thus, $\vert p^I q^I \vert \leq p^{L_{(I)}} q^{L_{(I)}} $ for all $(p, q) \in S'$. Therefore, 
	\begin{align*}
		\int_{S'}  \vert p^I q^I \vert \, h(p, q) \, \dd p \dd q & \leq \int_{S'}    p^{L_{(I)}} q^{L_{(I)}} \, h(p, q) \, \dd p \dd q \\
		& \leq \int_{\mathbb{R}^{2N}}   p^{L_{(I)}} q^{L_{(I)}} \, h(p, q) \, \dd p \dd q \\
		& = \Haf(B_{L_{(I)}})^2.
	\end{align*}
	The last line follows from Wick's theorem. 
	We then compute
	\begin{align*}
		V_{S'} & \leq \vert a_0 \vert + \sum_{k = 1}^\infty   \,   \sum_{\vert I \rvert = 2k} \vert  a_{\vecI}\vert \Haf(B_{L_{(I)}}) ^2 \\
		& =\vert a_0 \vert  + 2 \sum_{k = 1}^\infty   \,  \sum_{\vert I \rvert = 2k} \vert  a_{\vecI} \vert  \, \bmax^{2k + 2m} \left( \frac{2k + 2m }{e} \right)^{2k + 2m} \\
		&  \hspace{1.5in} \text{(Lemma \ref{lem:hafapprox})} \\
		& \leq \vert a_0 \vert + 2 \sum_{k = 1}^\infty   \, \frac{\alpha_k}{(2k)!} \gamma^k \, \bmax^{2k + 2m} \left( \frac{2k + 2m }{e} \right)^{2k + 2m} \\
		&  \hspace{1.5in} \text{(condition \eqref{eq:est-mu-use2})} \\
		& \leq \vert a_0 \vert +2 \sum_{k = 1}^\infty   \, \frac{\alpha_k}{(2k)!} \gamma^k \, \bmax^{2k + 2m} (2k + 2m)! \frac{1}{\sqrt{2\pi(2k + 2m)}} \\
		&  \hspace{1.5in} \text{(Stirling's approxiamtion on $(2k+2m)!$)} \\
		& =  \vert a_0 \vert + \frac{\sqrt{2}}{\sqrt{\pi}} \bmax^{ 2m} \sum_{k = 1}^\infty   \, \alpha_k \gamma^k \, \bmax^{2k} \frac{1}{\sqrt{2k + 2m}}\frac{(2k + 2m)!}{(2k)!}  \\
		& =   \vert a_0 \vert + \frac{\sqrt{2}}{\sqrt{\pi}} \bmax^{ 2m} \sum_{k = 1}^\infty   \, \alpha_k \gamma^k \, \bmax^{2k} \sqrt{2k + 2m} (2k+2m-1) \dots (2k+1) \\
	\end{align*}
	Note that if $k \geq m$, then
	\begin{align*}
		(2k+1) < (2k+2) < \dots < (2k+2m) \leq 4k.
	\end{align*}
	With $m = \lceil \frac{N}{2} \rceil \geq 1$, we obtain
	\begin{align*}
		\sqrt{2k + 2m} (2k+2m-1) \dots (2k+1) \leq (2k + 2m)^{2m- \frac{1}{2}} \leq (4k)^{2m-\frac{1}{2}}.
	\end{align*}
	Thus,
	\begin{align*}
		& \sum_{k = m+1}^\infty   \, \alpha_k \gamma^k \, \bmax^{2k} \sqrt{2k + 2m} (2k+2m-1) \dots (2k+1) \\
		&~~~~~~~~\leq  \sum_{k = m+1}^\infty  \, \alpha_k \gamma^k \, \bmax^{2k} (4k)^{2m - \frac{1}{2}} \\
		& ~~~~~~~~= 4^{2m - \frac{1}{2}}  \sum_{k = m+1}^\infty  \, \alpha_k  \gamma^k \, \bmax^{2k} k^{2m - \frac{1}{2}} \\
		& ~~~~~~~~ < \infty,
	\end{align*}
	which follows from a similar root test analysis as before.
	Therefore, $V_{S'}$ is also finite.
\end{proof}

We now move on to \eqref{eq:QIfinite}, which requires $Q_{\Haf^2}^{\textnormal{GBS-I}}$ to be finite. It follows from Lemma \ref{lem:hafapprox} that
\begin{align}
Q_{\Haf^2}^{\textnormal{GBS-I}} \leq \frac{1}{d} a_0^2 + \frac{1}{d} \sum_{k = 1}^K \left(\sumck a_I^2 I!  \right) \left( \frac{2 k \bmax }{e} \right)^{2k}.
\label{eq:est-use1}
\end{align}
If $K < \infty$, the assumption holds trivially, and therefore we are only interested in the case of $K = \infty$. Using the same kind of root test analysis, we obtain the following conditions for \eqref{eq:QIfinite} to hold.

\begin{lemma}
	\label{lem:swap12}
	Suppose for all $k \geq 0$, there exists $\alpha_k$ such that
	\begin{align}
		\label{eq:newest-mu-use2}
	 \sumck a^2_I I! \leq \alpha_k \frac{\gamma^k}{(2k)! },
	\end{align}
	where  $\alpha_k$ and $\gamma$ satisfy \eqref{eq:est-use4} and \eqref{eq:est-use51}.
	Then \eqref{eq:QIfinite} holds.
\end{lemma}

\begin{proof}
	The proof is omitted being similar to Proposition \ref{prop:hafsq111} and  Corollary \ref{prop:hafsq114}.
\end{proof}

So far, we have illustrated the near optimal conditions on the $a_I$'s so that \eqref{eq:absconvsq} and \eqref{eq:QIfinite} hold. To give an explicit expression for $n_{\Haf^2}^{\textnormal{GBS-I}}$ as in Theorem \ref{main:pp1}, we choose
\begin{equation}
	\label{eq:est-use5}
	\begin{gathered}
		\alpha_k = c k^q
	\end{gathered}
\end{equation}
for some $q$ and $c$. Then, we can upper bound $Q_{\Haf^2}^{\textnormal{GBS-I}}$ by a polylog. 
The proof is now straightforward. 

\begin{proof}[Proof of Theorem \ref{main:pp1}]
	With the given conditions, it is clear that \eqref{eq:absconvsq}-\eqref{eq:QIfinite} hold. 
	Further, for $K$ finite or infinite, we compute that
	\begin{align*}
		Q_{\Haf^2}^{\textnormal{GBS-I}}
		&\leq \frac{a_0^2}{d} +  \frac{c_2}{d} \sum_{k = 1}^K  k^{q_2} {\gamma_2}^{k} \bmax^{2k} \frac{(2k)!}{2^{2k} (k!)^2} ~~~~\text{(Lemma \ref{lem:hafapprox0})} \\
		& \leq \frac{a_0^2}{d} +  \frac{c_2}{d} \frac{1}{\sqrt{\pi}}  \sum_{k = 1}^\infty k^{q_2 - \frac{1}{2}} {\gamma_2}^{k}  \bmax^{2k} \hspace{0.1in} ~~~~\text{(Lemma \ref{lem:usefulstirling1})} \\
		&= \frac{1}{d} \left( a_0^2 + \frac{c_2}{\sqrt{\pi}}  \multilog_{\frac{1}{2}-{q_2}, K}({\gamma_2} \bmax^2) \right),
	\end{align*}
	and when $K = \infty$,  the polylog is finite provided that $\gamma_2 < \frac{1}{\bmax^2}$. 
	Therefore,
	\begin{align*}
		n_{\Haf^2}^{\textnormal{GBS-I}} & = \frac{Q_{\Haf^2}^{\textnormal{GBS-I}} - \mu_{\Haf^2}^2}{\delta \epsilon^2 \mu_{\Haf^2}^2}
		 \leq \frac{1}{\delta \epsilon^2 \mu_0^2} \left( \frac{1}{d} a_0^2 + \frac{1}{d} \frac{c_2}{\sqrt{\pi}}  \multilog_{\frac{1}{2}-{q_2}, K}({\gamma_2} \bmax^2)- \mu_0^2 \right) 
	\end{align*}
	This completes the proof of the theorem. 
\end{proof}

If an a priori lower bound for $\vert \mu_{\Haf^2} \vert$ is not available, we show in the next proposition that by imposing conditions on the $a_I$'s, we can approximate $\vert \mu_{\Haf^2} \vert$. For this, we need that all entries of $B$ are strictly positive. So $\bmin >0$. We divide the coefficients $a_I$'s into the positive part and the negative part
\begin{align*}
	\mathcal{I}_k^{\pm} = \left\{ I \mid \vert I \vert = 2k, \pm a_I \geq 0 \right\}.
\end{align*}

\begin{proposition}
	\label{eq:est_use52}
	Assume $\bmin >0$.
	Suppose there exists $c, q$ and $\gamma$ such that for all $1 \leq k \leq K$,
	\begin{align}
		\label{eq:newest-mu-use5}
		\sum_{I \in \mathcal{I}_k^{+}} a_I + \left( \frac{\bmax}{\bmin}\right)^{2k} \sum_{I \in \mathcal{I}_k^{-}} a_I   \geq c \frac{k^q}{(2k)! } \gamma^k.
	\end{align}
	If $K = \infty$, we further require
	$\gamma < \frac{1}{\bmin^2}$. 
	Then, 
	\begin{align}
		\mu_{\Haf^2}
		\geq \left( a_0 + \frac{c}{\sqrt{\pi}} e^{\frac{1}{25}- \frac{1}{6}} \multilog_{\frac{1}{2}-q, K}\left( \gamma \bmin^2 \right) \right).
	\end{align}
\end{proposition}

\begin{proof}
	We compute that
\begin{align*}
	\mu_{\Haf^2}
	& = \sum_{k = 0}^K \sumck a_I \Haf(B_I)^2 \\
	& = \sum_{k = 0}^K \left( \sum_{I \in \mathcal{I}_k^{+}} a_I \Haf(B_I)^2 + \sum_{I \in \mathcal{I}_k^{-}} a_I \Haf(B_I)^2  \right)\\
	& \geq  a_0 +  \sum_{k = 1}^K  \left( \sum_{I \in \mathcal{I}_k^{+}}  a_I \bmin^{2k} \frac{(2k)!^2}{2^{2k} (k!)^2} +  \sum_{I \in \mathcal{I}_k^{-}}  a_I \bmax^{2k} \frac{(2k)!^2}{2^{2k} (k!)^2} \right) ~~~~~~ \text{(use Lemma \ref{lem:hafapprox0})}\\
	& =  a_0 +  \sum_{k = 1}^K  \left( \sum_{I \in \mathcal{I}_k^{+}}  a_I  +  \left( \frac{\bmax}{\bmin}\right)^{2k} \sum_{I \in \mathcal{I}_k^{-}}  a_I  \right) \bmin^{2k} \frac{(2k)!^2}{2^{2k} (k!)^2} \\
	& \geq a_0 + c \sum_{k = 1}^K k^q \gamma^k \bmin^{2k} \frac{(2k)!}{2^{2k} (k!)^2} ~~~~~~ \text{(use \eqref{eq:newest-mu-use5})} \\
	& \geq a_0 + \frac{c}{\sqrt{\pi}} e^{\frac{1}{25}- \frac{1}{6}} \sum_{k = 1}^K k^{q -\frac{1}{2}} \gamma^k \bmin^{2k} \hspace{0.1in} \text{(use Lemma \ref{lem:usefulstirling1})} \\
	& \geq a_0 + \frac{c}{\sqrt{\pi}} e^{\frac{1}{25}- \frac{1}{6}} \multilog_{\frac{1}{2}-q, K}\left( \gamma \bmin^2 \right),
\end{align*}
which is finite when $K = \infty$ provided that $\gamma < \frac{1}{\bmin^2}$.
\end{proof}

From this bound, we obtain the following corollary of Theorem \ref{main:pp1}. 
\begin{corollary}
	\label{cor:mainpp1}
	Assume $\bmin >0$ and \eqref{eq:absconvsq}. 
	Suppose there exists $\gamma_\alpha, \gamma_\beta, c_\alpha, c_\beta$ and $q_\alpha, q_\beta$ such that for all $ 1 \leq k \leq K$
	\begin{gather*}
		\sum_{I \in \mathcal{I}_k^{+}} a_I + \left( \frac{\bmax}{\bmin}\right)^{2k} \sum_{I \in \mathcal{I}_k^{-}} a_I   \geq c_\alpha \frac{ k^{q_\alpha} \gamma_\alpha^{k}}{(2k)!}, 
	\end{gather*}
	and $$ \sum_{\vert I \vert = 2k} a^2_I I! \leq c_\beta \frac{k^{q_\beta} \gamma_\beta^{k}}{(2k)!}. $$
	If  $K = \infty$, we also require $\gamma_\alpha < \frac{1}{\bmin^2}$ and $\gamma_\beta < \frac{1}{\bmax^2}$.
	Then,
	\begin{align*}
		n_{\Haf^2}^{\textnormal{GBS-I}} \leq \left( \frac{1}{d}\frac{ a_0^2 + \frac{c_\beta}{\sqrt{\pi}}  \multilog_{\frac{1}{2}-{q_\beta}, K}({\gamma_\beta} \bmax^2)
		}{ \left( a_0 + \frac{c_\alpha}{\sqrt{\pi}} e^{\frac{1}{25}- \frac{1}{6}} \multilog_{\frac{1}{2}-q_\alpha, K}\left( \gamma_\alpha \bmin^2 \right) \right)^2 } - 1 \right) \frac{1}{\delta \epsilon^2}.
	\end{align*}
\end{corollary}

\begin{proof}
Proposition \ref{eq:est_use52} implies
	\begin{align*}
		\mu_{\Haf^2}
		\geq \left( a_0 + \frac{c_\alpha}{\sqrt{\pi}} e^{\frac{1}{25}- \frac{1}{6}} \multilog_{\frac{1}{2}-q_\alpha, K}\left( \gamma_\alpha \bmin^2  \right) \right).
	\end{align*}
	The rest of the proof is straightforward by substituting the lower bound above for $\mu_0$ in Theorem \ref{main:pp1}.
\end{proof}

	\subsection{Proof of Theorem \ref{main:pp2}}
We now turn to the case of using GBS-P to solve $\mathcal{I}^\times_{\Haf}(\epsilon, \delta)$. As before, we use i.i.d samples that are drawn from \eqref{eq:gbsfirstappear} (or \eqref{eq:gbshaf1}), and under the assumptions \eqref{eq:absconv} and \eqref{eq:QPfinite}, the problem $\mathcal{I}^\times_{\Haf}(\epsilon, \delta)$ can be solved by GBS-P, if one additionally assumes for all $I$, 
$$\Haf(B_I) = \Haf(B, I) = \vert \Haf(B, I) \vert, $$
which is certainly implied by $\bmin \geq 0$. 
If all $a_I \geq 0$, then Theorem \ref{thrm:probest} and Corollary \ref{cor:n-gbsp} imply that $\mathcal{E}_n^\textnormal{GBS-P}$ solves $\mathcal{I}^\times_{\Haf}$  with
\begin{equation*} 
	n \geq n_{\Haf}^{\textnormal{GBS-P}} =   \frac{Q_{\Haf}^{\textnormal{GBS-P}} - \mu_{\Haf}^2}{\delta \epsilon^2 \mu_{\Haf}^2}
\end{equation*}
where
\begin{align}
	Q_{\Haf}^{\textnormal{GBS-P}} & = \frac{1}{d} \sumdoublefinite a_Ia_J J! \frac{\Haf(B_I)}{\Haf(B_J)} \nonumber \\
	& = \frac{1}{d} \left( \sum_{k_1 = 0}^K \sum_{\vert I \vert = 2k_1} a_I \Haf(B_I) \right) \left( \sum_{k_2 = 0}^K \sum_{\vert J \vert = 2k_2} a_J J! \frac{1}{\Haf(B_J)} \right) \nonumber \\
	& = \frac{\mu_{\Haf}}{d}  \left( \sum_{k = 0}^K \sum_{\vert J \vert = 2k} a_J J! \frac{1}{\Haf(B_J)} \right),
	\label{eq:qform_gbsP}
\end{align}
\noindent
As in the previous section, we will discuss the rationale behind the conditions imposed in Theorem \ref{main:pp2}. Since many of the proofs use similar techniques as before, we will only highlight the key differences to avoid repetition. 

Recall \eqref{eq:absconv} requires
\begin{align*}
	\int_{\mathbb{R}^N} \sum_{k = 0}^K  \,  \sum_{\vert I \rvert = k}  \vert a_I  x^I \vert \, h(x) \, \dd x < \infty,
\end{align*} 
and it implies $\vert \mu_{\Haf} \vert < \infty$. 
When $K < \infty$, the assumption holds trivially, and we focus on the case when $K = \infty$. 
As before, we first give a sufficient condition for $\vert \mu_{\Haf} \vert < \infty$.  

\begin{proposition}
	\label{prop:pp-mu-1}
	Suppose for all $k \geq 1$ there exists $\alpha_k$ such that
	\begin{align}
		\label{eq:est-pp-mu-2}
		 \sumck \vert a_I \vert \leq \alpha_k \frac{\gamma^{k}}{k!}
	\end{align}
	where 
	\begin{align}
		\label{eq:est-pp-mu-3}
		\limsup_{k \rightarrow \infty} \alpha_k^{\frac{1}{k}} \leq 1
	\end{align}
	and
	\begin{align}
		\label{eq:est-pp-mu-4}
		\gamma < \frac{1}{2 \bmax}.
	\end{align}
	Then, $$\vert \mu_{\Haf} \vert < \infty.$$
\end{proposition}

\begin{proof}
	From the triangle inequality and Lemma \ref{lem:hafapprox}, we get
	\begin{align}
		\vert \mu_{\Haf} \vert 
		&  \leq \frac{1}{d} \sum_{k = 0}^\infty \sumck \vert  a_I \vert \vert \Haf(B_I) \vert \nonumber \\
		&\leq \vert a_0 \vert + \sqrt{2} \sum_{k = 1}^\infty \sumck  \vert a_I  \vert \left( \frac{2 k \bmax }{e} \right)^{k}. 	\label{eq:est-pp-mu-1}
	\end{align}
	Repeating the same kind of root test analysis as in Proposition \ref{prop:hafsq111} and Corollary \ref{prop:hafsq114}, we obtain the finiteness of $\vert \mu_{\Haf} \vert$.
\end{proof}

We now show that the conditions described as in Proposition \ref{prop:pp-mu-1} imply \eqref{eq:absconv}.

\begin{lemma}
	\label{lem:swap2}
	If the conditions in Proposition \ref{prop:pp-mu-1} are given, then \eqref{eq:absconv} holds.
\end{lemma}

\begin{proof}
The proof is completely analogous to Lemma \ref{lem:swap1}. We first divide the integral of \eqref{eq:absconv} into two parts
\begin{align*}
	\int_{\mathbb{R}^{N}} \sum_{k = 0}^\infty  \,  \sum_{\vert I \rvert = 2k}  \vert a_{\vecI} \vert \vert x^I \vert \, h(x) \, \dd x = V_S + V_{S'}
\end{align*}
where the $V_S$ denotes the integral over the set $S$, where $S$ consists of those $x$ that lie in the hypercube, that is $\vert x_n \vert \leq 1$ for all $n  = 1, \dots, N$ and $V_{S'}$ denotes the integral over the complement set $S'$. Repeating the same argument as in  Lemma \ref{lem:swap1} while replacing $\Haf^2$ with $\Haf$ everywhere, we prove that both $V_S$ and $V_{S'}$ are finite. 
\end{proof}

We now move on to \eqref{eq:QPfinite}, where we require $Q_{\Haf}^{\textnormal{GBS-P}}$ to be finite. As before, when $K < \infty$, the assumption is vacuous and we therefore focus on the case of $K = \infty$.
In the next lemma, we provide conditions on $a_I$'s that imply \eqref{eq:QPfinite}.

\begin{lemma}
	\label{lem:swap22}
	Assume $\bmin >0$. Suppose for all $k \geq 1$, 
	\begin{align}
		\label{eq:newnewcondcond}
		\sum_{\vert I \vert = 2k} \vert a_I \vert I! \leq \alpha_k {\gamma^k} k!
	\end{align}
	where $\alpha_k$ satisfy \eqref{eq:est-pp-mu-3} and $\gamma < 2\bmin$.
	Then, \eqref{eq:QPfinite} holds.
\end{lemma}

\begin{proof}
	It follows from Lemma \ref{lem:hafapprox} that
	\begin{align*}
		Q_{\Haf}^{\textnormal{GBS-P}} \leq  \frac{\mu_{\Haf}}{d} \left(a_0  + \frac{1}{\sqrt{2}} e^{\frac{1}{12}-\frac{1}{25}}\sum_{k = 1}^\infty \sumck \vert a_I\vert I! \left( \frac{e }{2k\bmin} \right)^{k} \right).
	\end{align*}
	Using the same kind of root test analysis as in Proposition \ref{prop:hafsq111} and Corollary \ref{prop:hafsq114}, it is straightforward that $Q_{\Haf}^{\textnormal{GBS-P}} < \infty$ assuming \eqref{eq:newnewcondcond}.
\end{proof}

In Theorem \ref{main:pp2}, $\alpha_k$ is chosen to be of the form $ck^q$, and the proof is straightforward.

\begin{proof}[Proof of Theorem \ref{main:pp2}]
With the given conditions, it is clear that \eqref{eq:absconv} and \eqref{eq:QPfinite} hold. Furthermore, for $K$ finite or infinite, it follows from Lemma \ref{lem:hafapprox0} and all $a_I \geq 0$ that
\begin{align*}
	Q_{\Haf}^{\textnormal{GBS-P}}
	& \leq \frac{\mu_{\Haf}}{d} \left(a_0  + \sum_{k = 1}^\infty \sumck \vert a_I \vert I! \bmin^{-k} \frac{2^{k} k!}{(2k)!} \right) \\
	& \leq \frac{\mu_{\Haf}}{d} \left(a_0  + c_2 \sqrt{\pi} e^{\frac{1}{6}-\frac{1}{25}} \sum_{k = 1}^K \sumck k^{q_2 + \frac{1}{2}} \gamma_2^k (2\bmin)^{-k} \right)  \hspace{0.1in} \text{(use Lemma \ref{lem:usefulstirling1})} \\
	& \leq \frac{\mu_{\Haf}}{d} \left( a_0 + c_2 \sqrt{\pi} e^{\frac{1}{6}-\frac{1}{25}} \multilog_{-\frac{1}{2} -q_2, K} \left(\frac{\gamma_2 }{2\bmin} \right) \right),
\end{align*}
and when $K = \infty$, the polylog is finite provided that $\gamma_2 < 2\bmin$.
Finally, 
	\begin{align*}
		n_{\Haf}^{\textnormal{GBS-P}} 
		& =  \frac{Q_{\Haf}^{\textnormal{GBS-P}} - \mu_{\Haf}^2}{\delta \epsilon^2 \mu_{\Haf}^2} \\
		&\leq \frac{1}{d \delta \epsilon^2} \frac{1}{\mu_0} \left( a_0 +   c_2 \sqrt{\pi} e^{\frac{1}{6}-\frac{1}{25}} \multilog_{-\frac{1}{2} -q_2, K} \left(\frac{\gamma_2 }{2\bmin} \right)  \right) - \frac{1}{\delta \epsilon^2}.
	\end{align*}
	This completes the proof of the theorem. 
\end{proof}

As before, if a lower bound for $\mu_{\Haf}$ is not available, we can impose conditions on $a_I$'s instead to obtain an estimate for $\mu_{\Haf}$.

\begin{proposition}
\label{prop:est-mu-use5}
Assume $\bmin > 0$. Suppose there exists $c, q$ and $\gamma$ such that for all $1 \leq k \leq K$,
\begin{align}
	\label{eq:est-mu-use5}
	\sum_{I \in \mathcal{I}_k^{+}} a_I + \left( \frac{\bmax}{\bmin}\right)^{k} \sum_{I \in \mathcal{I}_k^{-}} a_I \geq c \frac{k^q}{k!} \gamma^k.
\end{align}
If $K = \infty$, we further require
$\gamma < \frac{1}{2\bmin}$. 
Then, 
\begin{align}
	\mu_{\Haf}
	\geq \left( a_0 + \frac{c}{\sqrt{\pi}} e^{\frac{1}{25}- \frac{1}{6}} \multilog_{\frac{1}{2}-q, K}\left(2 \gamma \bmin \right) \right).
\end{align}
\end{proposition}

\begin{proof}
The proof is omitted here, being similar to  Proposition \ref{eq:est_use52}.

\end{proof}

\begin{corollary}
\label{cor:mainpp2}
Assume all $a_I \geq 0$, $\bmin >0$ and \eqref{eq:absconv}.
Suppose there exists $\gamma_\alpha, \gamma_\beta, c_\alpha, c_\beta$ and $q_\alpha, q_\beta$ such that for all $1 \leq k \leq K$
	\begin{gather}
	 c_\alpha \frac{ k^{q_\alpha} \gamma_\alpha^{k}}{k!} \leq \sum_{\vert I \vert = 2k} a_I, \label{eq:condcond21}
	\end{gather}
	and 
	\begin{align*}
		 \sum_{\vert I \vert = 2k} \vert a_I  \vert I! \leq c_\beta k^{q_\beta}{\gamma_\beta^k} \left( \frac{e}{k} \right)^k
	\end{align*}
	If $K = \infty$, we further require $\gamma_\alpha < \frac{1}{2 \bmin}$ and $\gamma_\beta < \frac{1}{2 \bmax}$. 
	Then,
	\begin{align}
		n_{\Haf}^{\textnormal{GBS-P}} \leq \left( \frac{1}{d}\frac{a_0 +  c_\beta \sqrt{\pi} e^{\frac{1}{6}-\frac{1}{25}} \multilog_{-\frac{1}{2} -q_\beta, K} \left(\frac{\gamma_\beta }{2\bmin} \right)  }{ a_0 + \frac{c_\alpha}{\sqrt{\pi}} e^{\frac{1}{25}- \frac{1}{6}} \multilog_{\frac{1}{2}-q_\alpha, K}\left( 2 \gamma_\alpha \bmin \right)} - 1 \right) \frac{1}{\delta \epsilon^2}.
	\end{align}
\end{corollary}

\begin{proof}
The proof is omitted here, being similar to Corollary \ref{cor:mainpp1}.
\end{proof}

	\subsection{Estimates for plain Monte Carlo}
As it has already been discussed in the introduction, a natural way to construct a Monte Carlo estimator for $\mathcal{I}^\times_{\Haf^2}(\epsilon, \delta)$ (or $\mathcal{I}^\times_{\Haf}(\epsilon, \delta)$) is to compute the average of $f$ evaluated at random draws from $h$, where $f$ and $h$ are defined as in \eqref{eq:fh} (or \eqref{eq:fh2}). Precisely, let $X_1, X_2, \dots$ be i.i.d. with probability density function $h$. We define
\begin{equation}
	\mathcal{E}^\textnormal{MC}_n = \frac{1}{n} \sum_{l= 1}^n f(X_l).
	\label{eq:MC}
\end{equation}
\noindent
We call $\mathcal{E}^\textnormal{MC}_n$ the plain MC estimator. Straightforwardly, the plain MC estimator is unbiased. Therefore, we can use WLLN to show that $\mathcal{E}^\textnormal{MC}_n$ solves $\mathcal{I}^\times_{\Haf^2}(\epsilon, \delta)$ (or $\mathcal{I}^\times_{\Haf^2}(\epsilon, \delta)$) with a large enough $n$. 

\subsection{Estimate for $n_{\Haf^2}^\textnormal{MC}$}
We first establish a guaranteed sample size $n_{\Haf^2}^\textnormal{MC}$ for $\mathcal{E}^\textnormal{MC}_n$ to solve $\mathcal{I}^\times_{\Haf^2}(\epsilon, \delta).$

\begin{lemma}
	\label{lem:n-mc-hafsq}
	Suppose \eqref{eq:absconvsq} holds and
	\begin{align}
		\label{eq:swapuse}
		\int_{\mathbb{R}^{2N}} \sumdoublefinite \vert a_I a_J p^{I + J} q^{I + J} \vert h(p, q) \, \dd p \dd q < \infty. \tag{A3'}
	\end{align}
	Then the plain MC estimator solves $\mathcal{I}^\times_{\Haf^2}(\epsilon, \delta)$ with
	\begin{equation} 
		n \geq  \frac{Q_{\Haf^2}^\text{MC} - \vert \mu_{\Haf^2} \vert^2}{\delta \epsilon^2 \vert \mu_{\Haf^2} \vert^2} \equiv n_{\Haf^2}^\textnormal{MC},
		\label{eq:mc22}
	\end{equation}
	where 
	\begin{align}
		Q_{\Haf^2}^\textnormal{MC} &= \displaystyle \sumdoublefinite a_I a_J \Haf(B_{I + J})^2. \label{eq:mc2}
	\end{align}
\end{lemma}

\begin{proof}
	For convenience, let $\mu = \mu_{\Haf^2}$. We compute the variance of $\mathcal{E}^\textnormal{MC}_n$ to be 
	
	\begin{align*}
		\text{Var}[\mathcal{E}^\textnormal{MC-C}_n] & =  \frac{1}{n^2} \left( \text{Var}[f(X_1)] + \dots \text{Var}[f(X_n)] \right) \\
		& = \frac{n}{n^2} \text{Var}[f(X_1)] \\
		& = \frac{1}{n} \mathbb{E}[ \left( f(X_1) - \mu \right)^2 ] \\
		& = \frac{1}{n} \left( \mathbb{E}[f(X_1)^2] - \mu^2 \right) \\
		& = \frac{1}{n} \left( \int_{\mathbb{R}^{2N}} f(p, q)^2 h(p, q) \, \dd p \dd q  - \mu^2 \right)  \\
		& = \frac{1}{n} \left( \int_{\mathbb{R}^{2N}}  \sumdoublefinite a_I a_J p^{I + J} q^{I + J} h(p, q) \, \dd p \dd q  - \mu^2 \right) \\
		& \overset{*}{=} \frac{1}{n} \left(  \sumdoublefinite a_I {a}_J \int_{\mathbb{R}^{2N}}  p^{I + J} q^{I + J} h(p, q) \, \dd p \dd q - \mu^2 \right) \\
		& =  \frac{1}{n} \left( \sumdoublefinite a_I {a}_J \Haf(B_{I + J})^2 - \mu^2 \right) \\
		& = \frac{1}{n} \left(Q_{\Haf^2}^\text{MC} - \mu^2 \right).
	\end{align*}
	In (*), the interchange of the integral and the sum follows from the condition \eqref{eq:swapuse} and Fubini's theorem. 
	
	Next, we show that $Q_{\Haf^2}^\text{MC} - \vert \mu \vert ^2$ is finite.
	Note that 
	\begin{align*}
		Q_{\Haf^2}^\text{MC} 
		& \leq \sumdoublefinite \vert a_I {a}_J \vert \Haf(B_{I + J})^2 \\
		& \leq  \sumdoublefinite \vert a_I a_J\vert \int_{\mathbb{R}^{2N}} \vert p^{I + J} q^{I + J}\vert h(p, q) \, \dd p \dd q \\
		& =  \int_{\mathbb{R}^{2N}} \sumdoublefinite \vert a_I a_J p^{I + J} q^{I + J}\vert h(p, q) \, \dd p \dd q ~~~~\text{(Fubini-Tolleni theorem)} \\
		& < \infty
	\end{align*}
	The finiteness follows from \eqref{eq:swapuse}. Futher, $\vert \mu \vert$ is also finite by \eqref{eq:absconvsq}. 
	Finally, via Chebyshev's inequality, we arrive at the desired formula for $ n_{\Haf^2}^\textnormal{MC}$.
\end{proof}

We now give estimates for $n_{\Haf^2}^\text{MC}$. 
\begin{theorem}
	\label{thrm:mcpp1}
	Suppose $ \vert \mu_{\Haf^2} \vert > \mu_0$ for some $\mu_0 >0$. Suppose there exists $c, q$ and $\gamma$ such that for all $1 \leq k \leq K$
	\begin{gather}
		 \sum_{\vert I \vert = 2k} \vert a_I \vert \leq c \frac{k^{q} \gamma^{k}}{(2k)!},
	\end{gather}
	If  $K = \infty$, we also require $\gamma < \frac{1}{4\bmax^2}$. Then
	\begin{align*}
		n_{\Haf^2}^{\textnormal{MC}} \leq \frac{1}{\delta \epsilon^2 \mu_0^2} \left( a_0^2 + \frac{2}{\sqrt{\pi}} \vert a_0 \vert c \multilog_{{\frac{1}{2}- q}, K}({\gamma} \bmax^{2}) +  \frac{c^2}{2\sqrt{\pi}} \multilog_{\frac{1}{2}-2q, K}\left(4 {\gamma} \bmax^2  \right) - \mu_0^2 \right).
	\end{align*}
\end{theorem}

The proof of Theorem \ref{thrm:mcpp1} requires that \eqref{eq:swapuse} holds. In the next lemma we show that this is indeed true.

\begin{lemma}
	\label{lem:swap13}
	Suppose the conditions of Theorem \ref{thrm:mcpp1} are given, then \eqref{eq:swapuse} holds.
\end{lemma}

\begin{proof}
The proof is similar to that of Lemma \ref{lem:swap1}. Let $S$ and $S'$ be defined the same as in Lemma \ref{lem:swap1}.
	We divide the integral in \eqref{eq:swapuse} into two parts $ V_S$ and $V_{S'}$, where $V_S$ denotes the integral over the set $S$ and $V_{S'}$ denotes the integral over the complement set $S'$. First, it is straightforward that $V_S$ is finite. Second, for $V_{S'}$, we apply Tolleni's theorem for non-negative functions and get 
	\begin{align*}
		V_{S'} = \sum_{k_1, k_2 = 0}^{K} \sum_{\substack{\vert I \vert = 2k_1 \\ \vert J \vert = 2k_2}} \vert  a_{\vecI} a_{J}\vert  \int_{S'}  \vert p^{I + J} q^{I + J} \vert \, h(p, q) \, \dd p \dd q.
	\end{align*}
	From Lemma \ref{lem:useful135}, there exists $L_{(I + J)} = (l_1, \dots, l_N) \in \mathbb{N}^N$ such that 
	$$\vert L_{(I + J)} \vert = 2k_1 + 2k_2 + 2m$$ 
	where $m = \lfloor \frac{N}{2} \rfloor$ for any given $I = (i_1, i_2, \dots, i_N)$ with $\vert I \vert = 2k_1$ and $J = (j_1, j_2, \dots, j_N)$ with $\vert J \vert = 2k_2$ with $k_1 + k_2 \geq 1$. Furthermore, for any $n \in \{1, 2, \dots, N\}$, 
	$$l_n \geq i_n +j_n  ~~~~\text{ and }~~~~ l_n \text{ is even.}$$
	Thus, $\vert p^{I + J} q^{I + J}  \vert \leq p^{L_{(I + J)}} q^{L_{(I + J)}} $ for all $(p, q) \in S'$. 
	From Wick's theorem, we get
	\begin{align*}
		V_{S'} &\leq  \sum_{k_1, k_2 = 0}^{K} \sum_{\substack{\vert I \vert = 2k_1 \\ \vert J \vert = 2k_2}} \vert  a_{\vecI} a_J \vert  \Haf(B_{L_{(I + J)}})^2 \\
		& = a_0^2 + \vert a_0 \vert \sum_{k_2 = 1}^{K}  \sum_{{ \vert J \vert = 2k_2}} \vert a_J \vert  \Haf(B_{L_{(J)}})^2 + \vert a_0 \vert \sum_{k_1 = 1}^{K} \sum_{{ \vert I \vert = 2k_1}} \vert a_I  \vert  \Haf(B_{L_{(I)}})^2 \\
		& ~~~~~~ + \sum_{k_1, k_2 = 1}^{K} \sum_{\substack{\vert I \vert = 2k_1 \\ \vert J \vert = 2k_2}} \vert  a_{\vecI} a_J \vert  \Haf(B_{L_{(I + J)}})^2.
	\end{align*}
	From a similar computation as in the proof of Lemma \ref{lem:swap1}, one can show that 
	$\displaystyle \sum_{k_2 = 1}^{K} \sum_{{ \vert J \vert = 2k_2}} \vert a_J \vert  \Haf(B_{L_{(J)}})^2$ and $\displaystyle \sum_{k_1 = 1}^{K} \sum_{{ \vert I \vert = 2k_1}} \vert a_I  \vert  \Haf(B_{L_{(I)}})^2$ are finite provided that $\gamma \bmax^2 < 1$. 
	
	We further note that
	\begin{align*}
		& \sum_{k_1, k_2 = 1}^{K} \sum_{\substack{\vert I \vert = 2k_1 \\ \vert J \vert = 2k_2}} \vert  a_{\vecI} a_J \vert  \Haf(B_{L_{(I + J)}})^2 \\
		\leq &\sum_{k_1, k_2 = 1}^{K}  \,  \sum_{\substack{\vert I \vert = 2k_1 \\ \vert J \vert = 2k_2}} \vert  a_{\vecI} a_J \vert \, \bmax^{2k_1 + 2k_2 + 2m} \left( \frac{2k_1 + 2k_2 + 2m }{e} \right)^{2k_1 + 2k_2 + 2m} \\
		&  \hspace{2.0in} \text{(use Lemma \ref{lem:hafapprox})} \\
		\leq &\sum_{k_1, k_2 = 1}^{K} c^2 (k_1k_2)^q \frac{1}{(2k_1)!(2k_2)!} \gamma^{k_1 + k_2} \,\bmax^{2k_1 + 2k_2 + 2m} \left( \frac{2k_1 + 2k_2 + 2m }{e} \right)^{2k_1 + 2k_2 + 2m}\\
		= & \sum_{l = 2}^{2K} \sum_{\substack{k_1 + k_2 = l \\ k_1 \geq 1, k_2 \geq 1}} c^2 (k_1k_2)^q \frac{1}{(2k_1)!(2k_2)!} \gamma^{l} \,\bmax^{2l+ 2m} \left( \frac{2l + 2m }{e} \right)^{2l + 2m} \\
		\leq &  \frac{c^2 \bmax^{2m}}{\sqrt{2\pi}}  \sum_{l = 2}^{2K} l^{2q} \bmax^{2l} \gamma^l\sum_{\substack{k_1 + k_2 = l \\ k_1 \geq 1, k_2 \geq 1}}   \, \frac{(2l)!}{(2k_1)!(2k_2)!} \sqrt{(2l + 2m)} (2l+2m-1) \dots (2l+1)  \\
		&  \hspace{1.5in} (\text{use}~ k_1k_2 \leq l^2 \text{ and Stirling's approxiamtion}) \\
		\leq & \frac{c^2 \bmax^{2m}}{2\sqrt{2\pi}}  \sum_{l = 2}^{2K} l^{2q} \bmax^{2l} \gamma^l 2^{2l} \sqrt{(2l + 2m)} (2l+2m-1) \dots (2l+1).
	\end{align*}
	If $l \geq m$, we have
	\begin{align*}
		(2l+1) < (2l+2) < \dots < (2l+2m) \leq 4l
	\end{align*}
	With $m = \lfloor \frac{N}{2} \rfloor \geq 1$, we obtain
	\begin{align*}
		\sqrt{(2l + 2m)} (2l+2m-1) \dots (2l+1) \leq (2l + 2m)^{2m- \frac{1}{2}} \leq (4l)^{2m-\frac{1}{2}}.
	\end{align*}
	Thus,
	\begin{align*}
		&\sum_{l= m+1}^\infty  \,  l^{2q} \bmax^{2l} \gamma^l 2^{2l} \sqrt{(2l + 2m)} (2l+2m-1) \dots (2l+1) \\
		&~~~~~~~~\leq  \sum_{k = m+1}^\infty  \,  l^{2q} (4 \gamma \bmax^{2})^l (4l)^{2m - \frac{1}{2}} \\
		& ~~~~~~~~= 4^{2m - \frac{1}{2}}  \sum_{k = m+1}^\infty  \, (4 \gamma \bmax^{2})^l l^{2q + 2m - \frac{1}{2}} \\
		& ~~~~~~~~\leq  4^{2m - \frac{1}{2}} \multilog_{\frac{1}{2} - 2m- 2q}(4 \gamma \bmax^{2})
	\end{align*}
	is finite given $4 \gamma \bmax^{2} < 1$. 
	In total, we have shown that $V_{S'} < \infty$, and therefore, $V_S + V_{S'}$ is finite. This proves \eqref{eq:swapuse}.
\end{proof}

\begin{proof}[Proof of Theorem \ref{thrm:mcpp1}]
	It is clear from Lemmas \ref{lem:swap1} and \ref{lem:swap13} that \eqref{eq:absconvsq} and \eqref{eq:swapuse} holds. 
	We just need to compute an upper bound for $Q^{\textnormal{MC}}_{\Haf^2}$. 
	\begin{align*}
		Q^{\textnormal{MC}}_{\Haf^2}
		& \leq  a_0^2 + \vert a_0 \vert \sum_{k_2 = 1}^K \sum_{\vert J \vert = 2k_2} \vert a_J \vert \Haf(B_{J})^2 + \vert a_0 \vert \sum_{k_1 = 1}^K \sum_{\vert I \vert = 2k_1} \vert a_I \vert \Haf(B_I)^2 \\
		& ~~~~~~~ + \sum_{k_1, k_2 = 1}^{K}  \sum_{\substack{\vert I \vert = 2k_1 \\ \vert J \vert = 2k_2}}  \vert a_I a_J \vert \Haf(B_{I + J})^2 \\
		& \leq a_0^2 + 2 \vert a_0 \vert c \sum_{k = 1}^K {k^{q}} {\gamma}^{k} \bmax^{2k} \frac{(2k)!}{2^{2k} (k)!^2}  \\
		&  ~~~~~~ + c^2 \sum_{k_1, k_2 = 1}^K \frac{k^{q}_1 k^{q}_2}{(2k_1)!(2k_2)!} {\gamma}^{(k_1+k_2)} \bmax^{2(k_1+k_2)} \frac{(2k_1 + 2k_2)!^2}{2^{2k_1 + 2k_2} (k_1 + k_2)!^2} \\
		& \hspace{2in} \text{(Lemma \ref{lem:hafapprox0})} \\
		& \leq  a_0^2 + \frac{2 \vert a_0 \vert c}{\sqrt{\pi}} \sum_{k = 1}^K k^{q - \frac{1}{2}} {\gamma}^{k} \bmax^{2k} ~~~~~~\text{(Lemma \ref{lem:usefulstirling1})} \\
		& ~~~~~~ + c^2 \sum_{l = 1}^K \sum_{k_1 + k_2 = l} l^{2q} \left( \frac{1}{4} {\gamma} \bmax^2 \right)^l \frac{(2l)!^2}{(l!)^2 (2k_1)!(2k_2)!} \\
		& \hspace{2in} (\textnormal{since } k_1 k_2 \leq l^2) \\
		& \leq a_0^2 + \frac{2 \vert a_0 \vert c}{\sqrt{\pi}} \multilog_{{\frac{1}{2}- q}, K}({\gamma} \bmax^{2})\\
		& ~~~~~~ + \frac{c^2}{\pi}  \sum_{l = 1}^K l^{2q-1}  \left( 4 {\gamma}  \bmax^2 \right)^l \sum_{k_1 + k_2 = l} \frac{(l!)^2}{(2k_1)!(2k_2)!} \\
		& \hspace{2in} \text{(Lemma \ref{lem:useful111})} \\
		& \leq a_0^2 + \frac{2 \vert a_0 \vert c}{\sqrt{\pi}} \multilog_{{\frac{1}{2}- q}, K}({\gamma} \bmax^{2}) +  \frac{c^2}{2\sqrt{\pi}} \sum_{l = 1}^K l^{2 q-\frac{1}{2}} \left( 4{\gamma}  \bmax^2 \right)^l  \\
		& \hspace{2in} \text{(Lemma \ref{lem:useful123})} \\
		& = a_0^2 + \frac{2 \vert a_0 \vert c}{\sqrt{\pi}}\multilog_{{\frac{1}{2}- q}, K}({\gamma} \bmax^{2}) + \frac{c^2}{2\sqrt{\pi}} \multilog_{\frac{1}{2}-2q, K}\left(4 {\gamma} \bmax^2  \right),
	\end{align*}
	which is finite given $4 {\gamma} \bmax^2 < 1$. 
\end{proof}

\begin{corollary}
\label{prop:mcuseful2}
Assume $\bmin > 0$. 
Suppose there exists $c_\alpha$, $c_\beta$, $q_\alpha$, $q_\beta$, $\gamma_\alpha$ and $\gamma_\beta$ such that for all $1 \leq k \leq K$
\begin{gather}
	\label{cor:cond2}
	\sum_{\vert I \vert = 2k} \vert a_I \vert \leq c_\beta \frac{k^{q_\beta} \gamma_\beta^{k}}{(2k)!}, \\
	\sum_{I \in \mathcal{I}_k^{+}} a_I + \left( \frac{\bmax}{\bmin}\right)^{2k} \sum_{I \in \mathcal{I}_k^{-}} a_I   \geq c_\alpha \frac{ k^{q_\alpha} \gamma_\alpha^{k}}{(2k)!}. \notag
\end{gather}
If  $K = \infty$, we also require $\gamma_\beta <\frac{1}{4 \bmax^2}.$
Then,
	\begin{align*}
		n^{\Haf^2}_{\textnormal{MC}} \leq \left( \frac{a_0^2 + \frac{2}{\sqrt{\pi}} a_0 c_\beta \multilog_{{\frac{1}{2}- q_\beta}, K}({\gamma_\beta} \bmax^{2}) + \frac{c_\beta^2}{2 \sqrt{\pi}} \multilog_{\frac{1}{2}-2q_\beta, K}\left(4 {\gamma_\beta} \bmax^2 \right)}{ \left( a_0 + e^{\frac{1}{25}- \frac{1}{6}} \frac{c_\alpha}{\sqrt{\pi}} \multilog_{\frac{1}{2} - q_\alpha, K}(\gamma_\alpha \bmin^2) \right)^2 } - 1 \right) \frac{1}{\delta \epsilon^2}.
	\end{align*}
\end{corollary}

\begin{proof}
The proof follows directly by substituting the bound from Proposition \ref{eq:est_use52} into Theorem \ref{thrm:mcpp1}.
\end{proof}

\subsubsection{Estimate of $n_{\Haf}^{\textnormal{MC}}$}
In this section, we focus on using plain MC for solving $\mathcal{I}^\times_{\Haf}(\epsilon, \delta)$. Similar to Lemma \ref{lem:n-mc-hafsq}, we first establish the guaranteed sample size $n_{\Haf}^{\textnormal{MC}}$.
\begin{lemma}
	\label{lem:n-mc-haf}
	Suppose \eqref{eq:absconv} holds and
		\begin{align}
		\label{eq:swapuse2}
		\int_{\mathbb{R}^N} \sumdoublefinite \vert a_I a_J x^{I + J}\vert h(x) \, \dd x < \infty. \tag{A3}
	\end{align}
	The plain MC estimator solves $\mathcal{I}^\times_{\Haf}(\epsilon, \delta)$ with
	\begin{equation*} 
		n \geq  \frac{Q_{\Haf}^\text{MC} - \vert \mu_{\Haf} \vert^2}{\delta \epsilon^2 \vert \mu_{\Haf} \vert^2} \equiv n_{\Haf}^\textnormal{MC},
	\end{equation*}
	where 
	\begin{align}
		Q_{\Haf}^\text{MC} &= \displaystyle \sumdoublefinite  a_I {a}_J \Haf(B_{I + J}). \label{eq:mc1}
	\end{align}
\end{lemma}

\begin{proof}
	The proof is completely analogous to Lemma \ref{lem:n-mc-hafsq} by replacing $\Haf^2$ with $\Haf$.
\end{proof}

In the next theorem, we give estimates for $n_{\Haf}^\textnormal{MC}$ using similar conditions as described in Theorem \ref{main:pp2}. 

\begin{theorem}
	\label{thrm:mcpp2}
	Suppose  $ \vert \mu_{\Haf} \vert > \mu_0$ for some $\mu_0>0$. Suppose there exists $c, q$ and $\gamma$ such that for all $k \geq 1$ 
	\begin{gather}
	 \sum_{\vert I \vert = 2k} 	\vert a_I \vert \leq c \frac{k^{q} \gamma^{k}}{k!},
	\end{gather}
	If  $K = \infty$, we also require $\gamma < \frac{1}{4\bmax}$.
	Then, 
	\begin{align*}
		n_{\Haf}^{\textnormal{MC}} \leq \frac{1}{\delta \epsilon^2 \mu_0^2} \left( a_0^2 + \frac{2}{\sqrt{\pi}} \vert a_0 \vert  c \multilog_{\frac{1}{2}- q, K}({\gamma} \bmax) +  \frac{c^2 }{2\sqrt{\pi}} \multilog_{\frac{1}{2} - 2 q, K}\left(4 {\gamma} \bmax  \right) - \mu_0^2 \right).
	\end{align*}
\end{theorem}

The proof of Theorem \ref{thrm:mcpp1} requires that \eqref{eq:swapuse2} holds. Using a similar technique as in the proof of Lemma \ref{lem:swap13}, one can show that this is indeed true.

\begin{lemma}
	\label{lem:swap23}
	Suppose the conditions of Theorem \ref{thrm:mcpp2} are given, then \eqref{eq:swapuse2} holds.
\end{lemma}

\begin{proof}
The proof is omitted being similar to  Lemma \ref{lem:swap13}.
\end{proof}

\begin{proof}[Proof of Theorem \ref{thrm:mcpp2}]
		We just need to compute $Q^{\Haf}_\textnormal{MC}$.
	\begin{align*}
		Q^{\Haf}_\textnormal{MC} & = \sumdoublefinite  a_I a_J \Haf(B_{I +J}) \\ 
		& = a_0^2 + a_0 \sum_{k_2 = 1}^K \sum_{\vert J \vert = 2k_2} a_J \Haf(B_{J}) + a_0 \sum_{k_1 = 1}^K \sum_{\vert I \vert = 2k_1} a_I \Haf(B_I) \\
		& ~~~~~~~ + \sum_{k_1, k_2 = 1}^{K}  \sum_{\substack{\vert I \vert = 2k_1 \\ \vert J \vert = 2k_2}}  a_I a_J \Haf(B_{I + J}) \\
		& \leq a_0^2 + 2 \vert a_0 \vert c \sum_{k = 1}^K \frac{k^{q}}{k!} {\gamma}^{k} \bmax^{k} \frac{(2k)!}{2^{k} (k)!}  \\
		&  ~~~~~~ + c^2 \sum_{k_1, k_2 = 1}^K \frac{k_1^{q} k_2^{q}}{k_1!k_2!} {\gamma}^{k_1 + k_2} \bmax^{k_1+k_2} \frac{(2k_1 + 2k_2)!}{2^{k_1 + k_2} (k_1 + k_2)!} \\
		& \hspace{2in} \text{(Lemma \ref{lem:hafapprox0})} \\
		& =  a_0^2 + 2 \vert a_0 \vert c \sum_{k = 1}^K k^{q} {\gamma}^{k} \bmax^{k} \frac{1}{\sqrt{\pi k}} ~~~~\text{(Lemma \ref{lem:usefulstirling1})}  \\
		&  ~~~~~~ + c^2 \sum_{l = 1}^K \sum_{k_1 + k_2 = l} l^{2q} \left( \frac{1}{2} {\gamma} \bmax \right)^l  \frac{(2l)!}{l!} \sum_{k_1 + k_2 = l} \frac{1}{k_1!k_2!} \\
		& \hspace{2in} (\textnormal{since } k_1 k_2 \leq l^2) \\
		& =  a_0^2 + \frac{2 \vert a_0 \vert c }{\sqrt{\pi}} \multilog_{\frac{1}{2}- q, K}({\gamma} \bmax) \\
		&  ~~~~~~ + c^2 \sum_{l = 1}^K \sum_{k_1 + k_2 = l} l^{2q} \left( \frac{1}{2} {\gamma} \bmax \right)^l  \frac{(2l)!}{(l!)^2} \sum_{k_1 + k_2 = l} \frac{l!}{k_1!k_2!}  \\
		& = a_0^2  + \frac{2 \vert a_0 \vert c }{\sqrt{\pi}} \multilog_{\frac{1}{2}- q, K}({\gamma} \bmax) +  c^2 \sum_{l = 1}^K \sum_{k_1 + k_2 = l} l^{2q} \left( {\gamma} \bmax \right)^l  \frac{(2l)!}{(l!)^2} \\
		& \leq a_0^2 + \frac{2 \vert a_0 \vert c}{\sqrt{\pi}} \multilog_{\frac{1}{2}- q, K}({\gamma} \bmax) + c^2 \sum_{l = 1}^K \sum_{k_1 + k_2 = l} l^{2q} \left( {\gamma} \bmax \right)^l  2^{2l} \frac{1}{\sqrt{\pi l}}\\
		& \hspace{2in} \text{(Lemma \ref{lem:usefulstirling1})} \\
		& \leq  a_0^2  + \frac{2 \vert a_0 \vert c}{\sqrt{\pi}} \multilog_{\frac{1}{2}- q, K}({\gamma} \bmax) +  \frac{c^2 }{\sqrt{\pi}} \multilog_{\frac{1}{2} - 2 q, K}\left(4 {\gamma} \bmax  \right),
	\end{align*}
	which is finite since we have assumed that $\gamma < \frac{1}{4 \bmax}$. Plugging in this upper bound completes the proof of this theorem.
\end{proof}

\begin{corollary}
Assume $\bmin > 0$. 
Suppose there exists $c_\alpha$, $c_\beta$, $q_\alpha$, $q_\beta$, $\gamma_\alpha$ and $\gamma_\beta$ such that for all $k \geq 1$ 
\begin{gather}
	\sum_{\vert I \vert = 2k} \vert a_I \vert  \leq c_\beta \frac{k^{q_\beta} \gamma_\beta^{k}}{k!}, \label{cor:cond1}, \\
	\sum_{I \in \mathcal{I}_k^{+}} a_I + \left( \frac{\bmax}{\bmin}\right)^{k} \sum_{I \in \mathcal{I}_k^{-}} a_I \geq c_\alpha \frac{ k^{q_\alpha} \gamma_\alpha^{k}}{k!}. \notag 
\end{gather}
If $K = \infty$, we further require $\gamma_\beta < \frac{1}{4 \bmax}.$
Then
	\begin{align*}
		n^{\Haf}_{\textnormal{MC}} \leq \left( \frac{a_0^2 + \frac{2}{\sqrt{\pi}} a_0  c_\beta \multilog_{\frac{1}{2}- q_\beta, K}({\gamma_\beta} \bmax) +  \frac{c_\beta^2 }{\sqrt{\pi}} \multilog_{\frac{1}{2} - 2 q_\beta, K}\left(4 {\gamma_\beta} \bmax  \right)}{ \left( a_0 + e^{\frac{1}{25}- \frac{1}{6}} \frac{c_\alpha}{\sqrt{\pi}} \multilog_{\frac{1}{2}-q_{\alpha}, K}(\gamma_\alpha \bmin) \right)^2} - 1 \right) \frac{1}{\delta \epsilon^2}.
	\end{align*}
\end{corollary}

\begin{proof}
The proof follows directly by substituting the bound from Proposition \ref{prop:est-mu-use5} into Theorem \ref{thrm:mcpp2}.
\end{proof}
	\section{Efficiency comparison}
In this section, we prove our main results which compares the efficiency between the GBS estimators and the MC estimator for solving $\mathcal{I}^\times_{\Haf^2}(\epsilon, \delta)$ and $\mathcal{I}^\times_{\Haf}(\epsilon, \delta)$. Throughout this section, we assume all $a_I \geq 0$ and $\bmin >0$. 

\subsection{Proof of Theorem \ref{thrm:mani1}}
\label{subsec:thrm1}
We first compare GBS-I and plain MC. In order to state the theorem precisely, let us first introduce the following notations.
For a given $k$, we define non-negative integers $s_k$ and $r_k$ such that 
\begin{align}
	\label{eq:sk}
	2k = Ns_k + r_k  ~~~~\text{with}~~~~ 1 \leq r_k \leq N.
\end{align}
Further we let $s_{\infty} = \infty$ and $r_{\infty} = 0$. 
Then, we define
\begin{align}
	m_k = (s_k!)^N (s_k +1)^{r_k}. \label{def:mk}
\end{align}
For $q \in \mathbb{R}$, let us define $q^{-}$ to be
\begin{align*}
	q^{-} = \min(q, 0).
\end{align*}
Let $\hfun_{q, N}: \mathbb{R} \rightarrow \mathbb{R}$ be given by
\begin{align}
	\hfun_{q, N}(z)  = \sum_{k = 1}^{N} \frac{{k}^q}{(2k)!} z^{2k} =  \frac{1^q}{2!} z^2 + \frac{2^q}{4!} z^4 \dots + \frac{N^q}{(2N)!} z^{2N}.
	\label{eq:cosh}
\end{align}
Let $G_{q,K, N}: \mathbb{R} \rightarrow \mathbb{R}$ be defined by
\begin{align}
	\label{def:funGK}
	G_{q,K, N}(z) = \hfun_{q, \lfloor \frac{N}{2} \rfloor}(z) + \left( 2\pi \right)^{\frac{N-1}{2}} N^{q - \frac{1}{2}}
	e^{\frac{N}{13}} \multilog_{0, N} \left(\frac{2z}{N}
	\right)
	\multilog_{\frac{1}{2} - \frac{N}{2}- q, s_K} \left(\frac{z^N}{N^N}
	\right),
\end{align}
where $s_K$ is given as in $\eqref{eq:sk}$. 
If $K = \infty$, then $G_{q, K, N}(z)$ is defined for $\vert z \vert < N$ only. If $K < \infty$, then $G_{q, K, N}(z)$ is defined for all $z \in \mathbb{R}$. 
Let $R_{q,K}$ be defined such that if $q\geq 0$
\begin{align}
	\label{def:funRKpos}
	R_{q,K}(z) = 2^{-q} \frac{1}{ 2\sqrt{\pi}} \multilog_{\frac{1}{2} - q, K} (z),
\end{align}
and if  $q < 0$
\begin{align}
	\label{def:funRKneg}
	R_{q,K}(z) = \frac{1}{ 2\sqrt{\pi}} \multilog_{\frac{1}{2} - 2q, K} (z).
\end{align}

We now introduce the space of $a_I$'s and $B$'s on which we will compare GBS-I and MC for solving $\mathcal{I}^\times_{\Haf^2}(\epsilon, \delta)$, and this also compares $n_{\Haf^2}^{\textnormal{GBS-I}}$ and $n_{\Haf^2}^{\textnormal{MC}}$ when $0 < \epsilon, \delta < 1$ are given. Let
\begin{align}
	\label{eq:Aspace}
	\mathcal{A} = \prod_{k=0}^K \mathcal{A}_k, ~~~~ \mathcal{A}_k = \mathbb{R}_{\geq 0}^{\mathcal{I}_k}
\end{align}
where
\begin{align*}
	\mathcal{I}_k = \{I \in \mathbb{N}^N \mid \vert I \vert = 2k\}
\end{align*}
and $\mathbb{R}_{\geq 0}$ denotes the set of non-negative real numbers. On $\mathcal{A}$, we consider the product topology of the usual topologies on the $\mathcal{A}_k$'s, noticing that
\begin{align*}
	\mathcal{A}_k =\mathbb{R}_{\geq 0}^{\mathcal{I}_k} \subseteq \mathbb{R}^{\mathcal{I}_k} \cong \mathbb{R}^{\vert \mathcal{I}_k \vert}.
\end{align*}
Next, we introduce the space of the allowed matrices
\begin{align}
	\label{eq:Bspace}
	\mathcal{B} = \left\{ B \in M_{N \times N}(\mathbb{R}_+) \mid B^\intercal = B, 0< \textnormal{spec}(B) < 1\right\}.
\end{align}
Here, $\mathbb{R}_+$ denotes the set of strictly positive real numbers. So, $\bmin >0$. For a fixed $B \in \mathcal{B}$, we introduce the following topological subspace $\mathcal{A}_B \subseteq \mathcal{A}$ given by
\begin{align*}
	\mathcal{A}_B = & \left\{ (a_I) \in \mathcal{A} \Bigm| 0 \neq \sum_{k = 0}^K \sumck a_I \Haf(B_I)^2 < \infty, \right. \\
	&\quad \quad \quad \quad \quad~~  \sum_{k = 0}^K \sumck a_I^2 I! \Haf(B_I)^2 < \infty, \\
	&\quad \quad \quad \quad \quad  \left. \sumdoublefinite a_I a_J \Haf(B_{I + J})^2 < \infty \right\}.
\end{align*}
By definition, $\mathcal{I}^\times_{\Haf^2}(\epsilon, \delta)$, $n_{\Haf^2}^{\textnormal{GBS-I}}$ and $n_{\Haf^2}^{\textnormal{MC}}$ are well-defined on $\mathcal{A}_B$. Furthermore, when $K < \infty$, we have $\mathcal{A}_B = \mathcal{A}$. Then, we define
$$ \mathcal{P} = \bigcup_{B \in \mathcal{B}} \mathcal{A}_B \subseteq \mathcal{A} \times \mathcal{B}, $$
on which we will compare $n_{\Haf^2}^{\textnormal{GBS-I}}$ and $n_{\Haf^2}^{\textnormal{MC}}$. For future reference, let us also define $\mathcal{A}^1_B \subseteq \mathcal{A}$ such that
\begin{align*}
	\mathcal{A}^1_B = & \left\{ (a_I) \in \mathcal{A} \Bigm| \sum_{k = 0}^K \sumck a_I \Haf(B_I)^2 =1, \right. \\
	&\quad \quad \quad \quad \quad~~  \sum_{k = 0}^K \sumck a_I^2 I! \Haf(B_I)^2 < \infty, \\
	&\quad \quad \quad \quad \quad  \left. \sumdoublefinite a_I a_J \Haf(B_{I + J})^2 < \infty \right\}.
\end{align*}

Given $s_1, s_2, \gamma_{\alpha}, \gamma_{\beta} \in \mathbb{R}_{+}$ and  $q_\alpha, q_\beta \in \mathbb{R}$,
define
\begin{equation}
\label{eq:cdef}
\begin{aligned}
&c_1 = 1 + \frac{1}{\sqrt{\pi}} e^{\frac{1}{25} - \frac{1}{6}} \multilog_{\frac{1}{2} - q_\alpha, K}\left(\gamma_\alpha \bmin^{2} \right), \\
&c_2 = 1 + \frac{1}{\sqrt{\pi}} \multilog_{\frac{1}{2}-q_\beta, K}\left(\gamma_\beta \bmax^{2} \right).
\end{aligned}
\end{equation}
We further introduce the subset $ \mathcal{B}_{\alpha, \beta} \subseteq \mathcal{B}$ defined by the following conditions.
We require
\begin{equation}
	\label{eq:cgamma}
	\begin{gathered}
		c_2^{-1} < c_1^{-1} \\
		c_2^{-1} { k^{q_\alpha} \gamma_\alpha^{k}} <  c_1^{-1} k^{q_\beta} \gamma_\beta^{k}, ~~~~ k = 1, 2, \dots, K.
	\end{gathered}
\end{equation}
Note that the condition \eqref{eq:cgamma} is clearly met if $\bmin < \bmax$, $q_\alpha \leq q_\beta$ and $\gamma_\alpha \leq \gamma_{\beta}$.
We further ask that
\begin{equation}
	\label{eq:thrmI1-3} 
	\begin{aligned}
		&\frac{s_1}{\epsilon^2 \delta} \frac{ c_1^{-2}}{d}  \left( 1 + \frac{1}{\sqrt{\pi}} G_{2 q_{\beta}, K, N}(\gamma_\beta \bmax) \right)\\
		& < \ln\left( 1 + e^{\frac{1}{25} - \frac{1}{6}}  R_{q_\alpha, K} (4 \gamma_\alpha \bmin^{2}) \right)+ 2 \ln\left(c_2^{-1} \right) - \ln(\epsilon^2 \delta + 1) + \frac{s_1}{\epsilon^2 \delta} - \ln(s_2).
	\end{aligned}
\end{equation}
If $K = \infty$, we also require
\begin{gather}
	4 \gamma_\beta \bmax^2 < 1 \label{eq:thrmI1-4}  \\
	\gamma_\beta \bmax < N \label{eq:thrmI1-5} 
\end{gather}	
In other words,
\begin{align*}
	\mathcal{B}_{\alpha, \beta} = \left\{ B \in \mathcal{B} \mid \eqref{eq:cgamma}, \eqref{eq:thrmI1-3}, \eqref{eq:thrmI1-4} \textnormal{ and } \eqref{eq:thrmI1-5} \textnormal{ hold} \right\}.
\end{align*}
Of course, when $K \neq \infty$, the conditions \eqref{eq:thrmI1-4} and \eqref{eq:thrmI1-5} are vacuous. We further note that $\mathcal{B}_{\alpha, \beta}$ depends on $N, K, \epsilon, \delta, s_1, s_2, \gamma_\alpha, \gamma_\beta, q_\alpha, q_\beta$.

We are now ready to give a precise statement for Theorem \ref{thrm:mani1}.

\begin{theorem}
	\label{thrm:mainI1}
	Let $N$, $K$ and $0 < \epsilon, \delta < 1$ be given. Let $s_1, s_2, \gamma_{\alpha}, \gamma_{\beta} \in \mathbb{R}_{+}$ and  $q_\alpha, q_\beta \in \mathbb{R}$. Suppose $B \in \mathcal{B}_{\alpha, \beta}$ and the $a_I$'s satisfy the following requirements: there exists $I$ such that $a_I \neq 0$ and 
	\begin{align*}
		c_2^{-1} < a_0 < c_1^{-1}.
	\end{align*}
	We further require that for all $1 \leq k \leq K$,
	\begin{gather}
		c_2^{-1} \mu_{\Haf^2} \frac{ k^{q_\alpha} \gamma_\alpha^{k}}{(2k)!} \leq \sum_{\vert I \vert = 2k} a_I \leq c_1^{-1} \mu_{\Haf^2} \frac{k^{q_\beta} \gamma_\beta^{k}}{(2k)!},
		\label{eq:thrmI1-1} \\
		\sumck a^2_I I! \leq c_1^{-2} \mu^2_{\Haf^2} \frac{k^{2q_\beta} \gamma_\beta^{2k}}{(2k)!^2} m_k. \label{eq:thrmI1-2} 
	\end{gather}
	Then, 
	\begin{align}
		s_2 \exp(s_1 \, n_{\Haf^2}^{\textnormal{GBS-I}} ) < n_{\Haf^2}^{\textnormal{MC}} < \infty. \label{eq:thrmI1-6} 
	\end{align}
	Moreover, for any $0 < \epsilon, \delta < 1$, $N$ and sufficiently large $K$, there exists a non-empty and open subset of $\mathcal{P} $ for which the conditions above are satisfied.  
\end{theorem}

We can further show that there are problems where GBS-P outperforms MC uniformly across $0< \epsilon, \delta <1$.
To facilitate this extension, we introduce a new condition
\begin{align}
	\frac{1}{d}  \left( s + \frac{1}{\sqrt{\pi}} G_{2 q_{\beta}, K, N}(\gamma_\beta \bmax) \right) <c \left( e^{\frac{1}{25} - \frac{1}{6}}  R_{q_\alpha, K} (4 \gamma_\alpha \bmin^{2}) \right) \label{eq:thrmI1-3-prime}
\end{align}
for some positive constants $s$ and $c$. 
We then define $\mathcal{B}'_{\alpha, \beta} \subseteq \mathcal{B}$ to be
\begin{align*}
	\mathcal{B}'_{\alpha, \beta} = \left\{ B \in \mathcal{B} \mid \eqref{eq:thrmI1-4}, \eqref{eq:thrmI1-5} \textnormal{ and }  \eqref{eq:thrmI1-3-prime} \textnormal{ hold} \right\}.
\end{align*}
Note that $\mathcal{B}'_{\alpha, \beta}$ depends on $N, K, \gamma_{\alpha}, \gamma_{\beta},q_\alpha, q_\beta, s$ and $c$.

\begin{theorem}
	\label{cor:mainI1}
	Let $N$ and $K$ be given.
	Let $\gamma_{\alpha}, \gamma_{\beta} \in \mathbb{R}_{+}$ and  $q_\alpha, q_\beta \in \mathbb{R}$. Let $s, c >0$. If $B \in \mathcal{B}'_{\alpha, \beta}$ and
	the $a_I$'s satisfy the following requirements: there exists $I$ such that $a_I \neq 0$, $a_0^2 \leq s$
	and for all $1 \leq k \leq K$
	\begin{gather}
		\frac{ k^{q_\alpha} \gamma_\alpha^{k}}{(2k)!} \leq \sum_{\vert I \vert = 2k} a_I \leq  \frac{k^{q_\beta} \gamma_\beta^{k}}{(2k)!}, \label{eq:corhafsq1} \\
		\sumck a^2_I I! \leq \frac{k^{2q_\beta} \gamma_\beta^{2k}}{(2k)!^2} m_k, \label{eq:corhafsq2}
	\end{gather}
	then, for all $0< \epsilon, \delta < 1$,
	\begin{align}
		n_{\Haf^2}^{\textnormal{GBS-I}} < c \, n_{\Haf^2}^{\textnormal{MC}} < \infty. \label{eq:corI1-6} 
	\end{align}
	Moreover, for any $c>0$,  $N$ and sufficiently large $K$, there exists a non-empty and open subset of $\mathcal{P}$ for which the conditions above are satisfied.
\end{theorem}

We can also compare GBS-I and MC as $N$ grows.
If $N$ and $K$ are correlated such that $K \geq \zeta N^2$ with $\zeta >0$, then we obtain an exponential speedup of GBS-I with $N$. To state this precisely, we introduce a new subset $\mathcal{B}''_{\alpha, \beta} \subseteq \mathcal{B}$. Let $\gamma_{\alpha}, \gamma_{\beta} \in \mathbb{R}_{+}$ and  $q_\alpha, q_\beta \in \mathbb{R}$. Let $p >0$. We first require that
\begin{align}
	G_{2 q_{\beta}, K, N}(\gamma_\beta \bmax) < C_G N^3 \label{eq:uniform1}
\end{align}
for some positive constant $C_G$. We further require that 
\begin{gather}
	R_{q_\alpha, K} (4 \gamma_\alpha \bmin^{2}) > \frac{1}{2\sqrt{\pi}} (\zeta N^2)^{2q_\alpha - \frac{1}{2}} C_R^{\zeta N^2} \label{eq:uniform2}
\end{gather}
where $C_R >1$ and $C_R^{\zeta N^2}$ dominates the right hand side of \eqref{eq:uniform2}. 
We finally require that for all $N \geq N_D$, 
\begin{align}
	\frac{1}{d} < C_D N^p \label{eq:uniform3}
\end{align}
for some positive constant $C_D$. 
Recall that
\begin{align*}
	\frac{1}{d} = \prod_{j=1}^N (1- \lambda_j^2)^{-\tfrac{1}{2}},
\end{align*}
where $\lambda_j$'s are eigenvalues of $B$ arranged in descending order. 
We define
\begin{align*}
	\mathcal{B}''_{\alpha, \beta} = \left\{ B \in \mathcal{B} \mid \eqref{eq:uniform1}, \eqref{eq:uniform2} \textnormal{ and }  \eqref{eq:uniform3} \textnormal{ hold} \right\}.
\end{align*}
Note that $\mathcal{B}''_{\alpha, \beta}$ depends $N, K, \gamma_{\alpha}, \gamma_{\beta},q_\alpha, q_\beta$ and $p$.

\begin{theorem}
	\label{cor:mainI1N}
	Let $N$, $K\geq \zeta N^2$ with $\zeta >0$ and $0 < \epsilon, \delta < 1$ be given.
	Let $\gamma_{\alpha}, \gamma_{\beta} \in \mathbb{R}_{+}$ and  $q_\alpha, q_\beta \in \mathbb{R}$. Let $p>0$. Suppose $B \in \mathcal{B}''_{\alpha, \beta}$ and the $a_I$'s satisfy \eqref{eq:corhafsq1} and \eqref{eq:corhafsq2} as in Theorem \ref{cor:mainI1}. If we further assume that $a_0$ is independent of $N$ and $a_0 \neq 0$, then
	\begin{align}
		&n_{\Haf^2}^{\textnormal{GBS-I}} < \frac{1}{\epsilon^2 \delta}\left(C_D N^p + \frac{C_G C_D N^{3+p}}{a_0^2 \sqrt{\pi}} - 1\right) \label{eq:ccpoly}\\
		&n_{\Haf^2}^{\textnormal{MC}} > \frac{1}{\epsilon^2 \delta} \frac{\frac{e^{\frac{1}{25} - \frac{1}{6}}}{2\sqrt{\pi}} (\zeta N^2)^{2q_\alpha - \frac{1}{2}} C_R^{\zeta N^2} }{C_D N^p \left(a_0^2 + \frac{1}{\sqrt{\pi}} C_G N^3\right)}. \label{eq:ccexp}
	\end{align}
	Moreover, for any $0 < \epsilon, \delta < 1$, $p >0$,  $N$ sufficiently large, and $K \geq \zeta N^2$ with $\zeta >0$, there exists a non-empty and open subset of $\mathcal{P} $ for which the conditions above are satisfied.
\end{theorem}

\begin{remark}
	The second part of the theorem is proved by a direct construction. Particularly, we require that
	$q_\alpha = \nu_\alpha N$ and $q_\beta = \nu_\beta N$ with $c \leq \nu_\alpha \leq \nu_\beta \leq -\frac{1}{4}$ for some constant $c$. Furthermore, we ask that $\gamma_{\alpha} = \tau_\alpha N^2$ and $\gamma_{\beta} = \tau_\beta N^2$ with $\frac{1}{4} < \tau_\alpha$, and $\frac{\tau_\beta^2}{4} < \tau_\alpha \leq \tau_\beta$.
	We must also have $\bmin = \frac{\tau_1}{N}$, $\bmax = \frac{\tau_2}{N}$ with some positive constants $\tau_1 < \tau_2$. Within this constructed subset, we can specify the constants explicitly. First, a sufficiently large $N$ means
	\begin{align*}
		N \geq N_0 = \max(N_C, N_D)
	\end{align*}
	where $N_D$ is the minimal such that
	\begin{align*}
		\tau_1 + \frac{\tau_2-\tau_1}{N_D} < 1,
	\end{align*}
	and $N_C$ is the minimal such that
	\begin{align*}
		N_G^{2q_\beta - \frac{1}{2} -3} \left( 2\pi \right)^{\frac{N_G-1}{2}}
		e^{\frac{N_G}{13}} \frac{(2\tau_2)^{N_G+1}-1}{2\tau_2 -1}\frac{1}{1-\tau_2^{N_G}} \frac{16 \sqrt{\pi}}{(\tau_2 e)^2} < 1.
	\end{align*}
	Furthermore,
	\begin{gather*}
		C_G = \frac{1}{8 \sqrt{\pi}} \left( {\tau_\beta \tau_2 e}\right)^2, \\
		C_D = \left(1 - \left(\tau_1 + \frac{\tau_2-\tau_1}{N_D} \right)^2 \right)^{-\frac{p}{2}}, \\
		C_R = 4 \tau_\alpha \tau_1^2 >1.
	\end{gather*}
\end{remark}

We now prove Theorem \ref{thrm:mainI1}, which requires a few lemmas. 
Lemma \ref{lem:truncate} below is particularly useful for upper bounding the guaranteed sample size for GBS estimators. 

\begin{lemma}
	\label{lem:truncate}
	Let $q \in \mathbb{R}$, then 
	\begin{align}
		\sum_{k = 1}^K m_k \frac{k^q}{(2k)!} z^{2k} \leq G_{q, K, N}(z).
		\label{eq:ggup}
	\end{align}
	If $K < \infty$, formula \eqref{eq:ggup} holds for the entire real line. If $K = \infty$, formula \eqref{eq:ggup} holds for $\vert z \vert < N$.
\end{lemma}

\begin{proof}
	We begin by truncating the sum over $k$ into pieces of length $N$. Let 
	$$F_s = \{ k \mid sN + 1 \leq 2k \leq (s+1)N \}, ~~~~ \text{where } s = 0, 1,2, \dots$$
	We compute the sum over each $F_s$. For $s = 0$, notice that $m_k = 1$ for all $k \in F_0.$ Then,
	\begin{align*}
		\sum_{k \in F_0} m_k \frac{k^q}{(2k)!} z^{2k} = \sum_{k \in F_0}  \frac{k^q}{(2k)!} z^{2k} =  \hfun_{q, \lfloor \frac{N}{2} \rfloor}(z),
	\end{align*}
	where $\hfun_{q, \lfloor \frac{N}{2} \rfloor}$ is defined as in \eqref{eq:cosh}. 
	For $s \geq 1$ and $k \in F_s$, we obtain
	\begin{align*}
		m_k = (s!)^N (s +1)^{r_k},
	\end{align*}
	and 
	\begin{align*}
		2k = Ns + r_k,
	\end{align*}
	where $1 \leq r_k \leq N$. Furthermore, if $q \geq 0$, then
		\begin{align*}
		k^q \leq \left( \frac{(s+1)N}{2}\right)^q \leq \left( \frac{2sN}{2}\right)^q = (sN)^q.
	\end{align*}  
	If $q < 0$, then
	\begin{align*}
		k^q \leq \left( \frac{sN}{2}\right)^q.
	\end{align*} 
	In other words, for any $q \in \mathbb{R}$,
	\begin{align*}
		k^q \leq 2^{-q^{-}}(sN)^q.
	\end{align*}
	We then compute
	\begin{align*}
		\sum_{k \in F_s} m_k\frac{k^q}{(2k)!} z^{2k} & = \sum_{k \in F_s} \frac{ (s!)^N (s +1)^{r_k} }{(Ns + r_k)!} k^q  z^{Ns + r_k} \\
		& \leq z^{Ns} 2^{-q^{-}} (sN)^q \frac{ (s!)^N}{(Ns)!} \sum_{k \in F_s} \frac{(s +1)^{r_k} z^{r_k} }{(Ns + 1)\dots (Ns + r_k)}
	\end{align*}
	Using Stirling's approximation on $s!$ and $(Ns)!$, we get that
	\begin{align*}
		\frac{ (s!)^N}{(Ns)!} &  \leq \frac{\left( \sqrt{2\pi s} \right)^N \left( \frac{s}{e}\right)^{Ns} e^{\frac{N}{12s}}}{\sqrt{2\pi Ns} \left( \frac{Ns}{e}\right)^{Ns} e^{\frac{1}{12Ns+1}}} \\
		& = \frac{1}{\sqrt{N}} \left( \sqrt{2\pi s} \right)^{N-1}  \left(\frac{1}{N}\right)^{Ns} e^{\frac{N}{12s} - \frac{1}{12Ns+1}} \\
		& \leq \frac{1}{\sqrt{N}} \left( \sqrt{2\pi s} \right)^{N-1}  \left(\frac{1}{N}\right)^{Ns} e^{\frac{N}{12}} \\
		& = \frac{1}{\sqrt{N}} \left( \sqrt{2\pi} \right)^{N-1}  e^{\frac{N}{12}} s^{\frac{N-1}{2}} \left(\frac{1}{N}\right)^{Ns}.
	\end{align*}
	Next, we compute
	\begin{align*}
		\frac{(s +1)^{r_k} z^{r_k} }{(Ns + 1)\dots (Ns + r_k)} =  \frac{N^{-r_k}(s +1)^{r_k} z^{r_k} }{(s + \frac{1}{N})\dots (s + \frac{r_k}{N})} \leq \left(\frac{2z}{N}\right)^{r_k},
	\end{align*}
	where the inequality was derived from $\frac{s+1}{\left( s + \frac{i}{N}\right)} \leq 2$ for all $1 \leq i \leq N$ and $s \geq 1$.  
	Together, we get
	\begin{align*}
		\sum_{k \in F_s}  m_k\frac{k^q}{(2k)!} z^{2k} 
		& \leq  2^{\frac{N}{2} - \frac{1}{2} - q^{-}} \pi^{\frac{N}{2} - \frac{1}{2}} N^{q - \frac{1}{2}}  e^{\frac{N}{12}} s^{\frac{N}{2} - \frac{1}{2} + q} \left(\frac{z}{N}\right)^{Ns} \sum_{k \in F_s} \left(\frac{2z}{N}\right)^{r_k}.
	\end{align*}
	We further compute
	\begin{align*}
		\sum_{k \in F_s} \left(\frac{2z}{N}\right)^{r_k} = \sum_{j = 1}^N \left(\frac{2z}{N}\right)^{j} 
		= \multilog_{0, N}\left(\frac{2z}{N}\right).
	\end{align*}
	Therefore, 
	\begin{align*}
		\sum_{k \in F_s} z^{2k} \frac{m_k}{(2k)!} k^q
		= 2^{\frac{N}{2} - \frac{1}{2} - q^{-}} \pi^{\frac{N}{2} - \frac{1}{2}} N^{q - \frac{1}{2}} 
		e^{\frac{N}{12}} \multilog_{0, N} \left(\frac{2z}{N}
		\right)  s^{\frac{N}{2} - \frac{1}{2} + q} \left(\frac{z}{N}\right)^{Ns}.
	\end{align*}
	We collect the sum over $F_s$ with $s \geq 0$ and derive
	if $K = \infty$, 
	\begin{align*}
		& \sum_{k = 1}^\infty  m_k\frac{k^q}{(2k)!} z^{2k}  \\
		= & \sum_{k \in F_0}   m_k\frac{k^q}{(2k)!} z^{2k}  + \sum_{s \geq 1} \sum_{k \in F_s}   m_k\frac{k^q}{(2k)!} z^{2k}  \\
		= & ~ \hfun_{q, \lfloor \frac{N}{2} \rfloor}(z) +  2^{\frac{N}{2} - \frac{1}{2} - q^{-}} \pi^{\frac{N}{2} - \frac{1}{2}} N^{q - \frac{1}{2}}
		e^{\frac{N}{12}} \multilog_{0, N} \left(\frac{2z}{N}
		\right)
		\multilog_{\frac{1}{2} - \frac{N}{2}- q} \left(\frac{z^N}{N^N}
		\right) \\
		= &~ G_{q,K, N}(z)
	\end{align*}
	provided that $\vert z \vert < N$. 
	When $K < \infty$, we can truncate the sum according to $s_K$. If $s_K = 0$, we simply get 
	\begin{align*}
		\sum_{k = 1}^K  \frac{m_k}{(2k)!} z^{2k} \leq \hfun_{q, \lfloor \frac{N}{2} \rfloor}(z) = G_{q, K, N}(z), 
	\end{align*}
	since $\multilog_{s, 0}(z) = 0$ for any $s$. 
	If $s_K \geq 1$, we have
	\begin{align*}
		& \sum_{k = 1}^K  \frac{m_k}{(2k)!} z^{2k} \\
		\leq &~ \hfun_{q, \lfloor \frac{N}{2} \rfloor}(z) +  2^{\frac{N}{2} - \frac{1}{2} - q^{-}} \pi^{\frac{N}{2} - \frac{1}{2}}  N^{q - \frac{1}{2}}
		e^{\frac{N}{12}} \multilog_{0, N} \left(\frac{2z}{N}
		\right)
		\multilog_{\frac{1}{2} - \frac{N}{2}- q, s_K} \left(\frac{z^N}{N^N}
		\right) \\
		= &~ G_{q, K, N}(z).
	\end{align*}
\end{proof}

The lemma below gives a convenient formula for getting the exponential reduction. We state the result in the general case where $\phi$ is an arbitrary matrix function, and $Q^{\textnormal{GBS}}_{\phi}$ can be $Q^{\textnormal{GBS-I}}_{\phi}$ or $Q^{\textnormal{GBS-P}}_{\phi}$, $n^{\textnormal{GBS}}_{\phi}$ can be $n^{\textnormal{GBS-I}}_{\phi}$ or $n^{\textnormal{GBS-P}}_{\phi}$. 

\begin{lemma}
	\label{lem:equivgoal}
	Let $D_1, D_2 >0$. Let
	\begin{align*}
		\tilde{D}_1 = \frac{D_1}{\epsilon^2 \delta}, ~~~~~~\tilde{D}_2 = D_2(\epsilon^2 \delta + 1).
	\end{align*}
	Let $U^{\textnormal{GBS}}_{\phi}$ denote an upper bound of $\frac{Q^{\textnormal{GBS}}_{\phi}}{ \mu_{\phi}^2}$ and  $L^{\textnormal{MC}}_{\phi}$ denote a lower bound of $\frac{Q_{\textnormal{MC}}}{\mu_{\phi}^2}$.
	If
	\begin{align}
		\tilde{D}_2 \exp(-\tilde{D}_1) \exp(\tilde{D}_1 U^{\textnormal{GBS}}_{\phi}) < L^{\textnormal{MC}}_{\phi},
		\label{eq:consistency}
	\end{align}
	then
	\begin{equation}
		\label{eq:consistencygoal}
		D_2 \exp(D_1 n^\textnormal{GBS}_{\phi} ) < n^\textnormal{MC}_{\phi}.
	\end{equation}
\end{lemma}

\begin{remark}
	For future reference, we note that \eqref{eq:consistency} is equivalent to
	\begin{align}
		\tilde{D}_1 U^{\textnormal{GBS}}_{\phi} < \ln{L^{\textnormal{MC}}_{\phi}} - \ln(\tilde{D}_2) +  \tilde{D}_1
		\label{eq:consistency2}
	\end{align}
	by taking the logarithm on both sides. 
\end{remark}

\begin{proof}
	Since $n^{\textnormal{GBS}}_{\phi}$ and $n^\textnormal{MC}_{\phi}$ are invariant under uniform scaling of the $a_I$'s, we assume without loss of generality that $\mu_\phi = 1$. 
	The proof of this statement is simply to notice that \eqref{eq:consistencygoal} is implied by
	\begin{align*}
		D_2 \exp(\frac{D_1}{\epsilon^2 \delta}U^{\textnormal{GBS}}_{\phi} - \frac{D_1}{\epsilon^2 \delta} ) < \frac{L^{\textnormal{MC}}_{\phi}}{\epsilon^2 \delta }- \frac{1}{\epsilon^2 \delta}.
	\end{align*}
	It suffices to show that
	\begin{align*}
		\epsilon^2 \delta \exp(-\frac{D_1}{\epsilon^2 \delta}) D_2 \exp(\frac{D_1}{\epsilon^2 \delta} U^{\textnormal{GBS}}_{\phi} ) + D_2 < L^{\textnormal{MC}}_{\phi},
	\end{align*}
	which is equivalent to 
	\begin{align}
		\label{eq:equiproof}
		\left(\epsilon^2 \delta \exp(-\frac{D_1}{\epsilon^2 \delta}) +  \exp(-\frac{D_1}{\epsilon^2 \delta} U^{\textnormal{GBS}}_{\phi}) \right) D_2 \exp(\frac{D_1}{\epsilon^2 \delta}U^{\textnormal{GBS}}_{\phi}) < L^{\textnormal{MC}}_{\phi}.
	\end{align}
	Further, we note that the following condition is enough to show \eqref{eq:equiproof}
	\begin{gather}
		\label{eq:equiproof1}
		(\epsilon^2 \delta  +1) \exp(-\frac{D_1}{\epsilon^2 \delta}) D_2 \exp(\frac{D_1}{\epsilon^2 \delta } U^{\textnormal{GBS}}_{\phi} ) < L^{\textnormal{MC}}_{\phi},
	\end{gather}
	since $ U^{\textnormal{GBS}}_{\phi} \geq 1$ and
	\begin{align*}
		\exp(-\frac{D_1}{\epsilon^2 \delta} U^{\textnormal{GBS}}_{\phi}) \leq \exp(-\frac{D_1}{\epsilon^2 \delta}).
	\end{align*}
	It is also clear that \eqref{eq:equiproof1} are equivalent to \eqref{eq:consistency}.
\end{proof}

We now show that the two guaranteed sample sizes are well-defined under the conditions specified by Theorem \ref{thrm:mainI1}. We also provide estimates for $Q_{\Haf^2}^{\textnormal{GBS-I}} $ and $Q_{\Haf^2}^{\textnormal{MC}}$. 

\begin{lemma}
	\label{lem:gbsiestimategsize}
	Given the conditions in Theorem \ref{thrm:mainI1}, $n^{\textnormal{GBS-I}}_{\Haf^2}$ is well-defined and 
	\begin{align*}
		Q_{\Haf^2}^{\textnormal{GBS-I}} \leq \frac{ c_1^{-2} \mu_{\Haf^2}^2}{d}  \left( 1 + \frac{1}{\sqrt{\pi}} G_{2 q_{\beta}-\frac{1}{2}, K, N}(\gamma_\beta \bmax) \right).
	\end{align*}
\end{lemma}

\begin{proof}
When $K$ is finite, $n^{\textnormal{GBS-I}}_{\Haf^2}$ being well-defined holds trivially. When $K$ is infinite, it follows from  Lemma \ref{lem:swap1} that the conditions \eqref{eq:thrmI1-1} and \eqref{eq:thrmI1-4} imply \eqref{eq:absconvsq}. We now compute that 
\begin{align*}
	Q_{\Haf^2}^{\textnormal{GBS-I}} &= \frac{1}{d} \sum_{k = 0}^K \sumck a_I^2 I! \Haf(B_I)^2  \\
	& = \frac{1}{d} a_0^2 + \frac{1}{d} \sum_{k = 1}^K \sumck a_I^2 I! \Haf(B_I)^2  \\
	& = \frac{c_1^{-2} \mu_{\Haf^2}^2}{d} \left( 1+ \sum_{k = 1}^K \left( \gamma_\beta \bmax \right)^{2k} \frac{m_k  k^{2 q_{\beta}}}{2^{2k} (k!)^2} \right) \\
	& \hspace{1in} \textnormal{(use \eqref{eq:thrmI1-2} and Lemma \ref{lem:hafapprox0}}) \\
	& \leq \frac{c_1^{-2} \mu_{\Haf^2}^2}{d} \left( 1 + \frac{1}{\sqrt{\pi}} \sum_{k = 1}^K \left( \gamma_\beta \bmax \right)^{2k} \frac{m_k k^{2 q_{\beta} -\frac{1}{2}}}{(2k)!} \right) \hspace{0.1in} \text{ (use Lemma \ref{lem:usefulstirling1})} \\
	& \leq \frac{ c_1^{-2} \mu_{\Haf^2}^2}{d}  \left( 1 + \frac{1}{\sqrt{\pi}} G_{2 q_{\beta}-\frac{1}{2}, K, N}(\gamma_\beta \bmax) \right) \hspace{0.1in} \text{ (use Lemma \ref{lem:truncate}).}
\end{align*}
We further use \eqref{eq:thrmI1-5} to ensure that $ G_{2 q_{\beta}, K, N}$ is finite, and thus \eqref{eq:QIfinite} holds. Therefore $n^{\textnormal{GBS-I}}_{\Haf^2}$ is well-defined.
\end{proof}

\begin{lemma}
	\label{lem:mcestimategsize}
	Given the conditions in Theorem \ref{thrm:mainI1}, $n^{\textnormal{MC}}_{\Haf^2}$ is well-defined and 
	\begin{align*}
		Q_{\Haf^2}^{\textnormal{MC}} \geq  {c_2^{-2} \mu_{\Haf^2}^2}  \left( 1 + e^{\frac{1}{25} - \frac{1}{6}} R_{q_\alpha, K} (4 \gamma_\alpha \bmin^{2}) \right).
	\end{align*}
\end{lemma}

\begin{proof}
	When $K$ is finite, $n^{\textnormal{MC}}_{\Haf^2}$ being well-defined holds trivially. When $K$ is infinite, it follows from  Lemma \ref{lem:swap1} and Theorem \ref{thrm:mcpp1} that the conditions \eqref{eq:thrmI1-1} and \eqref{eq:thrmI1-4} imply \eqref{eq:absconvsq} and \eqref{eq:swapuse}.
	Therefore, $n^{\textnormal{MC}}_{\Haf^2}$ is well-defined.
	
	To give a lower bound of $Q_{\Haf^2}^{\textnormal{MC}}$, we use Lemma \ref{lem:hafapprox0} to compute that
	\begin{align*}
		Q_{\Haf^2}^{\textnormal{MC}} 
		&= \sumdoublefinite a_I a_J \Haf(B_{I + J})^2\\
		& \geq a_0^2 + \sum_{l = 1}^K \sum_{k_1 + k_2 = l}  \sum_{\substack{\vert I \vert = 2k_1 \\ \vert J \vert = 2k_2}} a_I a_J \bmin^{2l} \frac{(2l)!^2}{2^{2l} l!^2}
	\end{align*}
	Using \eqref{eq:thrmI1-1}, we have for all $l \geq 1$ 
	\begin{align*}
		\sum_{k_1 + k_2 = l}  \sum_{\substack{\vert I \vert = 2k_1 \\ \vert J \vert = 2k_2}} a_I a_J 
		& \geq \sum_{\substack{k_1 + k_2 = l \\ k_1 \neq 0 \text{ and} \\ k_2 \neq 0}} \sum_{\substack{\vert I \vert = 2k_1 \\ \vert J \vert = 2k_2}} a_I a_J +  2 a_0 \sum_{\substack{\vert I \vert = 2l}} a_I\\
		& \geq \sum_{\substack{k_1 + k_2 = l \\ k_1 \neq 0 \text{ and}\\ k_2 \neq 0}} c_2^{-2} \mu_{\Haf^2}^2 \frac{k_1^{q_\alpha} k_2^{q_\alpha}}{(2k_1)!(2k_2)!} \gamma_\alpha^l + 2  c_2^{-2} \mu_{\Haf^2}^2 \frac{l^{q_\alpha}}{(2l)!} \gamma_\alpha^l \\
		& = c_2^{-2} \mu_{\Haf^2}^2 \gamma_\alpha^l  \left( \sum_{\substack{k_1 + k_2 = l \\ k_1 \neq 0 \text{ and}\\ k_2 \neq 0}} \frac{k_1^{q_\alpha} k_2^{q_\alpha}}{(2k_1)!(2k_2)!} + 2 \frac{l^{q_\alpha}}{(2l)!} \right) \\
		& \geq  c_2^{-2} \mu_{\Haf^2}^2 \gamma_\alpha^l S_l \sum_{k_1 + k_2 = l} \frac{1}{(2k_1)!(2k_2)!},
	\end{align*}
	since
	\begin{align*}
		\sum_{i=1}^n a_i b_i \geq \left(\sum_{i=1}^n b_i \right)\min_i a_i, ~~~~ \text{for all } a_i, b_i \geq 0.
	\end{align*}
	Here, 
	\begin{align}
		S_l = \min\left( \min_{\substack{k_1 + k_2 = l \\ k_1 \neq 0 \text{ and}\\ k_2 \neq 0}}  (k_1 k_2)^{q_\alpha}, l^{q_\alpha}\right).
	\end{align}
	We give lower bounds of $S_l$ in three cases:
	\begin{enumerate}
		\item When $q_\alpha = 0$, it is clear that $S_l = 1$.
		\item When $q_\alpha >0$, since for all $l \geq 2$,
		\begin{align*}
			\min_{\substack{k_1 + k_2 = l \\ k_1 \neq 0 \text{ and}\\ k_2 \neq 0}}  (k_1 k_2)^{q_\alpha} = (l-1)^{q_\alpha}
		\end{align*}
		and $ l^{q_\alpha} \geq (l-1)^{q_\alpha}$, we have 
		\begin{align*}
			S_l = (l-1)^{q_\alpha} \geq \left( \frac{l}{2} \right)^{q_\alpha}.
		\end{align*}
		This lower bound holds for $l = 1$, since the set of $k_1, k_2$ such that $k_1 + k_2 = 1$ with $k_1 \neq 0$ and $k_2 \neq 0$ is empty, it follows that 
		$$S_1 = 1 \geq  \left( \frac{1}{2} \right)^{q_\alpha}. $$
		
		\item When $q_\alpha < 0$, since for all $l \geq 2$,
		\begin{align*}
			\min_{\substack{k_1 + k_2 = l \\ k_1 \neq 0 \text{ and}\\ k_2 \neq 0}}  (k_1 k_2)^{q_\alpha} \geq \left(\frac{l}{2}\right)^{2q_\alpha} \geq l^{2q_\alpha}
		\end{align*}
		and $ l^{q_\alpha} \geq l^{2q_\alpha}$, we have
		\begin{align*}
			S_l \geq l^{2q_\alpha}.
		\end{align*}
		For $l = 1$, this bounds holds trivially, since $S_1 = 1 \geq 1^{2q_\alpha}.$
	\end{enumerate}
	To summarize, for any $l \geq 1$,
	\begin{align*}
		& S_l \geq \left( \frac{l}{2} \right)^{q_\alpha}  ~~~~\text{if } q_\alpha \geq 0\\
		& S_l \geq l^{2q_\alpha} ~~~~~~~~~\text{if } q_\alpha < 0.
	\end{align*}
	
	We further note that using the following well-known combinatorial fact
	\begin{align*}
		\sum_{k_1 + k_2 = l} \frac{(2l)!}{(2k_1)!(2k_2)!} = 2^{2l-1},
	\end{align*}
	we obtain
	\begin{align*}
		Q_{\Haf^2}^{\textnormal{MC}} & \geq {c_2^{-2} \mu_{\Haf^2}^2} \left( 1 + \frac{1}{2} \sum_{l =1}^K S_l 2^{2l}\bmin^{2l} \gamma_\alpha^{l} \frac{(2l)!}{2^{2l} l!^2} \right) \\
		& \geq {c_2^{-2} \mu_{\Haf^2}^2} \left( 1 + e^{\frac{1}{25} - \frac{1}{6}}  \frac{1}{2\sqrt{\pi}} \sum_{l =1}^K \frac{1}{\sqrt{l}} S_l (4\gamma_\alpha \bmin^2)^{l} \right). \\
		& \hspace{1.8in} (\text{Use Lemma } \ref{lem:usefulstirling1})
	\end{align*}
	If $ q_\alpha \geq 0$, then
	\begin{align*}
		Q_{\Haf^2}^{\textnormal{MC}}  &\geq  {c_2^{-2} \mu_{\Haf^2}^2}  \left( 1 + e^{\frac{1}{25} - \frac{1}{6}}  2^{-q_\alpha} \frac{1}{ 2\sqrt{\pi}} \multilog_{\frac{1}{2} - q_\alpha, K} (4 \gamma_\alpha \bmin^{2}) \right) \\
		& =   {c_2^{-2} \mu_{\Haf^2}^2}  \left( 1 + e^{\frac{1}{25} - \frac{1}{6}}  R_{q_\alpha, K} (4 \gamma_\alpha \bmin^{2}) \right).
	\end{align*}
	If $ q_\alpha < 0$, then
	\begin{align*}
		Q_{\Haf^2}^{\textnormal{MC}}  & \geq  {c_2^{-2} \mu_{\Haf^2}^2}  \left( 1 + e^{\frac{1}{25} - \frac{1}{6}}  \frac{1}{ 2 \sqrt{\pi}} \multilog_{\frac{1}{2} - 2 q_\alpha, K} (4 \gamma_\alpha \bmin^{2}) \right) \\
		& =   {c_2^{-2} \mu_{\Haf^2}^2}  \left( 1 +e^{\frac{1}{25} - \frac{1}{6}}  R_{q_\alpha, K} (4 \gamma_\alpha \bmin^{2}) \right).
	\end{align*}
\end{proof}

We next prove that the conditions in Theorem \ref{thrm:mainI1} imply that $n^{\textnormal{GBS-I}}_{\Haf^2}$ is exponentially less than $n^{\textnormal{MC}}_{\Haf^2}$. 

\begin{lemma}
	\label{lem:firstpart}
	If the conditions in Theorem \ref{thrm:mainI1} are satisfied, then
	\begin{align}
		\label{eq:expspeed2}
		s_2\exp(s_1 \, n_{\Haf^2}^{\textnormal{GBS-I}} ) < n_{\Haf^2}^{\textnormal{MC}} < \infty.
	\end{align}
\end{lemma}

\begin{proof}
	Plugging the bounds for $Q_{\Haf^2}^{\textnormal{GBS-I}}$ and $Q_{\Haf^2}^{\textnormal{MC}}$  from Lemma \ref{lem:mcestimategsize} and \ref{lem:gbsiestimategsize} into \eqref{eq:consistency2} and setting $D_1 = s_1$ and $D_2 = s_2$, we get 
	\begin{align*}
		&\frac{s_1}{\epsilon^2 \delta} \frac{ c_1^{-2}}{d}  \left( 1 + \frac{1}{\sqrt{\pi}} G_{2 q_{\beta}, K, N}(\gamma_\beta \bmax) \right)\\
		& < \ln\left( 1 + e^{\frac{1}{25} - \frac{1}{6}}  R_{q_\alpha, K} (4 \gamma_\alpha \bmin^{2}) \right)+ 2 \ln\left(c_2^{-1} \right) - \ln(\epsilon^2 \delta + 1) + \frac{s_1}{\epsilon^2 \delta} - \ln(s_2).
	\end{align*}
	This is the same as \eqref{eq:thrmI1-3}. 
	Using Lemma \ref{lem:equivgoal}, 
	it is straightforward that \eqref{eq:expspeed2} holds.
\end{proof}

We now give a few lemmas that will be useful for proving the second half of Theorem \ref{thrm:mainI1}.

\begin{lemma}
	\label{lem:bbexist}
	For any $0 < b_1 < b_2 < \frac{1}{N}$, there exists $B \in \mathcal{B}$ such that $\bmin = b_1$ and $\bmax = b_2$. 
\end{lemma}

\begin{proof}
	Consider $B$ of the following form
	\begin{align*}
		B &= \begin{bmatrix}
			b_2 & b_1 & \dots & b_1\\
			b_1 & b_2 & \dots & b_1 \\
			\vdots & & & \vdots \\
			b_1 & b_1 & \dots & b_2
		\end{bmatrix} = (b_2-b_1) \mathbb{I} + b_1 \mathds{1}
	\end{align*}
	where $\mathbb{I}$ is the identity matrix and $\mathds{1}$ is the matrix with all entries equal to 1. The eigenvalues of $b_1 \mathds{1}$ are given by $Nb_1, 0, \dots, 0$. Therefore, the eigenvalues of $B$ are computed by adding $(b_2-b_1)$ to each of the eigenvalue of $b_1 \mathds{1}$, that is, 
	$$b_2 - b_1 + Nb_1, b_2-b_1, \dots, b_2-b_1.$$ 
	As long as $b_2 \neq b_1$, the eigenvalues are all positive. Since
	\begin{align}
		b_2 - b_1 + Nb_1 = b_2 + (N-1)b_1 < Nb_2 < 1,
	\end{align}
	all eigenvalues are bounded above by 1.
	We thus have found a matrix $B \in \mathcal{B}$. 
\end{proof}

\begin{lemma}
	\label{lem:kinfBexsits}
	When $K = \infty$, if $q_\alpha \leq q_\beta$ and $\frac{N^2}{4} < \gamma_\alpha = \gamma_\beta < 4 N^2$, then $\mathcal{B}_{\alpha, \beta}$ is non-empty.
\end{lemma}

\begin{proof}
	Under the specified conditions, \eqref{eq:cgamma} clearly holds. 
	We now choose $0 < \bmin < \bmax < \frac{1}{N}$ so that \eqref{eq:thrmI1-3}, \eqref{eq:thrmI1-4} and \eqref{eq:thrmI1-5} hold.
	Let
	$$\bmax = \frac{\tau}{N}, ~~~~ \bmin = \xi \bmax$$
	with $ 0 < \tau < 1$ and $0 < \xi < 1$. 
	Notice that if
	\begin{align*}
		\frac{1}{4} < \tau < \frac{N}{2\sqrt{\gamma_\beta}} < 1,
	\end{align*}
	then
	\begin{align*}
	4 \gamma_\beta \bmax^2 = 4 \gamma_\beta \frac{\tau^2}{N^2} < 1, \\
		\gamma_\beta \bmax < \frac{1}{4 \bmax} = \frac{N}{4\tau} < N,
	\end{align*}
	and therefore \eqref{eq:thrmI1-4} and \eqref{eq:thrmI1-5} hold. Furthermore, under these conditions, $c_2$ are bounded above by a constant. Specifically
	\begin{align*}
		&c_2 < 1 + \frac{1}{\sqrt{\pi}} \multilog_{\frac{1}{2}-q_\beta, K}\left(\frac{1}{4}\right).
	\end{align*}
	
	Now we would like to take $\xi \tau$ sufficiently close to $\frac{N}{2\sqrt{\gamma_\alpha}}$ so that
	\begin{align*}
		4 \gamma_\alpha \bmin^2 = 4 \gamma_\alpha \xi^2 \bmax^2 = 4 \gamma_\alpha \xi^2 \frac{\tau^2}{N^2}
	\end{align*}
	becomes as close to 1 as we like and consequently
	\begin{gather*}
		\lim_{\xi \tau \rightarrow \frac{N}{2\sqrt{\gamma_\alpha}}} \ln\left( 1 + e^{\frac{1}{25} - \frac{1}{6}} R_{q_\alpha, K} (4 \gamma_\alpha \bmin^{2}) \right) = \infty.
	\end{gather*}
	This is only possible if $\gamma_\alpha = \gamma_\beta$. 
	At the same time, we would like $\tau$ to stay at least $\frac{1}{2} \left( \frac{N}{2 \sqrt{\gamma_\alpha}} - \frac{1}{4} \right)$ away from $\frac{1}{4}$. In other words, 
	$$ \frac{1}{4} + \frac{1}{2} \left( \frac{N}{2 \sqrt{\gamma_\alpha}} - \frac{1}{4} \right) < \tau < \frac{N}{2 \sqrt{\gamma_\alpha}}.$$
	Therefore, 
	\begin{align*}
		& G_{2q_\beta, K, N}(\gamma_\beta \bmax) < G_{2q_\beta, K, N} \left(\tfrac{N}{4 \tau}\right) < G_{2q_\beta, K, N}\left(N\left(\tfrac{N}{\sqrt{2\gamma_{\alpha}}} + \tfrac{1}{2}\right)^{-1} \right). 
	\end{align*}
	is bounded above by a constant. 
	Now we can freely choose $\tau$ and $\xi$ to ensure the RHS of \eqref{eq:thrmI1-3} eventually exceeds the LHS. Finally, it follows from Lemma \ref{lem:bbexist} that there exists $B$ in $\mathcal{B}$ such that $\bmax = \frac{\tau}{N}$ and $\bmin = \xi \bmax$ with the chosen $\tau$ and $\xi$. 
\end{proof}

\begin{lemma}
	\label{lem:kfBexsits}
	When $K < \infty$ is sufficiently large, if $q_\alpha \leq q_\beta$, $\frac{N^2}{4} < \gamma_\alpha \leq \gamma_\beta < N^2$, then $\mathcal{B}_{\alpha, \beta}$ is non-empty.
\end{lemma}

\begin{proof}
	The proof is similar to the case when $K$ is infinite. We will only highlight the major differences. As before, we write $\bmax = \frac{\tau}{N}, \bmin = \xi \bmax$ with $ 0 < \tau < 1$ and $0 < \xi < 1$. 
	Now we would like to take $\xi \tau$ sufficiently close to $1$ so that
	\begin{align*}
		4 \gamma_\alpha \bmin^2 = \frac{4 \gamma_\alpha }{N^2} \xi^2 \tau^2 > 1,
	\end{align*}
	and consequently
	\begin{gather*}
		\lim_{K \rightarrow \infty} \ln\left( 1 + e^{\frac{1}{25} - \frac{1}{6}} R_{q_\alpha, K} (4 \gamma_\alpha \bmin^{2}) \right) = \infty.
	\end{gather*}
	This is only possible if $\gamma_{\alpha} > \frac{N^2}{4}$. 
	We further notice that since $\gamma_{\beta} < N^2$, we have
	$$\tau < 1 < \frac{N^2}{\gamma_\beta}$$
	and
	\begin{align*}
		\gamma_\beta \bmax = \gamma_\beta  \frac{\tau}{N} < \frac{\gamma_{\beta}}{N} < N.
	\end{align*}
	Then,
	\begin{align*}
		&G_{q_\beta, K, N}(\gamma_\beta \bmax) < G_{q_\beta, K, N}(\frac{\gamma_{\beta}}{N})
	\end{align*}
	is bounded by a constant as $K$ goes to $\infty$. Hence, for a large enough $K$, we obtain \eqref{eq:thrmI1-3}. The rest of the proof is omitted. 
\end{proof}

We now turn to prove that for each $B \in \mathcal{B}_{\alpha, \beta}$, there exists open and non-empty subset of $\mathcal{A}_B $ for which \eqref{eq:thrmI1-1} and \eqref{eq:thrmI1-2} hold. 
To this end, we define 
$\mu_{\Haf^2}: \mathcal{A}_B \rightarrow \mathbb{R}_+$ given by
\begin{align}
	\mu_{\Haf^2}(a_I) = \sum_{k = 0}^K \sum_{\vert I \vert =2k} a_I \Haf(B_I)^2. \label{eq:contmapdefn}
\end{align}
The well-defined map is continuous (see Lemma \ref{lem:mucont} below). Let $\mathcal{A}_B^1 = \mu_{\Haf^2}^{-1}(1)$. We define $\pi: \mathcal{A}_B \rightarrow \mathcal{A}_B^1$ by
\begin{align*}
	\pi(a_I) = \left( \frac{a_I}{\mu_{\Haf^2}(a_I)}\right),
\end{align*}
and clearly $\pi$ is also continuous.

\begin{lemma}
	\label{lem:mucont}
	The map $\mu_{\Haf^2}$ defined as in \eqref{eq:contmapdefn}
	is continuous.
\end{lemma}

\begin{proof}
	Let $\epsilon >0$ and $(a_I^0) \in \mathcal{A}_B$. Consider
	\begin{align*}
		\mu^{-1}_{\Haf^2}(\mu_{\Haf^2}(a_I^0) - \epsilon, \mu_{\Haf^2}(a_I^0) + \epsilon) \subseteq \mathcal{A}_B.
	\end{align*}
	We choose $\epsilon_k > 0 $ such that 
	\begin{align}
		\label{eq:mucont1}
		\sum_{k = 0}^K \epsilon_k \bmax^{2k} \frac{(2k)!^2}{2^{2k} k!^2} < \epsilon
	\end{align}
	Particularly, if we choose $b$ such that $b \bmax = \frac{1}{2}$ and let
	\begin{align*}
		\epsilon_k = \frac{1}{2}\epsilon b^{2k} \frac{2^{2k} k!^2}{(2k)!^2}, ~~~~ k = 0, 1, \dots, K
	\end{align*}
	Then, \eqref{eq:mucont1} is satisfied. Let
	\begin{align*}
		V_k = \left\{ (a_I) \in \mathbb{R}^{\mathcal{I}_k} \mid \vert a_I^0 - a_I \vert < \epsilon_k \right\}.
	\end{align*}
	Then, 
	\begin{align*}
		\prod_{k=0}^K V_k \subseteq \mu^{-1}_{\Haf^2}(\mu_{\Haf^2}(a_I^0) - \epsilon, \mu_{\Haf^2}(a_I^0) + \epsilon)
	\end{align*}
\end{proof}

For each $B \in \mathcal{B}_{\alpha, \beta}$, we further introduce the following subsets
\begin{align*}
	&U_{B, \alpha, \beta} = \left\{ (a_I) \in \mathcal{A}_B \mid \exists\, a_I \neq 0,\right. \\
	&\hspace{1.4in} c_2^{-1}\mu_{\Haf^2} < a_0 < c_1^{-1}\mu_{\Haf^2},  \\
	&\hspace{1.4in} c_2^{-1}\mu_{\Haf^2} \frac{ k^{q_\alpha} \gamma_\alpha^{k}}{(2k)!}  < \sum_{\vert I \vert = 2k} a_I < c_1^{-1}\mu_{\Haf^2} \frac{k^{q_\beta} \gamma_\beta^{k}}{(2k)!}, ~~ k = 1, 2, \dots, K, \\
	&\hspace{1.3in} \sumck a^2_I I! < c_1^{-2}\mu_{\Haf^2}\frac{k^{2q_\beta} \gamma_\beta^{2k}}{(2k)!^2} m_k, ~~ k = 1, 2, \dots, K \}.
\end{align*}
and
\begin{align*}
	&U^1_{B, \alpha, \beta} = \left\{ (a_I) \in \mathcal{A}^1_B \mid c_2^{-1} < a_0 < c_1^{-1}, \right. \\
	&\hspace{1.5in} c_2^{-1} \frac{ k^{q_\alpha} \gamma_\alpha^{k}}{(2k)!}  < \sum_{\vert I \vert = 2k} a_I < c_1^{-1} \frac{k^{q_\beta} \gamma_\beta^{k}}{(2k)!} ,  ~~ k = 1, 2, \dots, K, \\
	&\hspace{1.4in} \sumck a^2_I I! < c_1^{-2} \frac{k^{2q_\beta} \gamma_\beta^{2k}}{(2k)!^2} m_k,  ~~ k = 1, 2, \dots, K \}.
\end{align*}
Clearly, $\pi^{-1}(U^1_{B, \alpha, \beta} ) = U_{B, \alpha, \beta}$.

\begin{lemma}
	\label{lem:secondpart}
	For any $B \in 	\mathcal{B}_{\alpha, \beta}$, $U_{B, \alpha, \beta} $ is open in $\mathcal{A}_B $.
\end{lemma}

\begin{proof}
	We begin by showing that $U^1_{B, \alpha, \beta}$ is open in $\mathcal{A}^1_B$. Using the subspace topology criterion, we want to find an open set $V \subseteq \mathcal{A}$ such that  $U^1_{B, \alpha, \beta} = V \cap \mathcal{A}^1_B$. Let
	\begin{align*}
		&V = \left\{ (a_I) \in \mathcal{A} \mid c_2^{-1} < a_0 < c_1^{-1}, \right. \\
		&\hspace{1in} c_2^{-1} \frac{ k^{q_\alpha} \gamma_\alpha^{k}}{(2k)!}  < \sum_{\vert I \vert = 2k} a_I < c_1^{-1} \frac{k^{q_\beta} \gamma_\beta^{k}}{(2k)!} ,  ~~ k = 1, 2, \dots \\
		&\hspace{0.9in} \sumck a^2_I I! < c_1^{-2} \frac{k^{2q_\beta} \gamma_\beta^{2k}}{(2k)!^2} m_k,  ~~ k = 1, 2, \dots \}.
	\end{align*}
	Clearly, $A$ is open in $\mathcal{A}$ and $U^1_{B, \alpha, \beta} = V \cap \mathcal{A}^1_B$.
	Since $\pi$ is continuous, we obtain that $\pi^{-1}(U^1_{B, \alpha, \beta})$ is open in $\mathcal{A}_B $.
\end{proof}

\begin{lemma}
	\label{lem:secondpart2}
	For any $B \in 	\mathcal{B}_{\alpha, \beta}$, $U_{B, \alpha, \beta} $ is non-empty.
\end{lemma}

\begin{proof}
	It suffices to prove that $U^1_{B, \alpha, \beta}$ is non-empty.
	Since $c_1 < c_2$, we can choose $t$ such that 
	\begin{align*}
	c_2^{-1} < t < c_1^{-1}
	\end{align*}
	First, we set $a_0 = t$. Then, for $\vert I \vert = 2k$ with $k \geq 1$, we set all $a_I = 0$ except for one $I$ such that $I! = m_k$ to be $a_I = \frac{t}{(2k)!} \gamma_{\alpha}^k  k^{q_\alpha}$.
	Therefore, 
	\begin{align*}
	&\sum_{I! = m_k} a_I = \frac{t}{(2k)!} \gamma_{\alpha}^k  k^{q_\alpha} \\ 
	& \sum_{I! = m_k} a_I^2 I! 
	= \frac{t^2}{(2k)!^2} \gamma_{\alpha}^{2k}  k^{2q_\alpha} m_k.
	\end{align*}
	This choice clearly satisfies the conditions outlined in $U^1_{B, \alpha, \beta}$. It is therefore left to check that the chosen $a_I$'s are in $\mathcal{A}^1_B$. Since
	\begin{align*}
		\sum_{\vert I \vert = 2k} a_I = \frac{t}{(2k)!} \gamma_{\alpha}^k k^{q_\alpha} \leq \frac{t}{(2k)!} \gamma_{\beta}^k k^{q_\beta},
	\end{align*}
	we get that
	\begin{align*}
	\mu_{\Haf^2}(a_I)
	&\leq t \left( 1+  \sum_{k=1}^K \frac{\gamma_\beta^k k^{q_\beta}}{(2k)!} \bmax^{2k} \frac{(2k)!^2}{2^{2k} k!^2} \right) \\
	& \leq  t \left( 1 + \sum_{k=1}^K \gamma_\beta^k k^{q_\beta} \bmax^{2k} \frac{1}{\sqrt{\pi k}} \right) ~~~~~ \text{(use Lemma \ref{lem:usefulstirling1})}\\
	& = t \left( 1 + \frac{1}{\sqrt{\pi}} \multilog_{\frac{1}{2}-q_\beta, K}\left(\gamma_\beta \bmax^{2} \right) \right) \\
	& = t c_2. 
	\end{align*}
	Similarly, we compute that
	\begin{align*}
		\mu_{\Haf^2}(a_I)  & \geq t \left( 1 +  \sum_{k=1}^K \frac{\gamma_\alpha^k  k^{q_\alpha}}{(2k)!} \bmin^{2k} \frac{(2k)!^2}{2^{2k} k!^2} \right) \\
		& \geq  t \left( 1 + \sum_{k=1}^K \gamma_\alpha^k  k^{q_\alpha} \bmin^{2k} \frac{1}{\sqrt{\pi k}} e^{\frac{1}{25} - \frac{1}{6}} \right) \\
		& = t \left( 1 + \frac{1}{\sqrt{\pi}} e^{\frac{1}{25} - \frac{1}{6}} \multilog_{\frac{1}{2} - q_\alpha, K}\left(\gamma_\alpha \bmin^{2} \right) \right) \\
		& = t c_1. 
	\end{align*}
Together, we have
\begin{align}
	t c_1 \leq \mu_{\Haf^2}(a_I) \leq t c_2,
	\label{eq:usefulrecent}
\end{align}
where $tc_1 < 1$ and $tc_2 > 1$.
If $\mu_{\Haf^2}(a_I) < 1$, then by moving $t$ closer to $c_2^{-1}$, we will, by \eqref{eq:usefulrecent}, at some point reach $\mu_{\Haf^2}(a_I) = 1$. Conversely, if $\mu_{\Haf^2}(a_I) > 1$, then by moving $s$ closer to $c_1^{-1}$, we will also, by \eqref{eq:usefulrecent}, at some point reach $\mu_{\Haf^2}(a_I) = 1$. This completes the proof.
\end{proof}

We are now ready to prove Theorem \ref{thrm:mainI1}.

\begin{proof}[Proof of Theorem \ref{thrm:mainI1}.]
In Lemma \ref{lem:firstpart}, we have shown that if all the conditions in  Theorem \ref{thrm:mainI1} are met, then \eqref{eq:thrmI1-6} follows. To prove the second half of the theorem, it is clear from Lemma \ref{lem:kinfBexsits} and Lemma \ref{lem:kfBexsits} that for any $N$, $0< \epsilon, \delta < 1$, and $K$ sufficiently large, there exists $\gamma_{\alpha}, \gamma_{\beta} \in \mathbb{R}_+$ and $q_\alpha, q_\beta \in \mathbb{R}$ such that $\mathcal{B}_{\alpha, \beta}$ is nonempty. Moreover, it is straightforward from the subspace topology criterion that $\mathcal{B}_{\alpha, \beta}$ is open in $\mathcal{B}$. Precisely, let $V$ be the subset of $M_{N \times N}(\mathbb{R}_+)$ where all conditions hold. It is clear that $V$ is open in $M_{N \times N}(\mathbb{R}_+)$ and $\mathcal{B}_{\alpha, \beta} = V \cap \mathcal{B}.$ Therefore, $\mathcal{B}_{\alpha, \beta}$ is open in $ \mathcal{B}$. Further, for each $B \in \mathcal{B}_{\alpha, \beta}$, it follows from Lemma \ref{lem:secondpart} and Lemma \ref{lem:secondpart2} that $U_{B, \alpha, \beta} $ is open and non-empty in $\mathcal{A}_B$. Therefore, 
$$\mathcal{P}_{\alpha, \beta} = \displaystyle \bigcup_{B \in  \mathcal{B}_{\alpha, \beta}} U_{B, \alpha, \beta}$$ is open and non-empty in $\mathcal{P}$.
This completes the proof of the theorem.
\end{proof}

	Theorems \ref{cor:mainI1} and \ref{cor:mainI1N} can be proved using the same kind of techniques as in Theorem \ref{thrm:mainI1}. The proof of Theorem \ref{cor:mainI1} is much simpler and we only give a sketch below.

\begin{proof}[Proof of Theorem \ref{cor:mainI1}.]
	Using the same techniques for proving Lemma \ref{lem:gbsiestimategsize} and \ref{lem:mcestimategsize}, we can easily show that
	\begin{align*}
		Q_{\Haf^2}^{\textnormal{GBS-I}} &\leq \frac{1}{d}  \left( t + \frac{1}{\sqrt{\pi}} G_{2 q_{\beta}-\frac{1}{2}, K, N}(\gamma_\beta \bmax) \right), \\
		Q_{\Haf^2}^{\textnormal{MC}} &\geq  e^{\frac{1}{25} - \frac{1}{6}} R_{q_\alpha, K} (4 \gamma_\alpha \bmin^{2}).
	\end{align*}
	Plugging these bounds in \eqref{eq:thrmI1-3-prime}, we obtain \eqref{eq:corI1-6}.
	Hence, the first part of the theorem is proved. To establish the second part, we apply the same reasoning as in Lemma \ref{lem:kinfBexsits} and Lemma \ref{lem:kfBexsits} to show that for any $s, c>0$, any $N$ and $K$ sufficiently large, there exists $\gamma_{\alpha}, \gamma_{\beta} \in \mathbb{R}_+$ and $q_\alpha, q_\beta \in \mathbb{R}$ such that $\mathcal{B}'_{\alpha, \beta}$ is nonempty. It is also clear that $\mathcal{B}'_{\alpha, \beta}$ is open in $\mathcal{B}$. Further, for each $B \in \mathcal{B}'_{\alpha, \beta}$, we introduce the following subset
	\begin{align*}
		&U_{B, \alpha, \beta} = \left\{ (a_I) \in \mathcal{A}_B \mid a_0^2 < s, \right. \\
		&\hspace{1.4in} \frac{ k^{q_\alpha} \gamma_\alpha^{k}}{(2k)!} < \sum_{\vert I \vert = 2k} a_I < \frac{k^{q_\beta} \gamma_\beta^{k}}{(2k)!}, ~~ k = 1, 2, \dots, K, \\
		& \hspace{1.4in} \sumck a^2_I I! < \frac{k^{2q_\beta} \gamma_\beta^{2k}}{(2k)!^2} m_k, ~~ k = 1, 2, \dots, K \}
	\end{align*}
	Clearly, $U_{B, \alpha, \beta} $ is open $\mathcal{A}_B$. To show that $U_{B, \alpha, \beta}$ is non-empty, we modify the proof of Lemma \ref{lem:secondpart} by choosing $0 < t < 1$ instead. Then, the $a_I$'s constructed in that proof is an element of $U_{B, \alpha, \beta}$. Therefore, 
	$$\mathcal{P}_{\alpha, \beta} = \displaystyle \bigcup_{B \in  \mathcal{B}'_{\alpha, \beta}} U_{B, \alpha, \beta}$$ is open and non-empty in $\mathcal{P}$.
\end{proof}

We now give lemmas that will be useful for proving Theorem \ref{cor:mainI1N}. 
\begin{lemma}
	\label{lem:ggexist_increasen}
	If $\zeta > 0$, $q \leq \frac{1-N}{2}$ and $z = \tau N$ with $0 < \tau < 1$, then for all $N \geq N_G$ and $K \geq \zeta N^2$,
	\begin{align*}
		G_{q, K, N}(z) < C_G N^3,
	\end{align*}
	where $$C_G = \frac{1}{8 \sqrt{\pi}} (\tau e )^2$$ and $N_G$ is the minimal such that
	\begin{align*}
		N_G^{q - \frac{1}{2} -3} \left( 2\pi \right)^{\frac{N_G-1}{2}}
		e^{\frac{N_G}{13}} \frac{(2\tau)^{N_G+1}-1}{2\tau -1}\frac{1}{1-\tau^{N_G}} \frac{16 \sqrt{\pi}}{(\tau e)^2} < 1.
	\end{align*}
\end{lemma}

\begin{proof}
	Recall that
	\begin{align*}
		G_{q, K, N}(z) & = \hfun_{q, \lfloor \frac{N}{2} \rfloor}(z) + \left( 2\pi \right)^{\frac{N-1}{2}} N^{q - \frac{1}{2}}
		e^{\frac{N}{13}} \multilog_{0, N} \left(\frac{2z}{N}
		\right)
		\multilog_{\frac{1}{2} - \frac{N}{2}- q, s_K} \left(\frac{z^N}{N^N}
		\right) \\
		& \leq \hfun_{q, \lfloor \frac{N}{2} \rfloor}(z) + \left( 2\pi \right)^{\frac{N-1}{2}} N^{q - \frac{1}{2}}
		e^{\frac{N}{13}} \multilog_{0, N} \left(\frac{2z}{N}
		\right)
		\multilog_{0, s_K} \left(\frac{z^N}{N^N}
		\right).
	\end{align*}
	First, we compute that 
	\begin{align*}
		\multilog_{0, s_K} \left(\frac{z^N}{N^N} \right) = \sum_{k =1}^{s_K} \tau^N < \frac{1}{1-\tau^N}.
	\end{align*}
	Then,
	\begin{align*}
		 \left( 2\pi \right)^{\frac{N-1}{2}} N^{q - \frac{1}{2}}
		e^{\frac{N}{13}} \multilog_{0, N} \left(\frac{2z}{N}
		\right)
		\multilog_{0, s_K} \left(\frac{z^N}{N^N}
		\right) < N^{q - \frac{1}{2}} \left( 2\pi \right)^{\frac{N-1}{2}}
		e^{\frac{N}{13}} \frac{(2\tau)^{N+1}-1}{2\tau -1}\frac{1}{1-\tau^N}.
	\end{align*}
	We now compute that
	\begin{align}
		\hfun_{q, \lfloor \frac{N}{2} \rfloor}(z) & = \sum_{k = 1}^{\lfloor \frac{N}{2} \rfloor} \frac{{k}^q}{(2k)!} z^{2k} \notag \\
		& \leq \sum_{k = 1}^{\lfloor \frac{N}{2} \rfloor} \frac{{k}^{\frac{1-N}{2}}}{(2k)!} (\tau N)^{2k} \notag \\
		& \leq \frac{1}{2 \sqrt{\pi}}\sum_{k = 1}^{\lfloor \frac{N}{2} \rfloor} \frac{(\tau e N)^{2k}}{2^{2k} k^{N + 2k}} \label{eq:hhnincrease1} \\
		& \leq \frac{1}{8 \sqrt{\pi}} \frac{N}{2} (\tau e N)^2 \label{eq:hhnincrease2} \\
		& = \frac{(\tau e)^2}{16 \sqrt{\pi}} N^3. \notag
	\end{align}
	The inequality \eqref{eq:hhnincrease1} follows from Stirling's formula. To derive \eqref{eq:hhnincrease2}, we define $f: \mathbb{R} \rightarrow \mathbb{R}$ to be
	\begin{align*}
		f(x) = \frac{(\tau e N)^{2x}}{2^{2x} x^{N + 2x}}.
	\end{align*}
	The derivative of $f(x)$ is given by
	\begin{align*}
		f'(x) = \left(\frac{\tau e N}{2}\right)^{2x} \frac{1}{x^{1 + N + 2x}} \left(-N + 2x \log\left(\frac{\tau N}{2 x}\right) \right).
	\end{align*}
	Notice that if $1 \leq x \leq \lfloor \frac{N}{2} \rfloor$, then $f'(x) <0$. Therefore, the terms in \eqref{eq:hhnincrease1} are decreasing for $ k = 1, 2, \dots,  \lfloor \frac{N}{2} \rfloor$, and thus each term can be bounded above by the first term, that is $\frac{(\tau e N)^2}{4}$. 
	Therefore, 
	\begin{align*}
		G_{q, \zeta N^2, N}(z) & \leq  \frac{(\tau e)^2}{8 \sqrt{\pi}} N^3 \left( \frac{1}{2} + N^{q - \frac{1}{2} -3} \left( 2\pi \right)^{\frac{N-1}{2}}
		e^{\frac{N}{13}} \frac{(2\tau)^{N+1}-1}{2\tau -1}\frac{1}{1-\tau^N} \frac{8 \sqrt{\pi}}{(\tau e)^2} \right).
	\end{align*}
\end{proof}

\begin{lemma}
	\label{lem:rrexist_increasen}
If $\zeta >0$ and $q < 0$, then for all $N$ and $K \geq \zeta N^2$,
\begin{align}
	R_{q, K}(z) > \frac{1}{2\sqrt{\pi}} (\zeta N^2)^{2q- \frac{1}{2}} z^{\zeta N^2}. \label{eq:rrexists}
\end{align}
If we further assume that $z >1$ and $q > cN$ for some negative constant $c$, then $z^{\zeta N^2}$ dominates the right hand side of \eqref{eq:rrexists} and $R_{q, K}(z) \rightarrow \infty$ as $N \rightarrow \infty$.
\end{lemma}

\begin{proof}
	The proof is straightforward by noting that
	\begin{align*}
		R_{q, \zeta N^2} (z) & = \frac{1}{2\sqrt{\pi}} \sum_{k = 1}^{K} k^{2q - \frac{1}{2}} z^k
		> \frac{1}{2\sqrt{\pi}} (\zeta N^2)^{2q - \frac{1}{2}} z^{\zeta N^2}.
	\end{align*}
	If $q > cN$, then
	\begin{align*}
		& \quad \zeta N^2 \log(z) + \left(2q -\frac{1}{2} \right) \log(\zeta N^2) \\
		& > \zeta N^2 \log(z) + 2c N \log(\zeta N^2) - \frac{1}{2} \log(\zeta N^2) \\
		&= \zeta N^2 \log(z) \left(  1+ \frac{2c}{\zeta N} \log(\frac{\zeta N^2}{z}) - \frac{1}{2 \zeta N^2} \log(\frac{\zeta N^2}{z}) \right).
	\end{align*}
	Since
	\begin{align*}
		\left(  1+ \frac{2c}{\zeta N} \log(\frac{\zeta N^2}{z}) - \frac{1}{2 \zeta N^2} \log(\frac{\zeta N^2}{z}) \right) \rightarrow 1
	\end{align*}
	as $N \rightarrow \infty$, the term $z^{\zeta N^2}$ dominates the right hand side of \eqref{eq:rrexists} and $R_{q, K}(z) \rightarrow \infty$ as $N \rightarrow \infty$.
\end{proof}

\begin{lemma}
	\label{lem:bbexist_increasen}
	Let $p > 0$. If $b_1 = \frac{\tau_1}{N}$ and $b_2 = \frac{\tau_2}{N}$ with $0 < \tau_1 < \tau_2$ and $\tau_2 + (N-1)\tau_1 < N$, then there exists $B \in \mathcal{B}$ such that $\bmin = b_1$, $\bmax = b_2$ and
	\begin{align*}
		\frac{1}{d} < C_D N^p, ~~~~ \forall N \geq N_D
	\end{align*}
	where $N_D$ is the minimal such that
	\begin{align*}
		\tau_1 + \frac{\tau_2-\tau_1}{N_D} < 1
	\end{align*}
	and
	\begin{align*}
		C_D = \left(1 - \left(\tau_1 + \frac{\tau_2-\tau_1}{N_D} \right)^2 \right)^{-\frac{p}{2}}.
	\end{align*}

\end{lemma}

\begin{proof}
	Since $d$ is completely determined by the eigenvalues of $B$, we first show that if at least $N - p\ln(N)$ many eigenvalues are close to zero, then $\tfrac{1}{d}$ can only grow like $N^p$.
	Let $m$ be some positive integer such that $	1 < m \leq p \ln(N).$
	We set
	\begin{align}
		\lambda_{m+1} = \lambda_{m+2} = \dots = \lambda_{N} = \sqrt{1 - \tau^2} ~~~~\textnormal{with} ~~~~N^{\tfrac{p}{2(m-N)}} \leq \tau < 1. \label{eq:bexistsn1}
	\end{align}
	Then, 
	\begin{align*}
		\prod_{j=1}^m (1- \lambda_j^2)^{-\tfrac{1}{2}} \leq (1- \lambda_1^2)^{-\tfrac{m}{2}} \leq (1- \lambda_1^2)^{-\tfrac{p}{2}\ln(N)} \leq ({1-\lambda_1^2})^{-\frac{p}{2}}N^{\frac{p}{2}},
	\end{align*}
	and
	\begin{align*}
		\prod_{j=m+1}^N (1- \lambda_j^2)^{-\tfrac{1}{2}} \leq (1- \lambda_1^2)^{-\tfrac{N-m}{2}} \leq \tau^{m-N} \leq N^{\frac{p}{2}}.
	\end{align*}
	In total, we have
	\begin{align*}
		\frac{1}{d} = \prod_{j=1}^N (1- \lambda_j^2)^{-\tfrac{1}{2}} \leq ({1-\lambda_1^2})^{-\frac{p}{2}} N^{p}.
	\end{align*}
	
	We now show that there exists $B \in \mathcal{B}$ with $\bmin = b_1$, $\bmax = b_2$ and eigenvalues as specified in \eqref{eq:bexistsn1}.
	Consider $B$ of the following form
	\begin{align*}
		B &= \begin{bmatrix}
			b_2 & b_1 & \dots & b_1 & b_1\\
			b_1 & b_2 & \dots & b_1 & b_1 \\
			\vdots & & & & \vdots \\
			b_1 & b_1 & \dots & b_1 + \sqrt{1 - \tau^2} & b_1 \\
			b_1 & b_1 & \dots & b_1 & b_1 + \sqrt{1 - \tau^2}
		\end{bmatrix} = b_1 \mathds{1} + \textnormal{diag}(v).
	\end{align*}
	Here, 
	$$v = (b_2 - b_1, \dots, b_2-b_1, \sqrt{1 - \tau^2}, \dots, \sqrt{1 - \tau^2})$$ contains $m-1$ entries of $b_2 - b_1$ and $N-m+1$ entries of $\sqrt{1 - \tau^2}$. Clearly, $B$ is symmetric with $\bmin = b_1$ and $\bmax = b_2$, provided that $\sqrt{1 - \tau^2} < b_2 - b_1$. From the Perron–Frobenius theorem, we obtain that
	\begin{align}
		\lambda_1 < \max_{i} \sum_{j = 1}^N B_{ij} = b_2 + (N-1)b_1 < 1. \label{eq:frob}
	\end{align}
	Furthermore, for all $x \in \mathbb{R}^N - 0$, we have that 
	\begin{align*}
		x^\intercal \mathds{1} x \geq 0
	\end{align*}
	and 
	\begin{align*}
		x^\intercal \textnormal{diag}(v) x = \sum_{i=1}^{m-1} b_2 x_i^2 +  \sum_{i=m}^{N} b_1 + \sqrt{1 - \tau^2} x_i^2 > 0.
	\end{align*}
	Therefore, 
	\begin{align*}
		x^\intercal B x > 0,
	\end{align*}
	which proves that $B$ is positive definite and $B \in \mathcal{B}$. 
	It remains to verify that the constructed matrix $B$ has $N-m$ eigenvalues equal to $\sqrt{1 - \tau^2}$. Define $u_i$ as the vector with 1 at the $i$-th entry, $-1$ at the $(i+1)$-th entry, and 0 elsewhere, i.e.,
	\begin{align*}
		u_i = [	0, \dots 0, 1, -1, \dots, 0].
	\end{align*}
	For $i = N-m+1, \dots, N-1$, we have
	\begin{align*}
		B u_i = \sqrt{1 - \tau^2} u_i,
	\end{align*}
	and the vectors $u_i$ are linearly independent. Therefore, there are $N-m$ eigenvalues equal to $\sqrt{1 - \tau^2}$.
	
	Finally, using \eqref{eq:frob}, we obtain that
	\begin{align*}
		({1-\lambda_1^2})^{-\frac{p}{2}} & \leq \left(1 - \left(b_2 + (N-1)b_1 \right)^2 \right)^{-\frac{p}{2}} \\
		& = \left(1 - \left(Nb_1 + b_2 -b_1 \right)^2 \right)^{-\frac{p}{2}} \\
		& = \left(1 - \left(\tau_1 + \frac{\tau_2-\tau_1}{N} \right)^2 \right)^{-\frac{p}{2}}.
	\end{align*}
\end{proof}

\begin{lemma}
	\label{lem:cormainI1cond}
	If the conditions in Theorem \ref{cor:mainI1N} are met, then \eqref{eq:ccpoly} and \eqref{eq:ccexp} hold. 
\end{lemma}

\begin{proof}
	It follows straightforwardly from the proof of Lemmas \ref{lem:gbsiestimategsize} and \ref{lem:mcestimategsize} that
	\begin{align*}
		Q_{\Haf^2}^{\textnormal{GBS-I}} &\leq \frac{1}{d}  \left( a_0^2 + \frac{1}{\sqrt{\pi}} G_{2 q_{\beta}-\frac{1}{2}, K, N}(\gamma_\beta \bmax) \right)\\
		Q_{\Haf^2}^{\textnormal{MC}} &\geq  a_0^2 + e^{\frac{1}{25} - \frac{1}{6}} R_{q_\alpha, K} (4 \gamma_\alpha \bmin^{2}).
	\end{align*}
	From \eqref{eq:uniform1} -- \eqref{eq:uniform2}, we derive
	\begin{align*}
		Q_{\Haf^2}^{\textnormal{GBS-I}} &< C_D N^p \left(a_0^2 + \frac{1}{\sqrt{\pi}} C_G N^3\right) \\
		Q_{\Haf^2}^{\textnormal{MC}} &> \frac{e^{\frac{1}{25} - \frac{1}{6}}}{2\sqrt{\pi}} (\zeta N^2)^{2q_\alpha - \frac{1}{2}} C_R^{\zeta N^2}
	\end{align*}
	Since $a_0^2 \leq \mu^2_{\Haf^2} \leq Q_{\Haf^2}^{\textnormal{GBS-I}}$,
	we further obtain that
	\begin{align*}
		&n_{\Haf^2}^{\textnormal{GBS-I}} < \frac{1}{\epsilon^2 \delta}\left(\frac{C_D N^p \left(a_0^2 + \frac{1}{\sqrt{\pi}} C_G N^3\right)}{a_0^2} - 1\right) = \frac{1}{\epsilon^2 \delta}\left(C_D N^p + \frac{C_G C_D N^{3+p}}{a_0^2 \sqrt{\pi}} - 1\right)\\
		&n_{\Haf^2}^{\textnormal{MC}} > \frac{1}{\epsilon^2 \delta} \frac{\frac{e^{\frac{1}{25} - \frac{1}{6}}}{2\sqrt{\pi}} (\zeta N^2)^{2q_\alpha - \frac{1}{2}} C_R^{\zeta N^2} }{Q_{\Haf^2}^{\textnormal{GBS-I}}} = \frac{1}{\epsilon^2 \delta} \frac{\frac{e^{\frac{1}{25} - \frac{1}{6}}}{2\sqrt{\pi}} (\zeta N^2)^{2q_\alpha - \frac{1}{2}} C_R^{\zeta N^2} }{C_D N^p \left(a_0^2 + \frac{1}{\sqrt{\pi}} C_G N^3\right)}. 
	\end{align*}
\end{proof}

We now prove there exists $q_\alpha \leq q_\beta$ and $0 < \gamma_{\alpha} \leq \gamma_{\beta}$ such that $\mathcal{B}''_{\alpha, \beta}$ is non-empty for any $0 < \epsilon, \delta < 1$, $p>0$ and $N$ sufficiently large. For convenience, let $q_\alpha = \nu_\alpha N$ and $q_\beta = \nu_\beta N$, and define $\gamma_{\alpha} = \tau_\alpha N^2$ and $\gamma_{\beta} = \tau_\beta N^2$. 

\begin{lemma}
	\label{lem:mainI1increasenBexists}
	When $N$ is sufficiently large and
	if $c \leq \nu_\alpha \leq \nu_\beta \leq -\frac{1}{4}$ for some constant $c$, $\frac{1}{4} < \tau_\alpha$, and $\frac{\tau_\beta^2}{4} < \tau_\alpha \leq \tau_\beta$, then $\mathcal{B}''_{\alpha, \beta}$ is non-empty.
\end{lemma}

\begin{proof}
	Since $c \leq \nu_\alpha \leq \nu_\beta \leq -\frac{1}{4}$, we have $2cN \leq 2q_\alpha \leq 2q_\beta \leq -\frac{N}{2}$. It then follows from Lemmas \ref{lem:ggexist_increasen} and \ref{lem:rrexist_increasen} that a sufficient condition for \eqref{eq:uniform1} and \eqref{eq:uniform2} is
	\begin{align*}
		&1 <	4 \gamma_{\alpha} \bmin^2 = 4 \tau_\alpha N^2 \bmin^2 \\
		&1 > \frac{\gamma_{\beta} \bmax}{N } = \tau_\beta N \bmax.
	\end{align*}
	This is equivalent to 
	\begin{align*}
		\frac{1}{2N \sqrt{\tau_{\alpha}}} < \bmin < \bmax < \frac{1}{\tau_\beta N}.
	\end{align*}
	Now we choose $b_1 = \frac{\tau_1}{N}$, $b_2 = \frac{\tau_2}{N}$ such that $\frac{1}{2 \sqrt{\tau_{\alpha}}} < \tau_1 < \tau_2 < \frac{1}{\tau_\beta}$ and $\tau_2 + (N-1)\tau_1 < N$. This can be done by first picking $\tau_1$ such that
	\begin{align*}
		\frac{1}{2 \sqrt{\gamma_{\alpha}}} < \tau_1 < \min\left(\frac{1}{\tau_\beta } , 1\right).
	\end{align*}
	Now we pick 
	$$\tau_2 = \tau_1 + \min\left(\frac{1}{2} \left(N- N\tau_1\right) ,\frac{1}{2} \left(\frac{1}{\tau_\beta } - \tau_1\right) \right).$$ By construction, $\tau_2 < \frac{1}{\tau_\beta}$ and $\tau_2 + (N-1)\tau_1 < N$.
	Finally, with the chosen $b_1, b_2$, we apply Lemma \ref{lem:bbexist_increasen} to obtain a $B \in \mathcal{B}$ so that \eqref{eq:uniform3} is satisfied with $\bmax = b_2$ and $\bmin = b_1$.
\end{proof}

\begin{remark}
Note that the minimal $N$ can be determined using Lemmas \ref{lem:ggexist_increasen} and \ref{lem:bbexist_increasen}. The constants given in the remark following Theorem \ref{cor:mainI1N} are also computed based on the bounds from Lemmas \ref{lem:ggexist_increasen}, \ref{lem:rrexist_increasen}, and \ref{lem:bbexist_increasen}.
\end{remark}

Finally, we prove Theorem \ref{cor:mainI1N}.
\begin{proof}[Proof of Theorem \ref{cor:mainI1N}.]
	In Lemma \ref{lem:cormainI1cond}, we have shown that if all the conditions in Theorem \ref{cor:mainI1N} are met, then  \eqref{eq:ccpoly} and \eqref{eq:ccexp} hold. To prove the second half of the statement, Lemma \ref{lem:mainI1increasenBexists} ensures the existence of $\gamma_{\alpha}, \gamma_{\beta} \in \mathbb{R}_+$ and $q_\alpha, q_\beta \in \mathbb{R}$ such that $\mathcal{B}''_{\alpha, \beta}$ is nonempty for any $0 < \epsilon, \delta < 1$, $p>0$ and $N$ sufficiently large. Moreover, by the subspace topology criterion, $\mathcal{B}''_{\alpha, \beta}$ is open in $\mathcal{B}$. For each $B \in \mathcal{B}''_{\alpha, \beta}$, we define $U_{B, \alpha, \beta} $ as in the proof of Theorem \ref{cor:mainI1}, and using the same argument we can show that this subset is open and non-empty in $\mathcal{A}_B$. Therefore, 
	$$\mathcal{P}_{\alpha, \beta} = \bigcup_{B \in  \mathcal{B}''_{\alpha, \beta}} U_{B, \alpha, \beta}$$ is open and non-empty in $\mathcal{P}$. 
\end{proof}
	\subsection{Proof of Theorem \ref{thrm:mani2}}
One can obtain a similar result when comparing $n_{\Haf}^{\textnormal{GBS-P}}$ and $n_{\Haf}^{\textnormal{MC}}$. In this section, we give a precise statement for Theorem \ref{thrm:mani2}.

The space of $a_I$'s and $B$'s on which we will compare GBS-P and MC is defined similarly as in the comparison of GBS-I and MC. Let $\mathcal{A}$ and $\mathcal{B}$ be defined as in \eqref{eq:Aspace} and \eqref{eq:Bspace} respectively. Let $\mathcal{A}_B \subseteq \mathcal{A}$ to be the following,
\begin{align*}
	\mathcal{A}_B = & \left\{ (a_I) \in \mathcal{A} \Bigm| 0 \neq \sum_{k = 0}^K \sumck a_I \Haf(B_I) < \infty, \right. \\
	&\quad \quad \quad \quad \quad  ~~ \sum_{k = 0}^K \sumck a_I I! \Haf(B_I) < \infty, \\
	&\quad \quad \quad \quad \quad \left. \sumdoublefinite a_I a_J \Haf(B_{I + J}) < \infty \right\}.
\end{align*}
It is clear from definition that $\mathcal{I}^\times_{\Haf}(\epsilon, \delta)$, $n_{\Haf}^{\textnormal{GBS-P}}$ and $n_{\Haf}^{\textnormal{MC}}$ are well-defined on $\mathcal{A}_B$. Then, 
$$ \mathcal{P} = \bigcup_{B \in \mathcal{B}} \mathcal{A}_B \subseteq \mathcal{A} \times \mathcal{B} $$ forms the problem space on which we will compare  $n_{\Haf}^{\textnormal{GBS-P}}$ and $n_{\Haf}^{\textnormal{MC}}$. For future reference, let us also define $\mathcal{A}^1_B \subseteq \mathcal{A}$ such that
\begin{align*}
	\mathcal{A}^1_B = & \left\{ (a_I) \in \mathcal{A} \mid \sum_{k = 0}^K \sumck a_I \Haf(B_I) =1, \right. \\
	&\quad \quad \quad \quad \quad  ~~ \sum_{k = 0}^K \sumck a_I I! \Haf(B_I) < \infty, \\
	&\quad \quad \quad \quad \quad \left. \sumdoublefinite a_I a_J \Haf(B_{I + J}) < \infty \right\}.
\end{align*}

We now introduce the subset $\mathcal{B}_{\alpha, \beta} \subseteq \mathcal{B}$. Let $s_1, s_2, \gamma_{\alpha}, \gamma_{\beta} \in \mathbb{R}_{+}$ and  $q_\alpha, q_\beta \in \mathbb{R}$. We define
\begin{equation}
	\label{eq:cdef2}
	\begin{aligned}
		&c_1 = 1 + \frac{1}{\sqrt{\pi}} e^{\frac{1}{25} - \frac{1}{6}} \multilog_{\frac{1}{2}-q_\alpha , K}\left(2 \gamma_\alpha \bmin \right)  \\
		&c_2 = 1 + \frac{1}{\sqrt{\pi}} \multilog_{\frac{1}{2}-q_\beta, K}\left(2 \gamma_\beta \bmax \right).
	\end{aligned}
\end{equation}
We require that for all $k = 0, 1, 2, \dots $
\begin{equation}
	\label{eq:cgamma2}
	\begin{gathered}
		c_2^{-1} < c_1^{-1} \\
			c_2^{-1} { k^{q_\alpha} \gamma_\alpha^{k}} <  c_1^{-1} k^{q_\beta} \gamma_\beta^{k}, ~~~~ k = 1, 2, \dots, K.
	\end{gathered}
\end{equation}
We further require that
\begin{equation}
	\label{eq:thrmI2-3} 
	\begin{aligned}
		& \frac{s_1}{\epsilon^2 \delta}  \frac{c_1^{-1}}{d} \left( 1 +  G_{q_\beta, K, N}\left(\sqrt{2\gamma_\beta \bmin^{-1}}\right) \right) \\
		&\leq \ln \left( 1 + 2 e^{\frac{1}{25} - \frac{1}{6}} R_{q_\alpha, K} (4 \gamma_\alpha \bmin) \right) + 2 \ln\left(c_2^{-1} \right) - \ln(\epsilon^2 \delta + 1) + \frac{s_1}{\epsilon^2 \delta} -\ln(s_2).
	\end{aligned}
\end{equation}
Here, the functions $G_{q, K, N}$ and $R_{q, K}$ are defined as before (see \eqref{def:funGK} -- \eqref{def:funRKneg}). 
If $K = \infty$, we also require that 
\begin{gather}
	4 \gamma_\beta \bmax < 1 \label{eq:thrmI2-4}  \\
	2\gamma_\beta \bmin^{-1} < N^2 \label{eq:thrmI2-5} 
\end{gather}
In other words,
\begin{align}
	\label{eq:spaceBalphabeta2}
	\mathcal{B}_{\alpha, \beta} = \left\{ B \in \mathcal{B} \mid  \eqref{eq:cgamma2}, \eqref{eq:thrmI2-3}, \eqref{eq:thrmI2-4} \textnormal{ and } \eqref{eq:thrmI2-5} \textnormal{ hold} \right\}.
\end{align}

We now give a precise statement for Theorem \ref{thrm:mani2}. 
 
\begin{theorem}
	\label{thrm:mainI2}
	Let $N$, $K$ and $0 < \epsilon, \delta < 1$ be given. Let $s_1, s_2, \gamma_{\alpha}, \gamma_{\beta} \in \mathbb{R}_{+}$ and  $q_\alpha, q_\beta \in \mathbb{R}$. Suppose $B \in \mathcal{B}_{\alpha, \beta}$ and the $a_I$'s satisfy the following requirements: there exists $I$ such that $a_I \neq 0$, and 
	\begin{align*}
		c_2^{-1} < a_0 < c_1^{-1}.
	\end{align*}
	We further require that for all $1 \leq k \leq K$,
	\begin{gather}
		c_2^{-1}  \mu_{\Haf} \frac{k^{q_\alpha} \gamma_\alpha^{k}}{k!} \leq \sum_{\vert I \vert = 2k} a_I \leq c_1^{-1}  \mu_{\Haf} \frac{k^{q_\beta} \gamma_\beta^{k}}{k!},
		\label{eq:thrmI2-1} \\
		\sumck a_I I! \leq c_1^{-1} \mu_{\Haf} \frac{k^{q_\beta} \gamma_\beta^{k}}{k!} m_k, \label{eq:thrmI2-2} 
	\end{gather}
	Then, 
	\begin{align}
		s_2\exp(s_1 \, n_{\Haf}^{\textnormal{GBS-P}} ) < n_{\Haf}^{\textnormal{MC}} < \infty. \label{eq:thrmI2-6} 
	\end{align}
	Moreover, for any $0 < \epsilon, \delta < 1$, $N$ and sufficiently large $K$, there exists a non-empty and open subset of $\mathcal{P} $ for which the conditions above are satisfied. 
\end{theorem}

Similar to Theorem \ref{cor:mainI1}, we can show that there are problems where GBS-P outperforms MC uniformly across $0< \epsilon, \delta <1$. We define the subset $\mathcal{B}'_{\alpha, \beta} \subseteq \mathcal{B}$ by requiring
\begin{equation}
	\frac{1}{d} \left( s +  G_{q_\beta, K, N}\left(\sqrt{2\gamma_\beta \bmin^{-1}}\right) \right) < c \left(1 + 2 e^{\frac{1}{25} - \frac{1}{6}} R_{q_\alpha, K} (4 \gamma_\alpha \bmin) \right), \label{eq:gbspbetanewcond}
\end{equation}
for some positive constants $s$ and $c$.
When $K = \infty$, we also require \eqref{eq:thrmI2-4} and \eqref{eq:thrmI2-5}. In other words,
\begin{align*}
	\mathcal{B}'_{\alpha, \beta} = \left\{ B \in \mathcal{B} \mid  \eqref{eq:thrmI2-4},   \eqref{eq:thrmI2-5}  \textnormal{ and } \eqref{eq:gbspbetanewcond} \textnormal{ hold} \right\}.
\end{align*}

\begin{theorem}
	\label{cor:mainI2}
	Let $N$ and $K$ be given.
	Let $\gamma_{\alpha}, \gamma_{\beta} \in \mathbb{R}_{+}$ and  $q_\alpha, q_\beta \in \mathbb{R}$. Let $c >0$. If $B \in \mathcal{B}'_{\alpha, \beta}$ and
	the $a_I$'s satisfy the following requirements: $a_0^2 \leq s$
	and for all $1 \leq k \leq K$
	\begin{gather}
	 \frac{k^{q_\alpha} \gamma_\alpha^{k}}{k!} \leq \sum_{\vert I \vert = 2k} a_I \leq  \frac{k^{q_\beta} \gamma_\beta^{k}}{k!},  \label{eq:corhaf1} \\
		\sumck a_I I! \leq \frac{k^{q_\beta} \gamma_\beta^{k}}{k!} m_k,  \label{eq:corhaf2}
	\end{gather}
	then, for all $0< \epsilon, \delta < 1$,
	\begin{align}
		n_{\Haf}^{\textnormal{GBS-P}} < c \, n_{\Haf}^{\textnormal{MC}} < \infty.
	\end{align}
	Moreover, for any $c>0$, $N$ and sufficiently large $K$, there exists a non-empty and open subset of $\mathcal{P}$ for which the conditions above are satisfied.
\end{theorem}

We now compare GBS-P and MC as $N$ grows. Suppose $N$ and $K$ are correlated such that $K \geq \zeta N^2$ with $\zeta >0$, then we obtain an exponential speedup of GBS-P with $N$. Let us define $\mathcal{B}''_{\alpha, \beta} \subseteq \mathcal{B}$ as follows. Let $\gamma_{\alpha}, \gamma_{\beta} \in \mathbb{R}_{+}$ and  $q_\alpha, q_\beta \in \mathbb{R}$. Let $p >0$. We first require
\begin{align}
	G_{q_{\beta}, K, N}\left(\sqrt{2 \gamma_{\beta} \bmin^{-1}} \right) < C_G N^3 \label{eq:uniformpp1}
\end{align}
for some constant $C_G >0$. We further require that 
\begin{gather}
	R_{q_\alpha, K} (4 \gamma_\alpha \bmin)  > \frac{1}{2\sqrt{\pi}} (\zeta N^2)^{2q_\alpha - \frac{1}{2}} C_R^{\zeta N^2} \label{eq:uniformpp2}
\end{gather}
for some $C_R >1$ where $C_R^{\zeta N^2}$ dominates the right hand side of \eqref{eq:uniformpp2} as $N$ grows. 
We finally require that for all $N \geq N_D$, 
\begin{align}
	\frac{1}{d} < C_D N^p \label{eq:uniformpp3}
\end{align}
for some constant $C_D >0$.
We define
\begin{align*}
	\mathcal{B}''_{\alpha, \beta} = \left\{ B \in \mathcal{B} \mid \eqref{eq:uniformpp1}, \eqref{eq:uniformpp2} \textnormal{ and }  \eqref{eq:uniformpp3} \textnormal{ hold} \right\}.
\end{align*}

\begin{theorem}
	\label{cor:mainI2N}
	Let $N$, $K\geq \zeta N^2$ with $\zeta >0$ and $0 < \epsilon, \delta < 1$ be given.
	Let $\gamma_{\alpha}, \gamma_{\beta} \in \mathbb{R}_{+}$ and  $q_\alpha, q_\beta \in \mathbb{R}$. Let $p>0$. Suppose $B \in \mathcal{B}''_{\alpha, \beta}$ and the $a_I$'s satisfy \eqref{eq:corhaf1} and \eqref{eq:corhaf2} as in Theorem \ref{cor:mainI2}. If we further assume that $a_0$ is independent of $N$ and $a_0 \neq 0$, then
	\begin{align}
		&n_{\Haf}^{\textnormal{GBS-P}}< \frac{1}{\epsilon^2 \delta}\left(C_D N^p + \frac{C_G C_D N^{3+p}}{a_0^2 \sqrt{\pi}} - 1\right) \label{eq:ccpoly2}\\
		&n_{\Haf}^{\textnormal{MC}} > \frac{1}{\epsilon^2 \delta} \frac{\frac{e^{\frac{1}{25} - \frac{1}{6}}}{\sqrt{\pi}} (\zeta N^2)^{2q_\alpha - \frac{1}{2}} C_R^{\zeta N^2} }{C_D N^p \left(a_0^2 + \frac{1}{\sqrt{\pi}} C_G N^3\right)}  \label{eq:ccexp2}
	\end{align}
	Moreover, for any $0 < \epsilon, \delta < 1$,  $p >0$,  $N$ sufficiently large, and $K \geq \zeta N^2$ with $\zeta >0$, there exists a non-empty and open subset of $\mathcal{P} $ for which the conditions above are satisfied.
\end{theorem}

\begin{remark}
	Similar to Theorem \ref{cor:mainI1N}, the second part of the statement is proved by direct construction, and we can specify the constants explicitly for this constructed subset. 
	Suppose
	$q_\alpha = \nu_\alpha N$ and $q_\beta = \nu_\beta N$ with $c \leq \nu_\alpha \leq \nu_\beta \leq -\frac{1}{2}$ for some constant $c$. Furthermore, we ask that $\gamma_{\alpha} = \tau_\alpha N$ and $\gamma_{\beta} = \tau_\beta N$ with $\frac{1}{4} < \tau_\alpha \leq \tau_\beta < \frac{1}{2}$.
	We must also have $\bmin = \frac{\tau_1}{N}$, $\bmax = \frac{\tau_2}{N}$ with some positive constants $\tau_1 < \tau_2$. Then, a sufficiently large $N$ means
	\begin{align*}
		N \geq N_0 = \max(N_C, N_D)
	\end{align*}
	where $N_D$ is the minimal such that
	\begin{align*}
		\tau_1 + \frac{\tau_2-\tau_1}{N_D} < 1,
	\end{align*}
	and $N_C$ is the minimal such that
	\begin{align*}
		N_G^{q_\beta - \frac{1}{2} -3} \left( 2\pi \right)^{\frac{N_G-1}{2}}
		e^{\frac{N_G}{13}} \frac{(2\tau_2)^{N_G+1}-1}{2\tau_2 -1}\frac{1}{1-\tau_2^{N_G}} \frac{16 \sqrt{\pi}}{(\tau_2 e)^2} < 1.
	\end{align*}
	Furthermore,
	\begin{gather*}
		C_G = \frac{1}{8 \sqrt{\pi}} 2 \tau_\beta \left( \frac{e}{\tau_1}\right)^2, \\
		C_D = \left(1 - \left(\tau_1 + \frac{\tau_2-\tau_1}{N_D} \right)^2 \right)^{-\frac{p}{2}}, \\
		C_R = 4 \tau_\alpha \tau_1 >1.
	\end{gather*}
\end{remark}

The techniques used to prove these results are very similar to those used in Section \ref{subsec:thrm1}. We first prove Theorem \ref{thrm:mainI2} and this requires a few lemmas. 

\begin{lemma}
	\label{lem:gbspestimategsize}
	Given the conditions in Theorem \ref{thrm:mainI2},  $n^{\textnormal{GBS-P}}_{\Haf}$ is well-defined and 
	\begin{align*}
		Q_{\Haf}^{\textnormal{GBS-P}} \leq \frac{c_1^{-1} \mu_{\Haf}^2}{d} \left( 1 + G_{q_\beta, K, N}\left(\sqrt{2 \gamma_\beta \bmin^{-1}}\right) \right).
	\end{align*}
\end{lemma}

\begin{proof}
	When $K$ is finite, $n^{\textnormal{GBS-P}}_{\Haf}$ being well-defined holds trivially. When $K$ is infinite, it follows from  Lemma \ref{lem:swap2} that the conditions \eqref{eq:thrmI2-1} and \eqref{eq:thrmI2-4} imply that \eqref{eq:absconv} holds. We now compute that 
	\begin{align*}
		Q_{\Haf}^{\textnormal{GBS-P}} 
		& = \frac{\mu_{\Haf}}{d}  \left( \sum_{k = 0}^K \sum_{\vert J \vert = 2k} a_J J! \frac{1}{\Haf(B_J)} \right) \\
		& =  \frac{\mu_{\Haf}}{d}  \left(a_0+ \sum_{k= 1}^K \sum_{\vert J \vert = 2k} a_J J! \frac{1}{\Haf(B_J)} \right) \\
		& \leq \frac{c_1^{-1} \mu_{\Haf}^2}{d} \left( 1 +  \sum_{k = 1}^K  m_k \frac{ k^{q_\beta} }{(2k)!} 2^k \gamma_\beta^k \bmin^{-k} \right) \\
		& \hspace{1in} \textnormal{(use \eqref{eq:thrmI2-2} and Lemma \eqref{lem:hafapprox0}}) \\
		& = \frac{c_1^{-1} \mu_{\Haf}^2}{d} \left( 1 + G_{q_\beta, K, N}\left(\sqrt{2 \gamma_\beta \bmin^{-1}} \right) \right) \hspace{0.1in} \text{ (use Lemma \ref{lem:truncate}).}
	\end{align*}
	We further use \eqref{eq:thrmI2-5} to ensure that $ G_{q_\beta, K, N}$ is finite, and thus \eqref{eq:QPfinite} holds. Therefore $n^{\textnormal{GBS-P}}_{\Haf}$ is well-defined.
\end{proof}

\begin{lemma}
	\label{lem:mcestimategsize2}
	Given the conditions in Theorem \ref{thrm:mainI2}, $n^{\textnormal{MC}}_{\Haf}$ is well-defined and 
	\begin{align*}
		Q_{\Haf}^{\textnormal{MC}} \geq {c_2^{-2} \mu_{\Haf}^2} \left( 1 + 2 e^{\frac{1}{25} - \frac{1}{6}} R_{q_\alpha, K} (4 \gamma_\alpha \bmin) \right).
	\end{align*}
\end{lemma}

\begin{proof}
	The proof is similar to that of Lemma \ref{lem:mcestimategsize}.
	When $K$ is finite, $n^{\textnormal{MC}}_{\Haf}$ being well-defined holds trivially. When $K$ is infinite, it follows from  Lemma \ref{lem:swap2} and Theorem \ref{thrm:mcpp2} that the conditions \eqref{eq:thrmI2-1} and \eqref{eq:thrmI2-4} imply \eqref{eq:absconv} and \eqref{eq:swapuse2}.
	Therefore, $n^{\textnormal{MC}}_{\Haf}$ is well-defined.
	
	To give a lower bound of $Q_{\Haf}^{\textnormal{MC}}$, we use Lemma \ref{lem:hafapprox0} to compute that
	\begin{align*}
		Q_{\Haf}^{\textnormal{MC}} 
		&= \sumdoublefinite a_I a_J \Haf(B_{I + J})\\
		& \geq a_0^2 + \sum_{l = 1}^K \sum_{k_1 + k_2 = l}  \sum_{\substack{\vert I \vert = 2k_1 \\ \vert J \vert = 2k_2}} a_I a_J \bmin^{l} \frac{(2l)!}{2^{l} l!}
	\end{align*}
	Using \eqref{eq:thrmI2-1}, we have for all $l \geq 1$ 
	\begin{align*}
		\sum_{k_1 + k_2 = l}  \sum_{\substack{\vert I \vert = 2k_1 \\ \vert J \vert = 2k_2}} a_I a_J 
		& \geq \sum_{\substack{k_1 + k_2 = l \\ k_1 \neq 0 \text{ and} \\ k_2 \neq 0}} \sum_{\substack{\vert I \vert = 2k_1 \\ \vert J \vert = 2k_2}} a_I a_J +  2 a_0 \sum_{\substack{\vert I \vert = 2l}} a_I\\
		& \geq \sum_{\substack{k_1 + k_2 = l \\ k_1 \neq 0 \text{ and}\\ k_2 \neq 0}} c_2^{-2} \mu_{\Haf}^2 \frac{k_1^{q_\alpha} k_2^{q_\alpha}}{k_1!k_2!} \gamma_\alpha^l + 2  c_2^{-2} \mu_{\Haf}^2 \frac{l^{q_\alpha}}{l!} \gamma_\alpha^l \\
		& \geq c_2^{-2} \mu_{\Haf}^2 \gamma_\alpha^l \left( \sum_{\substack{k_1 + k_2 = l \\ k_1 \neq 0 \text{ and}\\ k_2 \neq 0}} \frac{k_1^{q_\alpha} k_2^{q_\alpha}}{k_1!k_2!} + 2 \frac{l^{q_\alpha}}{l!} \right) \\
		& \geq  c_2^{-2} \mu_{\Haf}^2 \gamma_\alpha^l S_l \sum_{k_1 + k_2 = l} \frac{1}{k_1!k_2!}
	\end{align*}
	where
	\begin{align*}
		S_l = \min\left( \min_{\substack{k_1 + k_2 = l \\ k_1 \neq 0 \text{ and}\\ k_2 \neq 0}}  (k_1 k_2)^{q_\alpha}, l^{q_\alpha}\right).
	\end{align*}
 	The lower bounds of $S_l$ are given as before. 
	Recall, for any $l \geq 1$,
	\begin{align*}
		& S_l \geq \left( \frac{l}{2} \right)^{q_\alpha}  ~~~~\text{if } q_\alpha \geq 0\\
		& S_l \geq l^{2q_\alpha} ~~~~~~~~~\text{if } q_\alpha < 0.
	\end{align*}
	
	Using the following well-known combinatorial fact
	\begin{align*}
		\sum_{k_1 + k_2 = l} \frac{l!}{k_1!k_2!} = 2^{l},
	\end{align*}
	we obtain
	\begin{align*}
		Q_{\Haf}^{\textnormal{MC}} & \geq {c_2^{-2} \mu_{\Haf}^2} \left( 1 + \sum_{l =1}^K S_l \bmin^{l} \gamma_\alpha^{l} \frac{(2l)!}{l!^2} \right) \\
		& \geq {c_2^{-2} \mu_{\Haf}^2} \left( 1 + e^{\frac{1}{25} - \frac{1}{6}}  \frac{1}{\sqrt{\pi}} \sum_{l =1}^K \frac{1}{\sqrt{l}} S_l (4\gamma_\alpha \bmin)^{l} \right). \\
		 & \hspace{1.6in} \text{(use Lemma \ref{lem:usefulstirling1})}
	\end{align*}
	If $ q_\alpha \geq 0$, then
	\begin{align*}
		Q_{\Haf}^{\textnormal{MC}}  &\geq  {c_2^{-2} \mu_{\Haf}^2}  \left( 1 + e^{\frac{1}{25} - \frac{1}{6}} 2^{-q_\alpha} \frac{1}{ \sqrt{\pi}} \multilog_{\frac{1}{2} - q_\alpha, K} (4 \gamma_\alpha \bmin) \right) \\
		& ={c_2^{-2} \mu_{\Haf}^2} \left( 1 + 2 e^{\frac{1}{25} - \frac{1}{6}} R_{q_\alpha, K} (4 \gamma_\alpha \bmin) \right).
	\end{align*}
	If $ q_\alpha < 0$, then
	\begin{align*}
		Q_{\Haf}^{\textnormal{MC}}  &\geq  {c_2^{-2} \mu_{\Haf}^2}  \left( 1 + e^{\frac{1}{25} - \frac{1}{6}} \frac{1}{\sqrt{\pi}} \multilog_{\frac{1}{2} - 2 q_\alpha, K} (4 \gamma_\alpha \bmin) \right)\\
		\\
		& ={c_2^{-2} \mu_{\Haf}^2} \left( 1 + 2 e^{\frac{1}{25} - \frac{1}{6}} R_{q_\alpha, K} (4 \gamma_\alpha \bmin) \right).
	\end{align*}
\end{proof}

We next prove that the conditions in Theorem \ref{thrm:mainI2} imply that $n^{\textnormal{GBS-P}}_{\Haf}$ is exponentially less than $n^{\textnormal{MC}}_{\Haf}$. 

\begin{lemma}
	\label{lem:firstpartI2}
	If the conditions in Theorem \ref{thrm:mainI2} are satisfied, then
	\begin{align}
		\exp(n_{\Haf}^\textnormal{GBS-P}) \leq n_{\Haf}^\textnormal{MC}.
		\label{eq:newcondfollow}
	\end{align}
\end{lemma}

\begin{proof}
	Plugging the bounds for $Q_{\Haf}^{\textnormal{GBS-P}}$ and $Q_{\Haf}^{\textnormal{MC}}$  from Lemma \ref{lem:gbspestimategsize} and \ref{lem:mcestimategsize2}  into \eqref{eq:consistency2}, we get 
	\begin{align*}
		& \frac{c_1^{-1}}{d}\frac{1}{\epsilon^2 \delta}  \left( 1 +  G_{q_\beta, K}\left(\sqrt{2\gamma_\beta \bmin^{-1}}, N \right) \right) \\
		& \leq \ln \left( 1 + 2 e^{\frac{1}{25} - \frac{1}{6}} R_{q_\alpha, K} (4 \gamma_\alpha \bmin) \right) + 2 \ln\left(c_2^{-1} \right) - \ln(\epsilon^2 \delta + 1) + \frac{1}{\epsilon^2 \delta}. 
	\end{align*}
	This is the same as \eqref{eq:thrmI2-3}. 
	Using Lemma \ref{lem:equivgoal}, 
	it is straightforward that \eqref{eq:newcondfollow} holds.
\end{proof}

We now give lemmas required to prove the second half of Theorem \ref{thrm:mainI2}. 

\begin{lemma}
	\label{lem:kinfBexsitsI2}
 	When $K = \infty$, if $q_\alpha \leq q_\beta$ and $\frac{1}{4}N < \gamma_\alpha = \gamma_\beta < \frac{\sqrt{2}}{4} N$, then $\mathcal{B}_{\alpha, \beta}$ is non-empty.
\end{lemma}

\begin{proof}
	The proof largely follows the proof of Lemma \ref{lem:kinfBexsits}. Under the specified conditions, \eqref{eq:cgamma2} clearly holds. 
	We now choose $0 < \bmin < \bmax < \frac{1}{N}$ so that \eqref{eq:thrmI2-3}, \eqref{eq:thrmI2-4} and \eqref{eq:thrmI2-5} hold.
	Let
	$$\bmax = \frac{\tau}{N}, ~~~~ \bmin = \xi \bmax$$
	with $ 0 < \tau < 1$ and $0 < \xi < 1$. 
	If
	\begin{align*}
	 \tau < \frac{N}{4\gamma_\beta} ~~~~\textnormal{and} ~~~~\xi^2 \tau^2 > \frac{1}{2}
	\end{align*}
	then
	\begin{gather*}
		4 \gamma_\beta \bmax = 4 \gamma_\beta \frac{\tau}{N} < 1, \\
		2 \gamma_\beta \bmin^{-1} = \frac{4\gamma_{\beta} \bmin}{2 \bmin^2 } < \frac{1}{2 \bmin^2} = \frac{1}{2} \frac{N^2}{\xi^2 \tau^2} < N^2,
	\end{gather*}
	and therefore \eqref{eq:thrmI2-4}, \eqref{eq:thrmI2-5} hold. Furthermore, under these conditions, $c_2$ are bounded above by a constant. Specifically
	\begin{align*}
		&c_2 < 1 + \frac{1}{\sqrt{\pi}} \multilog_{\frac{1}{2}-q_\beta, K}\left(\frac{1}{2}\right).
	\end{align*}
	
	Now we would like to take $\xi \tau$ sufficiently close to $\frac{N}{4{\gamma_\alpha}}$ so that
	\begin{align*}
		4 \gamma_\alpha \bmin = 4 \gamma_\alpha \xi \bmax =\xi \tau \frac{ 4 \gamma_\alpha }{N}
	\end{align*}
	becomes as close to 1 as we like and consequently
	\begin{gather*}
		\lim_{\xi \tau \rightarrow \frac{N}{4{\gamma_\alpha}}} \ln\left( 1 + 2 e^{\frac{1}{25} - \frac{1}{6}} R_{q_\alpha, K} (4 \gamma_\alpha \bmin) \right) = \infty.
	\end{gather*}
	This is only possible if $\gamma_\alpha = \gamma_\beta$. Since $\gamma_\beta < \frac{\sqrt{2}}{4} N$, 
	\begin{align*}
		\sqrt{2 \gamma_\beta \bmin^{-1} } = \sqrt{\frac{2\gamma_{\beta}N}{\xi \tau}} < \sqrt{\tfrac{1}{2}} \tfrac{N}{\sqrt{\xi \tau}}
	\end{align*}
	We choose $\xi \tau$ such that
	$$  {\frac{1}{2}} \left( \frac{N}{4{\gamma_\alpha}} + \frac{1}{2} \right)=  {\frac{1}{2}}  + \frac{1}{2} \left( \frac{N}{4{\gamma_\alpha}} - {\frac{1}{2}} \right) < \tau \xi < \frac{N}{4{\gamma_\alpha}}.$$
	Therefore, 
	\begin{align*}
		& G_{q_\beta, K, N}(\sqrt{2 \gamma_\beta \bmin^{-1} }) < G_{q_\beta, K, N} \left(\sqrt{\tfrac{1}{2}} \tfrac{N}{\sqrt{\xi \tau}} \right) < G_{q_\beta, K, N}\left( N\left(\tfrac{N}{4{\gamma_\alpha}} + {\tfrac{1}{2}} \right)^{-1} \right). 
	\end{align*}
	is bounded above by a constant. 
	By squeezing $\tau\xi$ towards $\frac{N}{4{\gamma_\alpha}}$, we ensure the the RHS of \eqref{eq:thrmI2-3} will eventually exceed the LHS. Finally, it follows from Lemma \ref{lem:bbexist} that there exists $B$ in $\mathcal{B}$ such that $\bmax = \frac{\tau}{N}$ and $\bmin = \xi \bmax$ with the chosen $\tau$ and $\xi$. 
\end{proof}

\begin{lemma}
	\label{lem:kfiniteBexistsI2}
	When $K < \infty$ is sufficiently large, if $q_\alpha \leq q_\beta$ and $ \frac{N}{4} < \gamma_\alpha \leq \gamma_\beta < \frac{N}{2}$, then $\mathcal{B}_{\alpha, \beta}$ is non-empty.
\end{lemma} 

\begin{proof}
	The proof is similar to the case when $K$ is infinite and Lemma \ref{lem:kfBexsits}. As before, we write 
	$\bmax = \frac{\tau}{N}, \bmin = \xi \bmax$ with $ 0 < \tau < 1$ and $0 < \xi < 1$. 
	Notice that with $ \frac{N}{4} < \gamma_\alpha \leq \gamma_\beta < \frac{N}{2}$, we can choose $\xi \tau$ sufficiently close to $1$ so that
	\begin{align*}
		\frac{N}{4\gamma_\alpha} < \xi \tau  ~~~~\textnormal{and}~~~~ \frac{2\gamma_{\beta}}{N} < \xi \tau 
	\end{align*}
	Therefore, 
	\begin{align*}
		4 \gamma_\alpha \bmin = \frac{4 \gamma_\alpha }{N} \xi \tau > 1, \\
		2 \gamma_\beta \bmin^{-1} = \frac{2\gamma_{\beta} N}{\xi \tau} < N^2,
	\end{align*}
	and consequently
	\begin{gather*}
		\lim_{K \rightarrow \infty} \ln\left( 1 + 2 e^{\frac{1}{25} - \frac{1}{6}} R_{q_\alpha, K} (4 \gamma_\alpha \bmin) \right) = \infty,\\
		\lim_{K \rightarrow \infty} G_{q_\beta, K, N}(\sqrt{2 \gamma_\beta \bmin^{-1}}) < c
	\end{gather*}
	for some constant $c$.
	Hence, for a large enough $K$, we obtain \eqref{eq:thrmI2-3}. The rest of the proof is parallel to that of  Lemma \ref{lem:kfBexsits} and is left to the readers. 
\end{proof}

We next show that for each $B \in \mathcal{B}_{\alpha, \beta}$, there exists an open and non-empty subset of $\mathcal{A}_B $ for which \eqref{eq:thrmI2-1} and \eqref{eq:thrmI2-2} hold. We define similarly as before
$\mu_{\Haf}: \mathcal{A}_B \rightarrow \mathbb{R}_+$ given by
\begin{align}
	\mu_{\Haf}(a_I) = \sum_{k = 0}^K \sum_{\vert I \vert =2k} a_I \Haf(B_I). \label{eq:contmapdefn2}
\end{align}
Let $\mathcal{A}_B^1 = \mu_{\Haf}^{-1}(1)$. We define $\pi: \mathcal{A}_B \rightarrow \mathcal{A}_B^1$ by
\begin{align*}
	\pi(a_I) = \left( \frac{a_I}{\mu_{\Haf}(a_I)}\right).
\end{align*}
Clearly, both $\mu_{\Haf}$ and $\pi$ are continuous.
We then define for each $B \in \mathcal{B}_{\alpha, \beta}$ the following subsets
\begin{align*}
	&U_{B, \alpha, \beta} = \left\{ (a_I) \in \mathcal{A}_B \mid c_2^{-1} < a_0 < c_1^{-1}, \right. \\
	&\hspace{1.5in} c_2^{-1}\mu_{\Haf} \frac{ k^{q_\alpha} \gamma_\alpha^{k}}{k!}  < \sum_{\vert I \vert = 2k} a_I < c_1^{-1}\mu_{\Haf} \frac{k^{q_\beta} \gamma_\beta^{k}}{k!}, ~~ k = 1, 2, \dots, K, \\
	&\hspace{1.4in} \sumck a_I I! < c_1^{-1}\mu_{\Haf}\frac{k^{q_\beta} \gamma_\beta^{k}}{k!} m_k, ~~ k = 1, 2, \dots, K \}.
\end{align*}
and
\begin{align*}
	&U^1_{B, \alpha, \beta} = \left\{ (a_I) \in \mathcal{A}^1_B \mid c_2^{-1} < a_0 < c_1^{-1}, \right. \\
	&\hspace{1.4in} c_2^{-1} \frac{ k^{q_\alpha} \gamma_\alpha^{k}}{k!}  < \sum_{\vert I \vert = 2k} a_I < c_1^{-1} \frac{k^{q_\beta} \gamma_\beta^{k}}{k!} ,  ~~ k = 1, 2, \dots, K, \\
	&\hspace{1.3in} \sumck a^2_I I! < c_1^{-1}\frac{k^{q_\beta} \gamma_\beta^{k}}{k!} m_k,  ~~ k = 1, 2, \dots, K \}.
\end{align*}
Clearly, $\pi^{-1}(U^1_{B, \alpha, \beta} ) = U_{B, \alpha, \beta}$. Furthermore, using the same techniques as in Lemma \ref{lem:secondpart}, we can show that $U_{B, \alpha, \beta} $ is open in $\mathcal{A}_B $. The next lemma proves that $U_{B, \alpha, \beta} $ is non-empty. 

\begin{lemma}
	\label{lem:secondpartP2}
		For any $B \in 	\mathcal{B}_{\alpha, \beta}$, $U_{B, \alpha, \beta} $ is non-empty.
\end{lemma}

\begin{proof}
	We begin by choosing $a_I$'s the same way as in \ref{lem:secondpart2}, replacing $(2k)!$ with $k!$. Then, we obtain with the chosen $a_I$'s that
	\begin{gather*}
		\sum_{I! = m_k} a_I = \frac{t}{k!} \gamma_{\alpha}^k  k^{q_\alpha} \\
		\sum_{I! = m_k} a_I I! 
		= \frac{t}{k!} \gamma_{\alpha}^k  k^{q_\alpha} m_k.
	\end{gather*}
	Therefore, the conditions in $\mathcal{B}_{\alpha, \beta}$ are clearly met.  

	Since
	\begin{align*}
		\sum_{\vert I \vert = 2k} a_I = t \frac{\gamma_{\alpha}^k  k^{q_\alpha}}{k!} \leq t \frac{\gamma_{\beta}^k  k^{q_\beta}}{k!}
	\end{align*}
	we get that
	\begin{align*}
		H((a_I)) 
		&\leq t \left( 1+  \sum_{k=1}^K \frac{\gamma_\beta^k k^{q_\beta}}{k!} \bmax^{k} \frac{(2k)!}{2^{k} k!} \right) \\
		& \leq  t \left( 1 + \sum_{k=1}^K 2^k\gamma_\beta^k k^{q_\beta} \bmax^{k} \frac{1}{\sqrt{\pi k}} \right) ~~~~~ \text{(use Lemma \ref{lem:usefulstirling1})}\\
		& = t \left( 1 + \frac{1}{\sqrt{\pi}} \multilog_{\frac{1}{2}-q_\beta, K}\left(2 \gamma_\beta \bmax \right) \right) \\
		& = t c_2. 
	\end{align*}
	Similarly,
	we compute that
	\begin{align*}
		H((a_I))  & \geq t \left( 1 +  \sum_{k=1}^K \frac{\gamma_\alpha^k k^{q_\alpha} }{k!} \bmin^{k} \frac{(2k)!}{2^{k} k!} \right) \\
		& \geq  t \left( 1 + \sum_{k=1}^K 2^k \gamma_\alpha^k  k^{q_\alpha} \bmin^{k} \frac{1}{\sqrt{\pi k}} e^{\frac{1}{25} - \frac{1}{6}} \right) \\
		& = t \left( 1 + \frac{1}{\sqrt{\pi}} e^{\frac{1}{25} - \frac{1}{6}} \multilog_{\frac{1}{2} - q_\alpha, K}\left(2 \gamma_\alpha \bmin \right) \right)\\
		& = t c_1. 
	\end{align*}
	The rest of the proof is the same as Lemma \ref{lem:secondpart2}. 
\end{proof}

The proof of Theorem \ref{thrm:mainI2} follows directly from the lemmas in this section.
\begin{proof}[Proof of Theorem \ref{thrm:mainI2}]
The first part of the theorem is proved in Lemma \ref{lem:firstpartI2}. To prove the second half of the theorem, it is clear from Lemma \ref{lem:kinfBexsitsI2} and Lemma \ref{lem:kfiniteBexistsI2} that for any $0 < \epsilon, \delta < 1$, $N$ and $K$ sufficiently large, there exists $\gamma_{\alpha}, \gamma_{\beta} \in \mathbb{R}_+$ and $q_\alpha, q_\beta \in \mathbb{R}$ such that $\mathcal{B}_{\alpha, \beta}$ is nonempty. Moreover, it is straightforward from the subspace topology criterion that $\mathcal{B}_{\alpha, \beta}$ is open in $\mathcal{B}$. Then, we use arguments completely analogous to Lemma \ref{lem:secondpart} and Lemma \ref{lem:secondpart2} to show that $U_{B, \alpha, \beta} $ is open and non-empty in $\mathcal{A}_B$, for each $B \in \mathcal{B}_{\alpha, \beta}$. Therefore, 
$$\mathcal{P}_{\alpha, \beta} = \bigcup_{B \in  \mathcal{B}_{\alpha, \beta}} U_{B, \alpha, \beta}$$ is open and non-empty in $\mathcal{P}$. This completes the proof of this theorem.
\end{proof}

\begin{proof}[Proof of Theorem \ref{cor:mainI2}]
	The proof follows the same logic as Theorem \ref{cor:mainI1} and is left to the reader.
\end{proof}

The proof of Theorem \ref{cor:mainI2N} parallels that of Theorem \ref{cor:mainI1N}. We only provide the following lemma, which proves there exists $q_\alpha \leq q_\beta$ and $\gamma_\alpha \leq \gamma_\beta$ such that $\mathcal{B}''_{\alpha, \beta}$ is non-empty. For convenience, let $q_\alpha = \nu_\alpha N$ and $q_\beta = \nu_\beta N$, and define $\gamma_{\alpha} = \tau_\alpha N$ and $\gamma_{\beta} = \tau_\beta N$. 

\begin{lemma}
	\label{lem:increasenpplem}
	When $N$ is sufficiently large, if $c \leq \nu_\alpha \leq \nu_\beta \leq -\frac{1}{2}$ for some constant $c$ and $\frac{1}{4} < \tau_{\alpha} \leq \tau_{\beta} < \frac{1}{2}$, then $\mathcal{B}''_{\alpha, \beta}$ is non-empty.
\end{lemma}

\begin{proof}
	Since $c \leq \nu_\alpha \leq \nu_\beta \leq -\frac{1}{2}$, we have $q_\alpha \leq q_\beta \leq -\frac{N}{2}$. It then follows from Lemmas \ref{lem:ggexist_increasen} and \ref{lem:rrexist_increasen} that a sufficient condition for \eqref{eq:uniformpp1} and \eqref{eq:uniformpp2} is
	\begin{align*}
		&1 < 4 \gamma_{\alpha} \bmin = 4 \tau_\alpha N \bmin, \\
		& 1 > \frac{\sqrt{2 \gamma_{\beta} \bmin^{-1}}}{N } = \sqrt{\frac{2 \tau_\beta \bmin^{-1}}{N}}.
	\end{align*}
	This is equivalent to 
	\begin{align*}
		\bmin > \frac{1}{4 \tau_{\alpha} N} ~~~~\textnormal{and}~~~~
		\bmin > \frac{2 \tau_{\beta}}{N}.
	\end{align*}
	Now we pick $b_1 = \frac{\tau_1}{N}$ such that
	\begin{align*}
		\max\left(\frac{1}{4 \tau_{\alpha}}, 2\tau_{\beta} \right) < \tau_1 < 1.
	\end{align*}
	We pick $b_2 = \frac{\tau_2}{N}$ such that
	$$\tau_2 = \tau_1 + \frac{N- N\tau_1}{2}.$$ 
	By construction $\tau_2 + (N-1)\tau_1 < N.$
	Finally, with the chosen $b_1, b_2$, we apply Lemma \ref{lem:bbexist_increasen} to obtain a $B \in \mathcal{B}$ so that \eqref{eq:uniformpp3} is satisfied with $\bmax = b_2$ and $\bmin = b_1$.
\end{proof}

\begin{proof}[Proof of Theorem \ref{cor:mainI2N}.]
From Lemmas \ref{lem:gbspestimategsize} and \ref{lem:mcestimategsize2}, it is clear that the conditions \eqref{eq:uniformpp1}, \eqref{eq:uniformpp2}, and \eqref{eq:uniformpp3} imply \eqref{eq:ccpoly2} and \eqref{eq:ccexp2}. Thus, the first part of the theorem is established. To prove the second part, Lemma \ref{lem:increasenpplem} guarantees the existence of $\gamma_{\alpha}, \gamma_{\beta} \in \mathbb{R}_+$ and $q_\alpha, q_\beta \in \mathbb{R}$ such that $\mathcal{B}''_{\alpha, \beta}$ is nonempty for any $0 < \epsilon, \delta < 1$, $p>0$ and $N$ sufficiently large. The rest of the proof proceeds similarly as that of Theorem \ref{cor:mainI1N}, and we thus omit it to avoid redundancy.
\end{proof}

	\section{Numerical examples}
In this section, we run numerical examples to illustrate the performance of the GBS estimators and the MC estimators. 
The GBS samples in this section are obtained via classical simulation. Precisely, for a given problem, $\mathcal{I}^\times_{\Haf}(\epsilon, \delta)$ or $\mathcal{I}^\times_{\Haf^2}(\epsilon, \delta)$, we define
\begin{align*}
	S = \bigcup_{k = 0}^K \mathcal{I}_k, ~~~~ \mathcal{I}_k = \{I \in  \mathbb{N}^N  \mid \vert I \vert = 2k \}.
\end{align*}
We then sample according to the following discrete probability distribution: 
\begin{align*}
	P(I) = 
	\begin{cases}
		p_I & \text{if } I \in S, \\
		1 - \displaystyle \sum_{I \in S} p_I & \text{if } I \notin S,
	\end{cases}
\end{align*}
where the values of $p_I$ are pre-computed according to the formula provided in \eqref{eq:gbshaf1}. Particularly, the matrix hafnians appeared in $p_I$ is computed using the implementation in \cite{gupt2019walrus}, based on \cite{bulmer2022boundary}.

\subsection{Comparison between GBS-I and MC}
The first example, based on the conditions in Theorem \ref{thrm:mainI1}, compares the efficiency of GBS-I and MC with fixed $N$ and $0 < \epsilon, \delta < 1$ while varying $K$. 

\begin{example}
	\label{eg:1}
	Let $N = 3$. Let $q_\alpha = q_\beta = \frac{1}{2}$ and $\gamma_{\alpha} = \gamma_{\beta} = \gamma = 8.1825$. Let $s_1, s_2 >0.$
	We choose the matrix $B$ to be
	$$
	B = \begin{bmatrix}
		0.3421 & 0.3364 & 0.3244 \\
		0.3364 & 0.3392 & 0.3225 \\
		0.3244 & 0.3225 & 0.3520 \\
	\end{bmatrix}.$$
	For any $K < \infty$, the coefficients $a_I$'s are generated as follows:
	Let $$a_0 = 1.$$
	For $ 1 \leq k \leq K$ and $I \in \mathcal{I}_k$, we set the coefficients to be
	$$ a_I = \begin{cases}
		0, & \text{if } I! \neq m_k \\
		\frac{k^{\frac{1}{2}} \gamma^{k}}{\vert \mathcal{I}_k \vert (2k)!}, & \text{if } I! = m_k
	\end{cases}. $$
\end{example}

We first check that $B \in \mathcal{B}_{\alpha, \beta}$. Since the eigenvalues of $B$ are $\lambda_1 = 0.9999$, $\lambda_2 = 0.029$, $\lambda_3 = 0.0042$, we obtain that $B \in \mathcal{B}$. Furthermore, $$\frac{1}{d} = 223.7037.$$ 
Since $\gamma \bmin^2 = 0.8511 < 1$ and $\gamma \bmax^2 = 1.013 > 1$, we know that $c_1 < c_2$ for all $K$. Therefore, \eqref{eq:cgamma} is clearly satisfied. For \eqref{eq:thrmI1-3}, since
\begin{align*}
	&\gamma \bmax = 2.88 < 3 \\
	&4 \gamma \bmin^2 = 3.4044 > 1
\end{align*}
the left hand side is bounded above as $K$ grows while the right hand side blows up. Therefore, there is a $K$ large enough such that \eqref{eq:thrmI1-3} holds. Therefore, $B \in \mathcal{B}_{\alpha, \beta}$.

We then check that the $a_I$'s satisfy the conditions in Theorem \ref{thrm:mainI1}. By definition,
\begin{align*}
	\sum_{\vert I \vert = 2k} a_I &= \frac{k^{\frac{1}{2}} \gamma^{k}}{(2k)!}, \\
	\sumck a^2_I I! &=  \frac{k \gamma^{2k}}{(2k)!^2} m_k.
\end{align*}
We can further compute that
\begin{align*}
	&\mu_{\Haf^2} \leq 1 + \frac{1}{\sqrt{\pi}} \multilog_{0, K}\left(\gamma \bmax^{2} \right) = c_2, \\
	&\mu_{\Haf^2} \geq 1 + \frac{1}{\sqrt{\pi}} e^{\frac{1}{25} - \frac{1}{6}} \multilog_{0, K}\left(\gamma_\alpha \bmin^{2} \right) = c_1. \\
\end{align*}
Therefore, 
\begin{align*}
	\frac{\mu_{\Haf^2}}{c_2} \leq 1 ~~~~\text{and}~~~~\frac{\mu_{\Haf^2}}{c_1} \geq 1.
\end{align*}
From here, it is straightforward that all conditions in Theorem \ref{thrm:mainI1} are met.
Therefore, Theorem \ref{thrm:mainI1} implies that for there is a sufficiently large $K$ such that $\exp(n_{\Haf^2}^{\textnormal{GBS-I}}) < n_{\Haf^2}^{\textnormal{MC}}$. 

Now we give numerical approximations to $n_{\Haf^2}^{\textnormal{GBS-I}}$ and $n_{\Haf^2}^{\textnormal{MC}}$.
Recall that
\begin{align*}
	n_{\Haf^2}^{\textnormal{GBS-I}} = \frac{1}{\epsilon^2 \delta} \left( \frac{Q_{\Haf^2}^{\textnormal{GBS-I}}}{\mu^2_{\Haf^2}} -1 \right) \\
	n_{\Haf^2}^{\textnormal{MC}} = \frac{1}{\epsilon^2 \delta} \left( \frac{Q_{\Haf^2}^{\textnormal{MC}}}{\mu^2_{\Haf^2}} -1 \right) 
\end{align*}
We note the following bounds:
\begin{align*}
	&\frac{Q_{\Haf^2}^{\textnormal{GBS-I}}}{\mu^2_{\Haf^2}} \leq \frac{1}{c_1^{2}d}  \left( 1 + \frac{1}{\sqrt{\pi}} G_{2 q_{\beta}, K, N}(\gamma_\beta \bmax) \right) =: U_{\Haf^2}^{\textnormal{GBS-I}}\\
	&\frac{Q_{\Haf^2}^{\textnormal{MC}}}{\mu^2_{\Haf^2}} \geq \frac{1}{ c_2^{2} } \left( 1 + e^{\frac{1}{25} - \frac{1}{6}}  R_{q_\alpha, K} (4 \gamma_\alpha \bmin^{2}) \right) =: L_{\Haf^2}^{\textnormal{MC}}
\end{align*}
Table \ref{tab:gt1} provides these bounds for various $K$. 
It is evident that $L_{\Haf^2}^{\textnormal{MC}}$ grows unboundedly whereas $U_{\Haf^2}^{\textnormal{GBS-I}}$ plateaus. Specifically, using the explicit expressions of polylogs with integer orders, we find
\begin{align*}
	U_{\Haf^2}^{\textnormal{GBS-I}} < 6.215693 \textnormal{e+06}.
\end{align*}
At $K = 20$, we observe $U_{\Haf^2}^{\textnormal{GBS-I}}$ is less than $L_{\Haf^2}^{\textnormal{MC}}$, providing a theoretical guarantee of GBS-I's advantage. To estimate how large $K$ must be for an exponential advantage, we compute the leading order of large $K$ asymptotics of $L_{\Haf^2}^{\textnormal{MC}}$
\begin{align*}
	\ln\left( L_{\Haf^2}^{\textnormal{MC}}  \right)  & \sim \ln(\frac{\multilog_{0, K}(4 \gamma \bmin^{2})}{\left( \multilog_{0, K}(\gamma \bmax^2) \right)^2}) \\
	& \sim  \ln(\frac{(4 \gamma\bmin^{2})^K}{\left( \gamma \bmax^2)\right)^{2K}}) \\
	& \sim 3.3130 K.
\end{align*}
From this simple estimation, we obtain an upper bound on $K$ for GBS-I to achieve an exponential speedup over MC defined as in Theorem \ref{thrm:mainI1}. For example, if $c_1 = 10^{-4}$ and $c_2 = 1$, then $K \approx 2000.$

Note that, in practice, the computational advantage of GBS-I over MC can be observed for much smaller $K$. Figure \ref{fig:conv1} shows the convergence behavior of GBS-I and MC via numerical simulations. For $K = 5$, GBS-I performs similarly to MC.
For $K = 10, 15, 20$, GBS-I converges significantly faster than MC, while MC fails to converge even after one billion samples. This highlights the practical utility of GBS-I.

\begin{table}[h!]
	\centering
	\begin{tabular}{|c|c|c|c|}
		\hline
		$K$ & $\mu_{\Haf^2}$  & $U_{\Haf^2}^{\textnormal{GBS-I}}$ & $L_{\Haf^2}^{\textnormal{MC}}$ \\ 
		\hline
		5&2.946761 & 6.6616e+04 & 7.3824e+00 \\
		10&4.209808 & 1.4674e+05 & 1.0367e+03 \\
		15&4.999906 & 2.5913e+05 & 2.1783e+05 \\
		20&5.492133 & 4.5203e+05 & 5.5079e+07 \\
		25&5.798267 & 6.0213e+05 & 1.5508e+10 \\
		30&5.988598 & 7.5217e+05 & 4.6836e+12 \\
		35&6.106900 & 9.3946e+05 & 1.4861e+15 \\
		40&6.180401 & 1.0619e+06 & 4.8911e+17 \\
		45&6.226069 & 1.1683e+06 & 1.6556e+20 \\
		50&6.254442 & 1.2853e+06 & 5.7299e+22 \\
		\hline
	\end{tabular}
	\caption{Values of $\mu_{\Haf^2}$, $U_{\Haf^2}^{\textnormal{GBS-I}}$, $L_{\Haf^2}^{\textnormal{MC}}$ for different $K$ in Example \ref{eg:1}.}
	\label{tab:gt1}
\end{table}

\begin{figure}
	\centering
	\vspace{-1in}
	\begin{subfigure}{0.5\linewidth}
		\centering
		\caption{K = 5}
		\includegraphics[width=\linewidth]{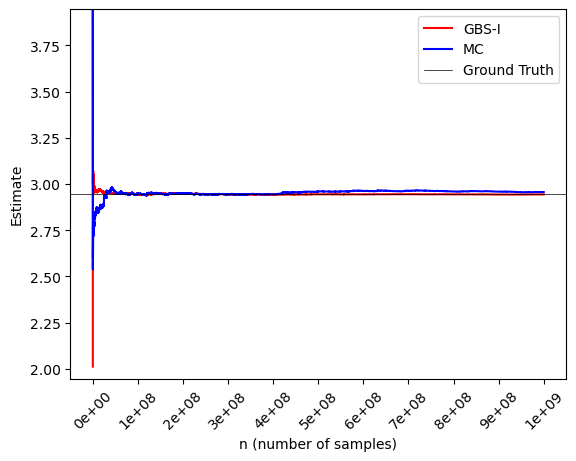}
	\end{subfigure}
	
	\begin{subfigure}{0.5\linewidth}
		\centering
		\caption{K = 10}
		\includegraphics[width=\linewidth]{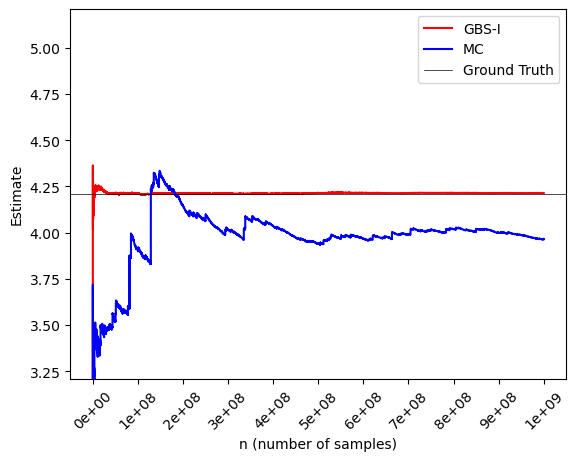}
	\end{subfigure}
	
	\begin{subfigure}{0.5\linewidth}
		\centering
		\caption{K = 15}
		\includegraphics[width=\linewidth]{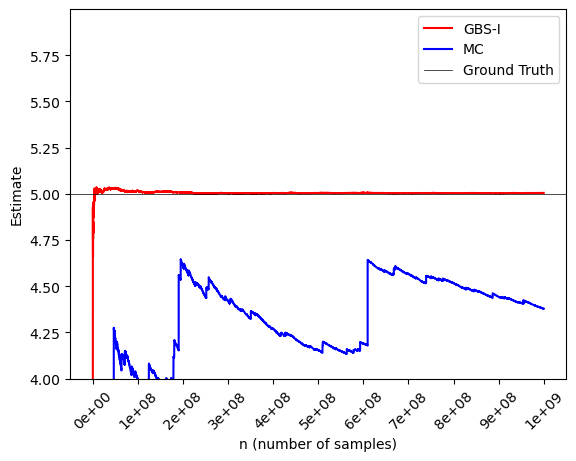}
	\end{subfigure}
	
	\begin{subfigure}{0.5\linewidth}
		\centering
		\caption{K = 20}
		\includegraphics[width=\linewidth]{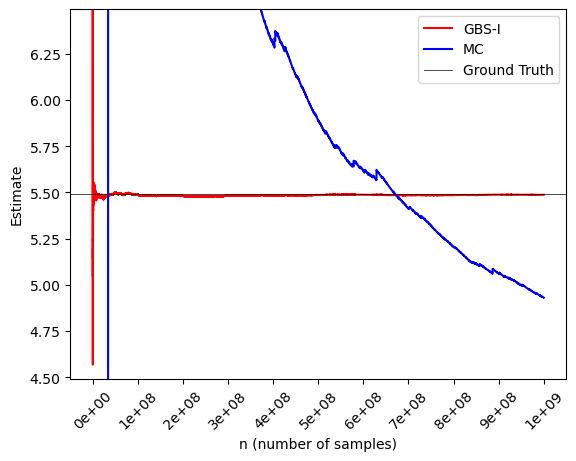}
	\end{subfigure}
	
	\caption{Convergence behavior of GBS-I and MC for solving Example \ref{eg:1}.}
	\label{fig:conv1}
\end{figure}

In the second example, we compare the performance of GBS-I and MC as $N$ grows.
\begin{example}
	\label{eg:1N}
	Let $K = N^2$. 
	Let $q_\alpha = q_\beta = q = -\frac{N}{4}$. Let $\gamma_\alpha = \frac{10}{13}N^2$ and $\gamma_\beta = \frac{20}{23} N^2$. We choose
	\begin{align*}
		b_2 = \frac{1.1}{N} ~~~~\text{and}~~~~
		b_1 = \frac{0.8}{N}.
	\end{align*}
	We can easily verify that for all $N \geq 2$,
	$$(N-1)b_1 + b_2 = 0.8 + \frac{0.3}{N} < 1.$$
	Using Lemma \ref{lem:bbexist_increasen}, we construct $B \in \mathcal{B}$ such that $\bmax = b_2$ and $\bmin = b_1$, and the term $\frac{1}{d}$ grows only polynomially in $N$. This procedure results in a matrix $B$ with entries equal to $\bmin$ everywhere except on the diagonal. To introduce more variation, we add symmetric positive noise, drawn from a normal distribution, and scaled it by $\bmax - \bmin$ to ensure that the maximum entry of $B$ remains $\bmax$ on the diagonal. Additionally, we keep $B_{1, N} = B_{N,1} = \bmin$ unchanged to ensure the minimum value of $B$ is $\bmin$. 
	The coefficients $a_I$'s are generated as follows.
	Let $$a_0 = 1.$$
	For $ 1 \leq k \leq K$ and $I \in \mathcal{I}_k$, we set the coefficients to be
	$$ a_I = \begin{cases}
		0, & \text{if } I! \neq m_k \\
		\frac{k^{q} \gamma^{k}}{\vert \mathcal{I}_k \vert (2k)!}, & \text{if } I! = m_k
	\end{cases}$$
	where
	\begin{align*}
		\gamma = \frac{1}{2} (\gamma_{\alpha} + \gamma_{\beta}).
	\end{align*}
\end{example}

We can verify numerically that after adding the noise, the spectrum of $B$ is strictly bounded between 0 and 1. Furthermore, $1/d$ is bounded above by $N^3$. 
Since
\begin{align*}
	4 \gamma_{\alpha} \bmin^2 = \frac{128}{65} >1\\
	\gamma_{\beta} \bmax = \frac{22}{23}N < N,
\end{align*}
the conditions \eqref{eq:uniform1} and \eqref{eq:uniform3} are clearly true using the proof of Lemma \ref{lem:mainI1increasenBexists}. Therefore, the noisy $B$'s we constructed belong to $\mathcal{B}''_{\alpha, \beta}.$
We can also easily verify that the $a_I$'s satisfy the conditions in 
Theorem \ref{cor:mainI1N}.

We now compare the values of $\frac{Q_{\Haf^2}^{\textnormal{GBS-I}}}{\mu^2_{\Haf^2}}$ and $\frac{Q_{\Haf^2}^{\textnormal{MC}}}{\mu^2_{\Haf^2}}$. We note the following bounds:
\begin{align*}
	&\frac{Q_{\Haf^2}^{\textnormal{GBS-I}}}{\mu^2_{\Haf^2}} \leq \frac{1}{d}  \left( 1 + \frac{1}{\sqrt{\pi}} G_{2 q, K, N}(\gamma_\beta \bmax) \right) =: U_{\Haf^2}^{'\textnormal{GBS-I}}\\
	&\frac{Q_{\Haf^2}^{\textnormal{MC}}}{\mu^2_{\Haf^2}} \geq \frac{1}{U_{\Haf^2}^{'\textnormal{GBS-I}} } \left( 1 + e^{\frac{1}{25} - \frac{1}{6}}  R_{q, K} (4 \gamma_\alpha \bmin^{2}) \right) =: L_{\Haf^2}^{'\textnormal{MC}}
\end{align*}
Here, we have used $a_0^2 = 1$ as a lower bound for $\mu^2_{\Haf^2}$ and $U_{\Haf^2}^{'\textnormal{GBS-I}}$ as an upper bound for $\mu^2_{\Haf^2}$, since $Q_{\Haf^2}^{\textnormal{GBS-I}} \geq \mu^2_{\Haf^2}$. Alternatively, one could use the standard polylog bounds given by
\begin{align*}
	&\mu_{\Haf^2} \leq 1 + \frac{1}{\sqrt{\pi}} \multilog_{\frac{1}{2} - q, K}\left(\gamma_\beta \bmax^{2} \right) = u_2, \\
	&\mu_{\Haf^2} \geq 1 + \frac{1}{\sqrt{\pi}} e^{\frac{1}{25} - \frac{1}{6}} \multilog_{\frac{1}{2} - q, K}\left(\gamma_\alpha \bmin^{2} \right) = u_1.
\end{align*}
However, since 
$$\gamma_\beta \bmax^2 = \frac{121}{115} >1,$$
and
\begin{align*}
	\multilog_{\frac{1}{2} - q, K}\left(\gamma_\beta \bmax^{2} \right) \geq \frac{\left(\gamma_\beta \bmax^2 \right)^{N^2}}{N^{1+ \frac{N}{2}}}
\end{align*}
we know that $u_1$ will eventually exceed $Q_{\Haf^2}^{\textnormal{GBS-I}}$, making it a weaker bound.
The table below provides the values of $u_1$, $u_2$, $U_{\Haf^2}^{'\textnormal{GBS-I}}$ and $L_{\Haf^2}^{'\textnormal{MC}}$. We also compare in Figure \ref{fig:gc-hafsq} their growth rate. Notably, $U_{\Haf^2}^{'\textnormal{GBS-I}}$ is bounded above by $N^3$ for all $N \geq 2$, while $L_{\Haf^2}^{'\textnormal{MC}}$ is bounded below by $1.3^{N^2}$ for all $N \geq 7$. 

\begin{figure}
	\centering
	\includegraphics[width=0.9\linewidth]{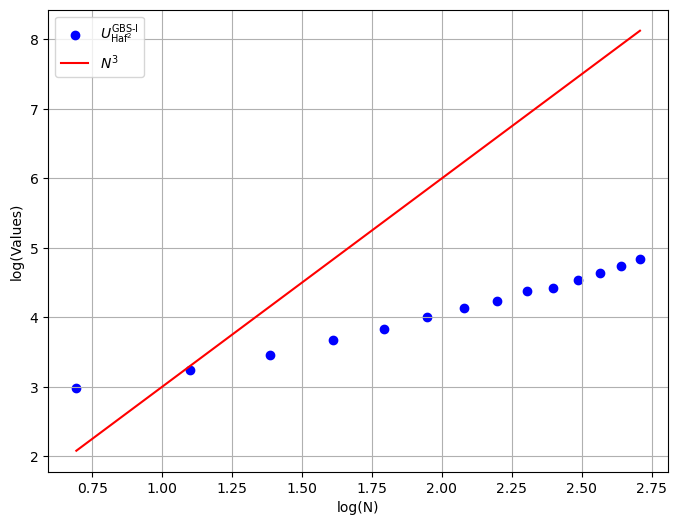}
	\includegraphics[width=0.9\linewidth]{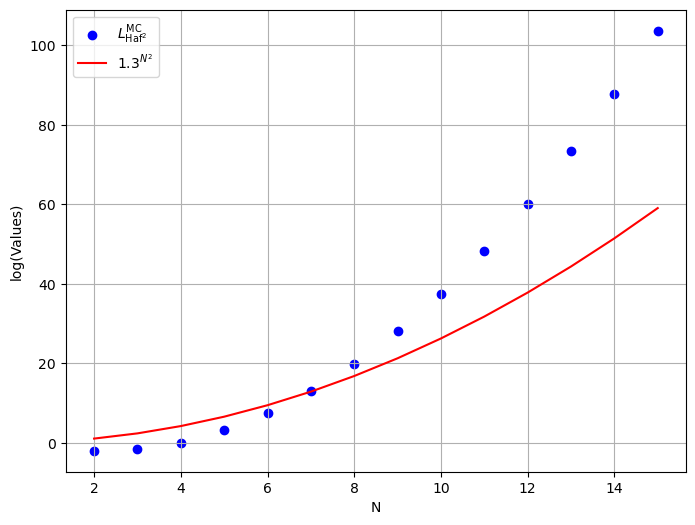}
	\caption{Growth rate comparison. Top: the log-log plot of $U_{\Haf^2}^{'\textnormal{GBS-I}}$ and $N^3$. Bottom: the semi-log plot of $L_{\Haf^2}^{'\textnormal{MC}}$ and $1.3^{N^2}$.}
	\label{fig:gc-hafsq}
\end{figure}

\begin{table}[h!]
	\centering
	\begin{tabular}{|c|c|c|c|c|}
		\hline
		$N$ & $u_1$ & $u_2$ & $U_{\Haf^2}^{'\textnormal{GBS-I}}$ & $L_{\Haf^2}^{'\textnormal{MC}}$\\ 
		\hline
		2&1.332015 &2.297854 &  1.9714e+01 & 1.3506e-01\\
		3&1.318709 &2.514682 &  2.5653e+01 & 2.2215e-01\\
		4&1.304452 &2.426189 &  3.1803e+01 & 1.0581e+00\\
		5&1.293139 &2.243492 &  3.9391e+01 & 2.1886e+01\\
		6&1.284130 &2.070371 &  4.6099e+01 & 1.7321e+03\\
		7&1.276913 &1.937098 &  5.4716e+01 & 4.6393e+05\\
		8&1.271103 &1.842611 &  6.2136e+01 & 4.4880e+08\\
		9&1.266403 &1.777380 &  6.9083e+01 & 1.5477e+12\\
		10&1.262585 &1.732178 &  7.9387e+01 & 1.8186e+16\\
		11&1.259472 &1.700232 &  8.3686e+01 & 8.2703e+20\\
		12&1.256925 &1.677064 &  9.3362e+01 & 1.2703e+26\\
		13&1.254835 &1.659822 &  1.0323e+02 & 7.0783e+31\\
		14&1.253116 &1.646698 &  1.1346e+02 & 1.4342e+38\\
		15&1.251698 &1.636524 &  1.2711e+02 & 1.0352e+45\\
		\hline
	\end{tabular}
	\caption{Comparison of $u_1$, $u_2$, $U_{\Haf^2}^{'\textnormal{GBS-I}}$ and $L_{\Haf^2}^{'\textnormal{MC}}$ for different $K$.}
	\label{tab:comparison}
\end{table}

	\subsection{Comparison between GBS-P and MC}
In the next example, we compare GBS-P and MC with $N$ and $0 < \epsilon, \delta < 1$ fixed while varying $K$. 

\begin{example}
	\label{eg:2}
	Let $N = 3$. Let $q_\alpha = q_\beta = \frac{1}{2}$ and $\gamma_{\alpha} = \gamma_{\beta} = \gamma = 1.4368$. Let $B$ be the same as in Example \ref{eg:1}. 
	For any $K < \infty$, the $a_I$'s selected as follows: Let
	$$ a_0 = 1. $$
	For $1 \leq k \leq K$ and $I \in \mathcal{I}_k$, we set
	\begin{align*}
		a_I = \begin{cases}
			0, & \text{if } I! \neq m_k \\
			\frac{k^{\frac{1}{2}} \gamma^{k}}{\vert \mathcal{I}_k \vert k!}, & \text{if } I! = m_k
		\end{cases}
	\end{align*}
\end{example}

We can easily verify $B \in \mathcal{B}_{\alpha, \beta}$ and the $a_I$'s satisfy the conditions in Theorem \ref{thrm:mainI2}. Therefore, for $K$ large enough, we obtain $\exp(n_{\Haf}^{\textnormal{GBS-P}}) < n_{\Haf}^{\textnormal{MC}}$. 
Similar to Example \ref{eg:1}, we give approximations to $n_{\Haf}^{\textnormal{GBS-P}}$ and $n_{\Haf}^{\textnormal{MC}}$ by comparing $\frac{Q_{\Haf}^{\textnormal{GBS-P}}}{\mu_{\Haf}} $ and $\frac{Q_{\Haf}^{\textnormal{MC}}}{\mu_{\Haf}} $. We note the following bounds:
\begin{align*}
	&\frac{Q_{\Haf}^{\textnormal{GBS-P}}}{\mu_{\Haf}} \leq \frac{1}{c_1^{2}d}  \left( 1 + 2 e^{\frac{1}{25} - \frac{1}{6}} R_{q_\alpha, K} (4 \gamma_\alpha \bmin)  \right) =: U_{\Haf}^{\textnormal{GBS-P}} \\
	&\frac{Q_{\Haf}^{\textnormal{MC}}}{\mu_{\Haf}} \geq \frac{1}{ c_2^{2} } \left( 1 + 2 e^{\frac{1}{25} - \frac{1}{6}}  R_{q_\alpha, K} (4 \gamma_\alpha \bmin) \right) =: L_{\Haf}^{\textnormal{MC}}.
\end{align*}
In Table \ref{tab:t2}, we present the values of $\mu_{\text{haf}}$, $U_{\Haf}^{\textnormal{GBS-P}}$ and $L_{\Haf}^{\textnormal{MC}}$ for different $K$. As $K$ increases, $L_{\Haf}^{\textnormal{MC}}$ grows unboundedly, whereas $U_{\Haf}^{\textnormal{GBS-P}}$ plateaus. 
At $K = 35$, we observe $U_{\Haf}^{\textnormal{GBS-P}}  < L_{\Haf^2}^{\textnormal{MC}}$, and thus GBS-P is guaranteed to outperform MC. 
In Figure \ref{fig:conv2}, we plot the convergence behavior of GBS-P and MC via numerical simulation. Notice that GBS-P often outperforms MC at much smaller values of $K$. For $K = 5$, GBS-P and MC are similar. However, for $K = 10, 15$ and $20$, GBS-P converges significantly faster than MC, and at $K=15$ and $K=20$, MC fails to converge even after one billion samples. 

\begin{table}[h!]
	\centering
	\begin{tabular}{|c|c|c|c|}
		\hline
		$K$ & $\mu_{\text{haf}}$ &  $U_{\Haf}^{\textnormal{GBS-P}}$ &$L_{\Haf}^{\textnormal{MC}}$ \\ 
		\hline
		5&3.305015 & 2.5167e+04 &1.1026e+00 \\
		10&5.241919 & 4.1425e+04 &7.4417e+00 \\
		15&6.814145 & 6.3214e+04 &7.5646e+01 \\
		20&8.085469 & 1.0820e+05 &9.2615e+02 \\
		25&9.111993 & 1.4557e+05 &1.2629e+04 \\
		30&9.940470 & 1.9117e+05 &1.8478e+05 \\
		35&10.608927 & 2.6950e+05 &2.8411e+06 \\
		40&11.148116 & 3.3487e+05 &4.5319e+07 \\
		45&11.583009 & 4.0922e+05 &7.4366e+08 \\
		50&11.933760 & 5.2450e+05 &1.2480e+10 \\
		\hline
	\end{tabular}
	\caption{Comparison of $\mu_{\text{haf}}$, $U_{\Haf}^{\textnormal{GBS-P}}$, $L_{\Haf}^{\textnormal{MC}}$ for different $K$ in Example \ref{eg:2}.}
	\label{tab:t2}
\end{table}

\begin{figure}
	\centering
	\vspace{-1in}
	
	\begin{subfigure}{0.5\linewidth}
		\centering
		\caption{K = 5}
		\includegraphics[width=\linewidth]{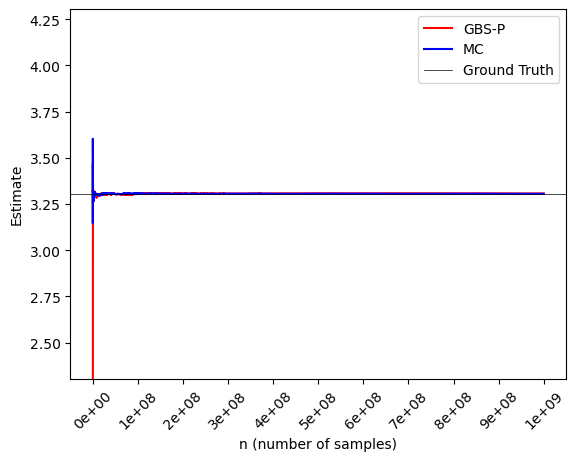}
	\end{subfigure}
	
	\begin{subfigure}{0.5\linewidth}
		\centering
		\caption{K = 10}
		\includegraphics[width=\linewidth]{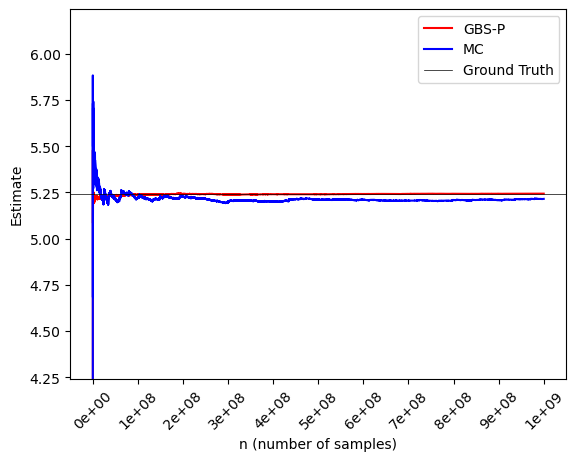}
	\end{subfigure}
	
	\begin{subfigure}{0.5\linewidth}
		\centering
		\caption{K = 15}
		\includegraphics[width=\linewidth]{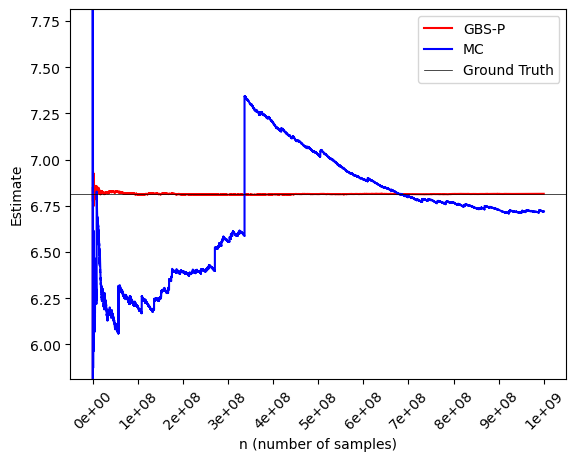}
	\end{subfigure}
	
	\begin{subfigure}{0.5\linewidth}
		\centering
		\caption{K = 20}
		\includegraphics[width=\linewidth]{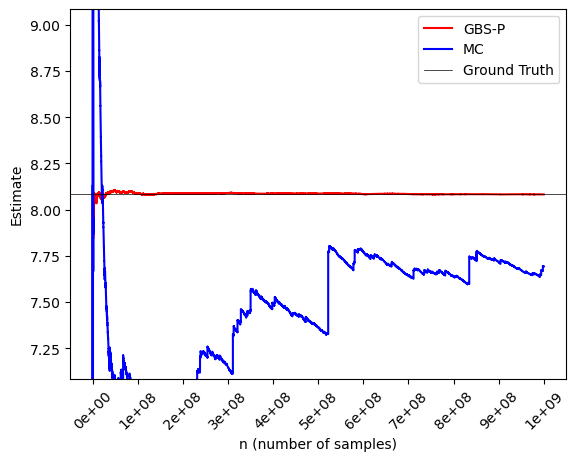}
	\end{subfigure}
	
	\caption{Convergence behavior of GBS-P and MC for solving Example \ref{eg:2}.}
	\label{fig:conv2}
\end{figure}

In the last example, we compare GBS-P and MC as $N$ grows. 
\begin{example}
	Let $K = N^2$. Let $q_\alpha = q_\beta = q= -\frac{N}{2}$. Let $\gamma_\alpha = \frac{10}{21}N$ and $\gamma_\beta = \frac{5}{11} N$.
	We set 
	\begin{align*}
		b_2 = \frac{1.11}{N} ~~~~\text{and}~~~~
		b_1 = \frac{0.96}{N}.
	\end{align*}
	We can easily verify that for all $N \geq 4$,
	$$(N-1)b_1 + b_2 = 0.96 + \frac{0.15}{N} < 1.$$
	Then, using the same approach as in Example \ref{eg:1N}, we construct a matrix $B$ with added symmetric positive noise such that $\bmax = b_2$ and $\bmin = b_1$. The $a_I$'s are generated as follows. Let
	$$ a_0 = 1. $$
	For $1 \leq k \leq K$ and $I \in \mathcal{I}_k$, we set
	\begin{align*}
		a_I = \begin{cases}
			0, & \text{if } I! \neq m_k \\
			\frac{k^{q} \gamma^{k}}{\vert \mathcal{I}_k \vert k!}, & \text{if } I! = m_k
		\end{cases}
	\end{align*}
	where 
	\begin{align*}
		\gamma = \frac{1}{2}\left(\gamma_{\alpha} + \gamma_{\beta}\right).
	\end{align*}
\end{example}

We can verify numerically that the constructed matrices $B$ with added noise have eigenvalues bounded strictly between 0 and 1. Furthermore, the term $\frac{1}{d}$ is bounded above by $N^3$. 
Since
\begin{align*}
	4 \gamma_{\alpha} \bmin &= \frac{64}{35} >1\\
	\sqrt{2 \gamma_\beta \bmin^{-1}} &= \sqrt{\frac{125}{132}}N < N
\end{align*}
it follows straightforwardly from the proof of Lemma \ref{lem:increasenpplem} that \eqref{eq:uniformpp1} and \eqref{eq:uniformpp2} hold. Therefore, $B \in \mathcal{B}_{\alpha, \beta}$. We can also easily verify that the $a_I$'s satisfy the conditions in 
Theorem \ref{cor:mainI2N}. 

We now compare the values of $\frac{Q_{\Haf}^{\textnormal{GBS-P}}}{\mu^2_{\Haf}}$ and $\frac{Q_{\Haf}^{\textnormal{MC}}}{\mu^2_{\Haf}}$. We note the following bounds:
\begin{align*}
	&\frac{Q_{\Haf}^{\textnormal{GBS-P}}}{\mu^2_{\Haf}} \leq \frac{1}{a_0^2 d}  \left( 1 + \frac{1}{\sqrt{\pi}} G_{q, K, N}\left(\sqrt{2 \gamma_\beta \bmin^{-1}}\right) \right) =: U_{\Haf}^{'\textnormal{GBS-P}}\\
	&\frac{Q_{\Haf}^{\textnormal{MC}}}{\mu^2_{\Haf}} \geq \frac{1}{U_{\Haf^2}^{'\textnormal{GBS-P}} } \left( 1 + 2 e^{\frac{1}{25} - \frac{1}{6}}  R_{q, K} (4 \gamma_\alpha \bmin) \right) =: L_{\Haf}^{'\textnormal{MC}}
\end{align*}
As in Example \ref{eg:1N}, we have used $a_0^2 = 1$ as a lower bound for $\mu^2_{\Haf}$ and $U_{\Haf}^{'\textnormal{GBS-P}}$ as an upper bound for $\mu^2_{\Haf^2}$, since $Q_{\Haf}^{\textnormal{GBS-P}} \geq \mu^2_{\Haf^2}$. Alternatively, one could use the standard polylog bounds given by
\begin{align*}
	&\mu_{\Haf} \leq 1 + \frac{1}{\sqrt{\pi}} \multilog_{\frac{1}{2} - q, K}\left(2 \gamma_\beta \bmax \right) = u_2, \\
	&\mu_{\Haf} \geq 1 + \frac{1}{\sqrt{\pi}} e^{\frac{1}{25} - \frac{1}{6}} \multilog_{\frac{1}{2} - q, K}\left(2 \gamma_\alpha \bmin \right) = u_1.
\end{align*}
However, since 
$$2 \gamma_\beta \bmax = \frac{111}{110} > 1$$
we deduce that $u_1$ will eventually exceed $U_{\Haf}^{'\textnormal{GBS-P}}$, making it a worse bound. 
In Table \ref{tab:comparison}, we compare the values of $U_{\Haf}^{'\textnormal{GBS-P}}$ and $L_{\Haf}^{'\textnormal{MC}}$. We observe in Figure \ref{fig:gc-haf} that $U_{\Haf}^{'\textnormal{GBS-P}}$ is bounded above by $N^3$ for all $N \geq 7$, while $L_{\Haf}^{'\textnormal{MC}}$ is bounded below by $1.25^{N^2}$ for all $N \geq 16$. 

\begin{table}[h!]
	\centering
	\begin{tabular}{|c|c|c|c|c|}
		\hline
		$N$ & $u_1$ & $u_2$ & $U_{\Haf}^{'\textnormal{GBS-P}}$ & $L_{\Haf}^{'\textnormal{MC}}$\\ 
		\hline
		5&1.532053 &1.686043 &  4.6089e+02 & 4.5189e-03\\
		6&1.504478 &1.642551 &  5.0178e+02 & 4.4674e-03\\
		7&1.487472 &1.616808 &  4.8000e+02 & 8.2543e-03\\
		8&1.476620 &1.600839 &  4.5692e+02 & 8.1501e-02\\
		9&1.469519 &1.590578 &  4.6146e+02 & 3.5479e+00\\
		10&1.464786 &1.583821 &  4.6120e+02 & 4.4804e+02\\
		11&1.461590 &1.579294 &  4.4327e+02 & 1.6361e+05\\
		12&1.459408 &1.576222 &  4.2257e+02 & 1.6983e+08\\
		13&1.457907 &1.574119 &  4.0907e+02 & 4.9588e+11\\
		14&1.456869 &1.572668 &  4.0040e+02 & 4.1365e+15\\
		15&1.456148 &1.571662 &  4.1662e+02 & 9.4681e+19\\
		16&1.455645 &1.570962 &  4.1552e+02 & 6.6523e+24\\
		17&1.455293 &1.570473 &  4.2021e+02 & 1.3665e+30\\
		18&1.455046 &1.570130 &  4.4962e+02 & 7.9191e+35\\
		19&1.454873 &1.569890 &  4.8091e+02 & 1.3786e+42\\
		20&1.454752 &1.569721 &  5.1791e+02 & 7.1962e+48\\
		21&1.454666 &1.569603 &  5.6590e+02 & 1.1234e+56\\
		22&1.454605 &1.569519 &  5.9574e+02 & 5.5473e+63\\
		23&1.454563 &1.569460 &  6.4661e+02 & 8.1308e+71\\
		24&1.454533 &1.569418 &  6.8751e+02 & 3.7369e+80\\
		25&1.454512 &1.569389 &  7.3139e+02 & 5.2909e+89\\
		\hline
	\end{tabular}
	\caption{Comparison of $u_1$, $u_2$, $U_{\Haf}^{'\textnormal{GBS-P}}$ and $L_{\Haf}^{'\textnormal{MC}}$ for different $N$.}
	\label{tab:comparison}
\end{table}

\begin{figure}
\centering
\includegraphics[width=0.9\linewidth]{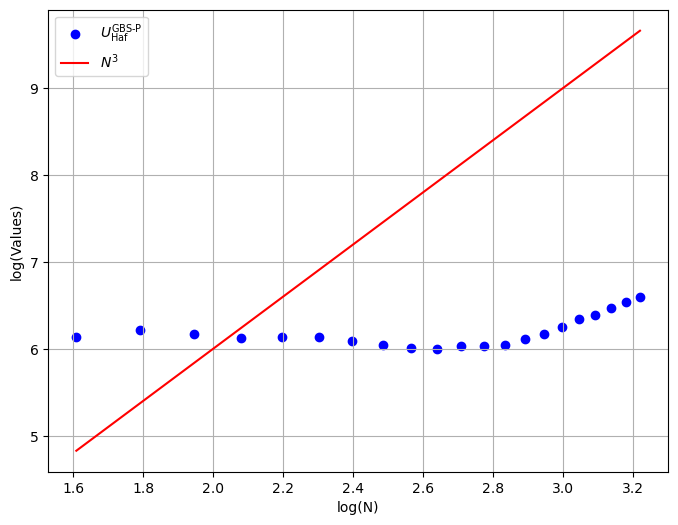}
\includegraphics[width=0.9\linewidth]{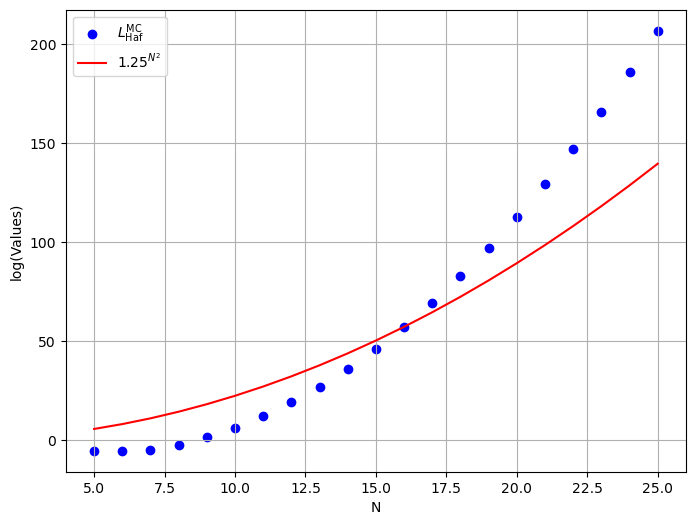}
\caption{Growth rate comparison. Top: the log-log plot of $U_{\Haf}^{'\textnormal{GBS-P}}$ and $N^3$. Bottom: the semi-log plot of $L_{\Haf}^{'\textnormal{MC}}$ and $1.25^{N^2}$.}
\label{fig:gc-haf}
\end{figure}

	\newpage
	\appendix
	\section*{Appendices}
\section{Code accessibility}
Python scripts used to generate the numerical results are available from the GitHub repository: \url{https://github.com/sshanshans/GBSGE}.

\section{Compatibility with GBS hardware}
\label{subsec:prepare}
In order to use samples from the GBS hardware for the Gaussian expectation problem as presented in the text, we must know how to prepare the device to encode the matrix $B$. Since $B$ is symmetric,  it can be diagonalized with a unitary matrix $U$ such that
\begin{align*}
	\matB = \makemat{U} \makemat{D} \makemat{U}^T, ~~~~
	\makemat{D} =  \textnormal{diag}(d_1, \dots, d_N), ~~~~ 1 > d_1 \geq d_2 \geq \dots \geq d_N \geq 0.
\end{align*}
The next theorem specifies the needed Gaussian state detailed in Section 
\ref{subsec:gbsdist} in terms of $U$ and the $d_i$'s. The proof largely follows \cite{hamilton2017gaussian}. 

\begin{theorem}
	\label{thrm:covform}
	Given $B$, $U$ and $D$ as above,
	the Gaussian state $\hat{\rho}(0, \Sigma)$ with
	\begin{equation*}
		\makemat{\Sigma} =
		\frac{1}{2}\begin{bmatrix}
			\makemat{U} & 0 \\
			0 & \makemat{U}
		\end{bmatrix} 
		\makemat{R} \makemat{R}^\intercal
		\begin{bmatrix}
			\makemat{U}^\intercal & 0 \\
			0 & \makemat{U}^\intercal
		\end{bmatrix}
	\end{equation*}
	where 
	\begin{equation*}
		\makemat{R} = \begin{bmatrix}
			\bigoplus_{\indn=1}^N \cosh r_\indn & \bigoplus_{\indn=1}^N \sinh r_\indn \\
			\bigoplus_{\indn=1}^N \sinh r_\indn & \bigoplus_{\indn=1}^N \cosh r_\indn
		\end{bmatrix}, ~~~~  r_\indn \equiv \tanh^{-1} (d_\indn)
	\end{equation*}
	gives rise to $\matA = \matB \oplus \matB.$
\end{theorem}

\begin{remark}
	The assumption that all eigenvalues of $\matB$ must be strictly bounded below 1 is necessary for $r_n$ to be well-defined, as the range of $\tanh$ is $(-1,1)$.
\end{remark}

\begin{proof}
	We begin by noting that 
	in order to produce $C$, we must have
	\begin{align*}
		\makemat{\Sigma}^{-1}_Q = \begin{bmatrix}
			\matI_N & -\matB \\
			-\matB & \matI_N
		\end{bmatrix}.
	\end{align*}
	Equivalently, we must show
	\begin{align}
		\makemat{\Sigma}_Q & = \begin{bmatrix}
			(\matB^2 -\matI_n)^{-1} & 0 \\
			0 & (\matB^2 -\matI_n)^{-1} 
		\end{bmatrix} 
		\begin{bmatrix}
			-\matI_N & -\matB \\
			-\matB & -\matI_N
		\end{bmatrix} \\
		& =
		\begin{bmatrix}
			\makemat{U} \bigoplus_{j=1}^N \cosh^2(r_j)  \makemat{U}^\intercal & 0 \\
			0 & \makemat{U} \bigoplus_{j=1}^N \cosh^2(r_j) \makemat{U}^\intercal
		\end{bmatrix}
		\begin{bmatrix}
			\matI_N & \matB \\
			\matB & \matI_N
		\end{bmatrix} \label{eq:ap11}\\
		& = \begin{bmatrix}
			\makemat{U} \bigoplus_{\indn=1}^N \cosh^2(r_\indn) \makemat{U}^\intercal & \makemat{U} \bigoplus_{\indn=1}^N \cosh(r_\indn) \sinh(r_\indn)  \makemat{U}^\intercal\\
			\makemat{U} \bigoplus_{\indn=1}^N \cosh(r_\indn) \sinh(r_\indn) \makemat{U}^\intercal & \makemat{U} \bigoplus_{\indn=1}^N \cosh^2(r_\indn) \makemat{U}^\intercal
		\end{bmatrix} \label{eq:ap12}
	\end{align}
	In \eqref{eq:ap11}, we used the identity of hyperbolic functions $\cosh^2 = 1/(1 - \tanh^2)$ and the decomposition $B = UDU^\intercal$ to compute the inverse of $\matB^2 -\matI_n$. In  \eqref{eq:ap12}, we used the identity of hyperbolic functions $\cosh(r) \tanh(r) = \sinh(r)$.
	
	We start computing
	\begin{align*}
		\makemat{\Sigma} & = \frac{1}{2} \begin{bmatrix}
			\makemat{U} \bigoplus_{\indn=1}^N \cosh(2r_\indn) \makemat{U}^\intercal & \makemat{U} \bigoplus_{\indn=1}^N  \sinh(2r_\indn)  \makemat{U}^\intercal \\
			\makemat{U} \bigoplus_{\indn=1}^N \sinh(2r_\indn)  \makemat{U}^\intercal & \makemat{U} \bigoplus_{\indn=1}^N \cosh(2r_\indn) \makemat{U}^\intercal
		\end{bmatrix} \\
		& = \begin{bmatrix}
			\makemat{U} \bigoplus_{\indn=1}^N (\cosh^2(r_\indn)-1/2) \makemat{U}^\intercal & \makemat{U} \bigoplus_{\indn=1}^N \cosh(r_\indn) \sinh(r_\indn)  \makemat{U}^\intercal\\
			\makemat{U} \bigoplus_{\indn=1}^N \cosh(r_\indn) \sinh(r_\indn) \makemat{U}^\intercal & \makemat{U} \bigoplus_{\indn=1}^N (\cosh^2(r_\indn)-1/2) \makemat{U}^\intercal
		\end{bmatrix}.
	\end{align*}
	The last equality is derived from
	\begin{equation*}
		\cosh(2r) = 2 \cosh^2(r) - 1, ~~~~ \sinh(2r) = 2\cosh(r)\sinh(r).
	\end{equation*}
	Given $\makemat{\Sigma}_Q = \makemat{\Sigma} + \frac{1}{2}\matI_{2\dimn}$, we get
	\begin{align*}
		\makemat{\Sigma}_Q & = \begin{bmatrix}
			\makemat{U} \bigoplus_{\indn=1}^N \cosh^2(r_\indn) \makemat{U}^\intercal & \makemat{U} \bigoplus_{\indn=1}^N \cosh(r_\indn) \sinh(r_\indn)  \makemat{U}^\intercal\\
			\makemat{U} \bigoplus_{\indn=1}^N \cosh(r_\indn) \sinh(r_\indn) \makemat{U}^\intercal & \makemat{U} \bigoplus_{\indn=1}^N \cosh^2(r_\indn) \makemat{U}^\intercal
		\end{bmatrix},
	\end{align*}
	which is the result we aimed to demonstrate.
\end{proof}

\begin{lemma}
\label{lem:dform}
Let $d$ be given as in \eqref{eq:gbshaf1}. Suppose $B$ is real and has eigenvalues $-1 < \lambda_1, \dots, \lambda_N  < 1$. Then
\begin{align*}
	d = \prod_{\indn=1}^\dimn \sqrt{1 - \lambda^2_j}.
\end{align*}
\end{lemma}

\begin{proof}
	As observed in the proof of Theorem \ref{thrm:covform},
		\begin{align*}
			\makemat{\Sigma}^{-1}_Q 
			= \matI_{2N} - \begin{bmatrix}
				0 & \matI_N \\
				\matI_N & 0 
			\end{bmatrix} \left(\matB \oplus \matB \right) = \begin{bmatrix}
				\matI_N & -\matB \\
				-\matB & \matI_N
			\end{bmatrix}.
		\end{align*}
		Consequently,
		\begin{align*}
			\text{det}\left(\makemat{\Sigma}_Q \right)^{-1} & = \text{det}\left(\makemat{\Sigma}^{-1}_Q \right) \\
			& = \text{det}\left( \begin{bmatrix}
				\matI_N & -\matB \\
				-\matB & \matI_N
			\end{bmatrix} \right) \\
			& = \text{det}\left(\matI_N - \matB^2 \right) \\
			& = \prod_{\indn=1}^\dimn (1 - \lambda^2_j),
		\end{align*}
		which completes the proof of the formula. The block matrix determinant formula is given by
		\begin{align*}
			\label{eq:detblock}
			\det \left( \begin{bmatrix}
				A & B \\
				C & D
			\end{bmatrix} \right) =\det(A)\det \left(D-CA^{-1}B\right)
		\end{align*}
		provided that $A$ is invertible. 
		
\end{proof}

\section{Total computation cost}
It is important to highlight that the comparison made between $n^{\textnormal{GBS-I}}_{\Haf^2}$ and $n^{\textnormal{MC}}_{\Haf^2}$ or $n^{\textnormal{GBS-P}}_{\Haf}$ and $n^{\textnormal{MC}}_{\Haf}$ as stated in the main results inherently favors the plain MC estimator when looking from the perspective of the total computation cost.
In Proposition \ref{prop:cost-MC1} below, we will see that plain MC requires evaluating $f$ at each sampled point, and each evaluation involves a summation over as many terms as the size of the index set, which can also be expensive to compute for large $N$ and $K$. Further, when $K = \infty$, we must truncate the evaluation at a finite $K$, resulting in round-off errors that do not vanish with larger sample size. On the other hand, we will see from Propositions \ref{prop:cost-GBSI} and \ref{prop:cost-GBSP} that when $K$ is finite the GBS estimators display a polynomial runtime in $N$ and $K$ for processing each GBS sample. Moreover, the GBS estimators will avoid truncation when $K = \infty$, as indicated in Corollaries \ref{prop:cost-GBSI2} and \ref{prop:cost-GBSP}. 

\subsection{Use GBS-I for $\mathcal{I}^\times_{\Haf^2}(\epsilon, \delta)$}
In this section, we give an estimate to the total computation cost of using GBS-I to solve $\mathcal{I}^\times_{\Haf^2}(\epsilon, \delta)$. Let us denote it by $C_{\Haf^2}^{\textnormal{GBS-I}}$.

\begin{proposition}
	\label{prop:cost-GBSI}
	Let $K < \infty$. If $a_I$ can be computed or retrieved using constant time steps in the worst case, then
	\begin{align}
		C_{\Haf^2}^{\textnormal{GBS-I}} \leq c_1 N^3 + c_2 K n_{\Haf^2}^{\textnormal{GBS-I}}
	\end{align}
	for some constants $c_1$ and $c_2$. 
\end{proposition}

\begin{proof}
	Notice that
	\begin{align*}
		C_{\Haf^2}^{\textnormal{GBS-I}} \leq C_\text{init} + n_{\Haf^2}^{\textnormal{GBS-I}} \left( C_\text{eval} + C_\text{sample}\right),
	\end{align*}
	where $C_\text{init}$ denotes the cost for initializing the sampling process, $C_\text{eval}$ denotes the cost of evaluating $\frac{d}{I_i!} a_{\vecI_i}$, and $C_\text{sample}$ denotes the cost of a single draw.
	For a covariance matrix $B$ with eigenvalues strictly bounded below 1, programming $B$ onto a GBS device using the Takagi decomposition incurs a runtime complexity of $O(N^3)$, where $N$ is the size of $B$. Hence, $C_\text{init} \leq c_1 N^3$ for some constant $c_1$. Additionally, computing $d$ at this stage costs at most $O(N)$.
	Once the GBS device is configured, obtaining a sample takes constant time, yielding $C_\text{sample} = O(1)$. For $C_\text{eval}$, computing $I!$ for each sampled $I$ takes worst-case $O(K)$ steps, along with computing $a_I$, which requires worst-case $O(1)$ steps. Thus, the cost at this stage amounts to $c_2 K n_{\Haf^2}^{\textnormal{GBS-I}}$ for some constant $c_2$. 
	This completes the proof.
\end{proof}

For $K = \infty$, the analysis is almost the same, except in $C_\text{eval}$ we cannot simply use $O(K)$ as an upper bound for the cost of $I!$, and use the worst case cost of $a_I$. Instead, we give an approximation to the computational cost in the average case.

\begin{corollary}
	\label{prop:cost-GBSI2}
	If $K = \infty$, and suppose $a_I$ can be computed on the fly using $C'_a$ steps. Then, on average, 
	\begin{align}
		C_{\Haf^2}^{\textnormal{GBS-I}} \leq c_1 N^3 + c_2 d \multilog_{-\frac{1}{2}}\left( N^{2} \bmax^{2} \right) n_{\Haf^2}^{\textnormal{GBS-I}}.
	\end{align}
	for some constants $c_1$ and $c_2$. 
\end{corollary}

\begin{proof}
	Let $C_\text{init}$, $C_\text{eval}$  and $C_\text{sample}$ be given as in Proposition \ref{prop:cost-GBSI}.
	We begin by noting that the average number of photons in the GBS samples can be approximated by
	\begin{align*}
		\sum_{k = 0}^\infty (2k) \sum_{\vert I \vert = 2k} p_I
		& =  \sum_{k = 1}^\infty (2k) \sum_{\vert I \vert = 2k} \frac{d}{I!} \Haf(B_I)^2 \\
		& \leq  d \sum_{k = 1}^\infty (2k) \sum_{\vert I \vert = 2k} \frac{1}{I!} \frac{(2k)!^2}{2^{2k} k!^2} \bmax^{2k}\\
		& \hfill \text{(use formula \eqref{eq:usefulhaf})} \\
		& \leq d \sum_{k = 1}^\infty (2k) \frac{N^{2k}}{(2k)!} \frac{(2k)!^2}{2^{2k} k!^2} \bmax^{2k} \\
		& \hfill \text{(use Lemma \ref{lem:computeI})} \\
		& \leq d \sum_{k = 1}^\infty (2k) N^{2k} \frac{1}{\sqrt{\pi k}} \bmax^{2k} \\
		& \hfill \text{(use Lemma \ref{lem:usefulstirling1})} \\
		& = \frac{2d}{\sqrt{\pi}} \sum_{k = 1}^\infty  \sqrt{k} N^{2k} \bmax^{2k} \\
		& = \frac{2d}{\sqrt{\pi}} \multilog_{-\frac{1}{2}}\left( N^{2} \bmax^{2} \right).
	\end{align*}
	Then, it takes on average $O(d \multilog_{-\frac{1}{2}}\left( N^{2} \bmax^{2} \right))$ many steps to compute $I!$. In total, we have $n_{\Haf^2}^{\textnormal{GBS-I}} O( \frac{2d}{\sqrt{\pi}} \multilog_{-\frac{1}{2}}\left( N^{2} \bmax^{2} \right))$ many steps for $C_\text{eval}$.
	The rest of the proof follows similarly to Proposition \ref{prop:cost-GBSI}.
\end{proof}

\subsection{Use GBS-P for $\mathcal{I}^\times_{\Haf}(\epsilon, \delta)$}

Next, we discuss the total computation cost for using GBS-P to solve $\mathcal{I}^\times_{\Haf}(\epsilon, \delta)$, which we denote by $C_{\Haf}^{\textnormal{GBS-P}}$. We further introduce
\begin{align*}
	\mathcal{I}_{\leq K} = \cup_{k = 0}^K \mathcal{I}_k, ~~~~ \mathcal{I}_k = \{ I \mid \vert I \vert = 2k\}. 
\end{align*}
Since $\vert \mathcal{I}_k \vert = \binom{N+2k-1}{2k}$, we get that
\begin{align*}
	 \vert \mathcal{I}_{\leq K} \vert = \sum_{k = 0}^K \binom{N+2k-1}{2k} =: s_{N, K}. 
\end{align*}

\begin{proposition}
	\label{prop:cost-GBSP}
	Let $K < \infty$. If $a_I$ can be computed or retrieved using constant time steps in the worst case, then
	\begin{align}
		C_{\Haf}^{\textnormal{GBS-P}} \leq c_1 N^3 + c_2 \log(s_{N, K}) n_{\Haf}^{\textnormal{GBS-P}}+ c_3 s_{N, K}.
	\end{align}
	for some constants $c_1$, $c_2$ and $c_3$. 
\end{proposition}

\begin{proof}
	With $K < \infty$, we can construct a lookup table to track the appearances of $I$'s from the GBS samples. Clearly, the size of the lookup table is $s_{N, K}$. For each sample, we use binary search to update the table, resulting in a runtime complexity of $O(\log(s_{N,  K}))$. Together, we have $n_{\Haf}^{\textnormal{GBS-P}}O(\log(s_{N, K}))$ for constructing the lookup table. Then, we make a one-time computation to convert each entry in the table into a probability estimate by dividing by the total number of GBS samples, and then to multiply each entry with their corresponding $a_I$. The output is given by the sum of all entries in the table. This step requires $O(s_{N, K})$ runtime. The total computational cost is therefore summarized as in the statement of this proposition. 
\end{proof}

\begin{corollary}
	\label{prop:cost-GBSP2}
	If $K = \infty$, and suppose $a_I$ can be computed using constant steps. Then, on average,
	\begin{align}
		\label{eq:costgbsinifite}
		C_{\Haf}^{\textnormal{GBS-P}} \leq c_1 N^3 + n_{\Haf}^{\textnormal{GBS-P}} \log(n_{\Haf}^{\textnormal{GBS-P}}) + c_2 n_{\Haf}^{\textnormal{GBS-P}}
	\end{align}
	for some constant $c_1$ and $c_2$.
\end{corollary}

\begin{proof}
	When $K = \infty$, we can construct a dynamically-sized lookup table as introduced in Section \ref{sec:probestimator}. The size of the table is at most $n_{\Haf}^{\textnormal{GBS-P}}$. Using a binary search, we can update the table with a runtime complexity of at most $\log(n_{\Haf}^{\textnormal{GBS-P}})$ for each sample. Therefore, the total computational cost for constructing such a table is $n_{\Haf}^{\textnormal{GBS-P}} \log(n_{\Haf}^{\textnormal{GBS-P}})$. Then, as in the $K< \infty$ case, we make a one-time computation to convert the table entries into probability estimates and multiply each entry with their corresponding $a_I$. This step requires $O(n_{\Haf}^{\textnormal{GBS-P}})$ steps. Overall, the computational cost is summarized as in \eqref{eq:costgbsinifite}. 
\end{proof}

\subsection{Use MC for $\mathcal{I}^\times_{\Haf^2}(\epsilon, \delta)$ and $\mathcal{I}^\times_{\Haf}(\epsilon, \delta)$}

Finally, we discuss the total computational costs. Let $C_{\Haf^2}^{\textnormal{MC}}$ denote the computation cost for using $\mathcal{E}^\textnormal{MC}_n$ to solve $\mathcal{I}^\times_{\Haf^2}(\epsilon, \delta)$

\begin{proposition}
	\label{prop:cost-MC1}
	Let $K < \infty$. If $a_I$ can be computed or retrieved using $C_a$ time steps in the worst case, then
	\begin{align}
		C_{\Haf^2}^{\textnormal{MC}} \leq c_1 N^3 + c_2 K \, s_{N, K } \, n_{\Haf^2}^{\textnormal{MC}}
	\end{align}
	for some constants $c_1$ and $c_2$.
\end{proposition}

\begin{proof}
	Let $C_\text{init}$, $C_\text{eval}$  and $C_\text{sample}$ be given as in Proposition \ref{prop:cost-GBSI}.
	First, the initialization stage involves $O(N^3)$ steps for performing Cholesky decomposition in order to prepare sampling from the multivariate Gaussian distribution $h$. Them, for $C_\text{sample}$, producing a single draw costs $O(1)$ steps. To evaluate $f$ at $x = (x_1, \dots, x_N)$, we need to sum $s_{N, K}$ many terms, and for each term the runtime cost is $O(K)$ steps in the worst case. Therefore, the total evaluation cost is $C_\text{eval} \leq O(K) s_{N, K}$. The rest of the proof is similar to  Proposition \ref{prop:cost-GBSI}.
\end{proof}

Immediately, we also have 
\begin{align}
	C_{\Haf}^{\textnormal{MC}} \leq c_1 N^3 + c_2 s_{N,K} \, n_{\Haf}^{\textnormal{MC}}.
\end{align}
When $K = \infty$, the plain MC method will need to use a finite truncation to evaluate $f$. Therefore, the computation cost for the infinite case is the same as the finite case, but there will always be an error made by such compromise.

\section{Hardness of Gaussian expectation}
The exact computation of $\mathcal{I}_{\Haf}$ or $\mathcal{I}_{\Haf^2}$ is worst-case $\#$P-hard, since $\Haf(B)$ is worst-case $\#$P-hard. The currently known fastest algorithm for computing a matrix hafnian is by \cite{bjorklund2019faster} which requires $O(N^32^{N/2})$ time steps for a matrix of size $N$. \cite{bulmer2022boundary} used the finite difference sieve method and reduced the complexity quadratically for when there are repeated rows and columns in the matrix. Nevertheless, the algorithm complexity remains exponential in $N$.
The approximation problem of the matrix hafnians can be somewhat easier. For non-negative matrices, there exist quasi-polynomial deterministic algorithm for approximating the logarithm of the hafnian within additive error \cite{barvinok2017approximating} and randomized polynomial time algorithm to approximate within exponential factor \cite{barvinok1999polynomial} and sub-exponential factor \cite{rudelson2016hafnians}. However, there are no known polynomial (deterministic or randomized) algorithms that can approximate hafnian within multiplicative error for non-negative matrices. In fact, we still don't know if the multiplicative approximation to the hafnian of the non-negative matrices is $\#$P-hard.

There are some instances of  $\mathcal{I}_{\Haf}$ or $\mathcal{I}_{\Haf^2}$ that are known to be easy. For example, when the coefficients $a_I$'s follow a specific design expressed as ${z^I}/{I!}$, as indicated by the Hafnian master theorem \cite{kocharovsky2022hafnian}. Precisely, for any symmetric $2N \times 2N$ matrix $S = S^T$:
\begin{equation}
	\frac{1}{\sqrt{\det (I - ZPS) }} = \sum_{k=0}^\infty \sum_{\sumI = 2k} \frac{z^I}{I!} \Haf(S_{\augI}).
\end{equation}
Here, $z^I = z_1^{i_1}z_2^{i_2} \dots z_N^{i_N}$ and the $2N \times 2N$ matrices $Z$ and $P$ have a $(2 \times 2)$-block structure:
\begin{equation}
	Z = \begin{bmatrix}
		\textnormal{diag}(z) & 0\\
		0 & \textnormal{diag}(z)
	\end{bmatrix}, ~~~~
	P =\begin{bmatrix}
		0 & \mathbb{I}_N\\
		\mathbb{I}_N & 0 
	\end{bmatrix}.
\end{equation}
In this case, the computational complexity is given by the matrix determinant, which is known to be of order $O(N^3).$ This equation holds when $z$ is small such that the real part of $S^{-1} - ZP$ positive definite in Proof 1 and $ I - S^{1/2}ZPS^{1/2}$ in Proof 2 of \cite{kocharovsky2022hafnian}.
It remains an open question whether there is a provably hard problem in the space that we have identified in Theorems \ref{thrm:mainI1} and  \ref{thrm:mainI2}. We will defer this to future exploration.

\section{Use plain MC as a baseline}
Classical numerical integration methods, such as those in the quadrature family, approximate integrals by discretizing the integration domain into grid points and summing weighted function evaluations at those points. These methods are well-tuned for integrals of a few variables but suffer from the curse of dimensionality. High-dimensional quadrature methods are computationally intractable because the number of grid points grows exponentially with the number of variables 
$N$.

Monte Carlo methods, on the other hand, are well suited for approximating high-dimensional integrals. As we have seen from Chebyshev's inequality, their computational cost does not grow with $N$ (at least not in an explicit way). Therefore, to integrate an arbitrary function against a multivariate Gaussian function as described in the Gaussian expectation problem, the plain MC estimator defined in the paper is in fact the most natural choice.

Improvements to MC methods can be made via variance reduction techniques. For example, in importance sampling, one samples from a different probability density distribution in order to reduce the variance of the function to be evaluated and thus improve the efficiency of the MC method. In that sense, GBS-I is indeed a variance reduction technique. Other strategies for variance reduction include the following. 
In antithetic sampling, one exploits symmetry in the integrand function to impose negative correlation among samples. However, since the functions $f$ and $h$ in \eqref{eq:I}, the symmetry only increases the variance.
In stratified sampling, the integration domain is divided into smaller sub-regions, and sampling is enforced from each sub-region. The drawback of stratified sampling strategies is clear: they suffer from the curse of dimensionality.
In Quasi Monte Carlo or its randomized version is about using samples that are spread out more evenly. However, unlike the standard MC, QMC suffers from the curse of dimensionality. 
These methods are standard and covered comprehensively in Monte Carlo textbooks; see, for example, \cite{robert1999monte, rubinstein2016simulation, mcbook2013owen, practicalqmc2023owen} for detailed treatments.
Given the clear disadvantages of these methods, plain MC remains the best choice as a baseline method that can address an arbitrary family of Gaussian expectation problems.
	\section{The Guassian expectation problem related to Boson Sampling}
\label{sec:boson}
In this section, we describe the Gaussian expectation problem in complex variables related to Boson Sampling. 
Let $f: \mathbb{C}^{N} \rightarrow \mathbb{C}$ take the form of 
\begin{equation}
f(z) = \sum_{k=0}^K \sum_{\sumI = k} a'_{\vecI} \vert z \vert^{2I},
\end{equation} 
with $a'_I \in \mathbb{C}$ and $\vert z \vert^{2I} = (z \bar{z})^I$. Here, $\bar{z}$ denotes the complex conjugate of $z$. Recall the complex Gaussian distribution function is defined in the following way. Let
\begin{equation*}
	\mathcal{Z} = \int_{\mathbb{C}^N} \exp( - \bar{z} S z) \, \dd z \wedge \dd \bar z
\end{equation*}
where $S$ is a symmetric invertible complex matrix of size $N \times N$, and
\begin{equation*}
	\dd z \wedge \dd \bar z = \prod_{i= 1}^N \dd z_i \dd \bar{z}_i.
\end{equation*}
It is well known that
\begin{equation*}
	\mathcal{Z} = (-2\pi i)^{N} (\det(S))^{-1}.
\end{equation*}
Then, the complex Gaussian distribution with zero mean and covariance function $B' = S^{-1}$ is given by
\begin{equation}
h(z) = \mathcal{Z}^{-1}  \exp(- \bar{z} S z).
\label{eq:complexgaussian}
\end{equation}
\noindent
The Gaussian weighted problem in complex variables is defined as follows
\begin{align}
\mu_{\textnormal{Per}} = \int_{\mathbb{C}^N} f(z) h(z) \, \dd z \wedge \dd \bar z.
\label{eq:complexI}
\end{align}
It turns out that one can use the second form of Wick's theorem to transform the integral \eqref{eq:complexI} into a weighted sum of a different matrix functions, known as the {\it permanent}. 

Recall the permanent of a matrix is given by the following.

\begin{definition}[Permanent]
	Let $M$ be an aribtrary (real or complex) matrix of size $m \times m$, where $m$ can be even or odd. The permanent of a matrix $M$ is 
	\begin{equation*}
		\Per(M) = \sum_{\sigma \in \mathcal{S}_{m}} \prod_{j = 1}^m M_{j \sigma(j)},
	\end{equation*}
	where $\mathcal{S}_{m}$ is the symmetric group on $m$ elements. The permanent of an empty matrix is 0.
\end{definition}

The second form of Wick's theorem is given below.

\begin{theorem}[Wick \cite{wick1950evaluation}]
	\label{thrm:wick2}
	Let $z = (z_1, \dots, z_N) \in \mathbb{C}^N$ and $\vert z \vert^{2I} = (z \bar{z})^I$. Let $g$ be the zero mean Complex Gaussian distribution \eqref{eq:complexgaussian}. Then
	\begin{equation}
	\int_{\mathbb{C}^N} \vert z \vert^{2I} h(z) \, \dd z \wedge \dd \bar z  = \Per(B'_I).
	\end{equation}	
\end{theorem}
\noindent
In words, the expected value of $\vert z \vert^{2I}$ is the permanent of $B'_I$, where $B'_I$ is the corresponding sub- or super-matrix of $B'$ determined by $I$.

\noindent
We therefore obtain that
\begin{align}
\mu_{\Per}= \sumInf a'_{\vecI} \Per(B'_\vecI),
\end{align}
and one can proceed solving this problem using the general GBS estimators described in Section \ref{sec:gbs-estimators}.

	\section{Technical lemmas}
We give all the technical lemmas used in the paper.

\begin{lemma}
	\label{lem:usefulstirling1}
	For $k \geq 1$,
	\begin{align*}
		\frac{1}{\sqrt{\pi k}}  e^{\frac{1}{25}- \frac{1}{6}} \leq  \frac{(2k)!}{2^{2k}k!^2} \leq \frac{1}{\sqrt{\pi k}}.
	\end{align*}
\end{lemma}

\begin{proof}
	The proof is a direct application of Stirling's formula. We compute that 
	\begin{align*}
		\frac{\sqrt{2\pi 2k} \left( \frac{2k}{e}\right)^{2k} e^{\frac{1}{24k+1}}}{2\pi k \left( \frac{2k}{e}\right)^{2k} e^{\frac{2}{12k}}} \leq \frac{(2k)!}{2^{2k}k!^2} \leq \frac{\sqrt{2\pi 2k} \left( \frac{2k}{e}\right)^{2k} e^{\frac{1}{24k}}}{2\pi k \left( \frac{2k}{e}\right)^{2k} e^{\frac{2}{12k+1}}}
	\end{align*}
	which further simplifies to
	\begin{align*}
		\frac{1}{\sqrt{\pi k}}  e^{\frac{1}{25}- \frac{1}{6}} \leq  \frac{(2k)!}{2^{2k}k!^2} \leq \frac{1}{\sqrt{\pi k}}.
	\end{align*}
\end{proof}

\begin{lemma}
	\label{lem:useful111}
	For all $k \geq 1$, 
	\begin{align*}
		k! \frac{4^{k}}{\sqrt{\pi k}} e^{\frac{1}{25} - \frac{1}{6}} \leq \frac{(2k)!}{k!} \leq k! \frac{	4^{k}}{\sqrt{\pi k}}
	\end{align*}
\end{lemma}

\begin{proof}
This follows from Lemma \ref{lem:usefulstirling1} using  simple algebra. 
\end{proof}

\begin{lemma}
	\label{lem:useful123}
	For all $l \geq 1$,
	\begin{equation*}
		\frac{\sqrt{\pi l}}{2} \leq \displaystyle \sum_{k_1 + k_2 = l} \frac{l!^2}{(2k_1)! (2k_2)!}  \leq \frac{\sqrt{\pi l}}{2} \, e^{\frac{1}{6} - \frac{1}{25}}.
	\end{equation*}
\end{lemma}

\begin{proof}
From Lemma \ref{lem:usefulstirling1}, we get
\begin{align*}
	\frac{(2l)!}{2^{2l}} \sqrt{\pi l} \leq l!^2 \leq \frac{(2l)!}{2^{2l}} \sqrt{\pi l} \, e^{\frac{1}{6} - \frac{1}{25}}.
\end{align*}
Then, 
\begin{align*}
	 \frac{\sqrt{\pi l} }{2^{2l}}\sum_{k_1 + k_2 = l} \frac{(2l)!}{(2k_1)! (2k_2)!} \leq \sum_{k_1 + k_2 = l} \frac{l!^2}{(2k_1)! (2k_2)!} \leq  e^{\frac{1}{6} - \frac{1}{25}} \frac{\sqrt{\pi l}}{2^{2l}} \sum_{k_1 + k_2 = l} \frac{(2l)!}{(2k_1)! (2k_2)!}.
\end{align*}
Notice that 
\begin{align*}
	 \sum_{k_1 + k_2 = l} \frac{(2l)!}{(2k_1)! (2k_2)!} = \sum_{k_1 \geq 0} \binom{2l}{2k_1} = 2^{2l-1}
\end{align*}
which is a well-known result on the sum of even index binomial coefficients. 
After simplification, we get
\begin{align*}
	\frac{\sqrt{\pi l}}{2} \leq \displaystyle \sum_{k_1 + k_2 = l} \frac{l!^2}{(2k_1)! (2k_2)!}  \leq \frac{\sqrt{\pi l}}{2} \, e^{\frac{1}{6} - \frac{1}{25}}.
\end{align*}
\end{proof}

\begin{lemma}
	\label{lem:useful135}
	Let $N$ be given, and let $m = \lceil \frac{N}{2} \rceil$. Let $I = (i_1, \dots, i_N) \in \mathbb{N}^N$ and $\vert I \vert = 2k$ for some integer $k \geq 0$. Then, there exists $L = (l_1, \dots, l_N) \in \mathbb{N}^N$ such that
	\begin{enumerate}
		\item[(1)] $\vert L \vert = 2k + 2m$
		\item[(2)] for any $n \in \{1, 2, \dots, N\}$, $l_n$ is even and $l_n \geq i_n$. 
	\end{enumerate}
\end{lemma}

\begin{proof}
We begin by dividing the tuple entries $\{i_1, \dots, i_N\}$ into the even part and the odd part. We define the corresponding index sets to be
\begin{align*}
	&\mathcal{A}_\text{even} = \{n ~\vert ~i_n \text{ is even. }\} \\
	&\mathcal{A}_\text{odd} = \{n ~\vert~ i_n \text{ is odd. }\}
\end{align*}
Note that $\mathcal{A}_\text{odd}$ may be empty. In that case, we define $L = (l_1, \dots, l_N)$ such that
\begin{align*}
	l_1 = 2m + i_1
\end{align*}
and 
\begin{align*}
	l_n = i_n, ~~~~ n \neq 1.
\end{align*}
It is straightforward that such $L$ satisfy (1) and (2). 

If  $\mathcal{A}_\text{odd}$ is non-empty, then it follows from $\vert I \vert = 2k$ that the size $s = \vert \mathcal{A}_\text{odd} \vert$ must be even. From  $m = \lceil \frac{N}{2} \rceil$, we get $s \leq 2m$. In this case, we define $L = (l_1, \dots, l_N)$ such that
\begin{align*}
	l_n = \begin{cases}
		i_n + 1 & \text{if } n \in \mathcal{A}_\text{odd} \\
		i_n & \text{if } n \in \mathcal{A}_\text{even}
	\end{cases}
\end{align*}
which clearly meets (2). To satisfy (1), let $n' \in \mathcal{A}_\text{odd}$, 
and increase further the value $l_{n'}$ by $2m - s$, i.e., 
\begin{align*}
	l_{n'} = i_{n'} + 1 + 2m -s
\end{align*}
which is again an even number. This completes the proof of this lemma. 
\end{proof}

\begin{lemma}
	\label{lem:computeI}
	Let $I$ be an $N$-tuple. If $m \geq 0$, then
	\begin{equation*}
	\sum_{\sumI = m} \frac{1}{\vecfactorial} = \frac{N^{m}}{m!}.
	\end{equation*}
\end{lemma}
\begin{proof}
	The left-hand side is the coefficient of the term $x^m$ of the product of $N$ copies of $e^x$, i.e., $e^{Nx} = e^x e^x \dots e^x$. To see this, note that
	\begin{align*}
	e^x e^x \dots e^x &= (1 + x + \frac{x^2}{2!}+ \dots) \dots (1 + x + \frac{x^2}{2!}+ \dots)\\
	&= \sum_{m=0}^\infty \sum_{\sumI  = m} \frac{1}{\vecfactorial} x^m.
	\end{align*}
	However,
	\begin{equation*}
	e^{Nx} = 1 + Nx + \frac{N^2x^2}{2!} + \dots
	\end{equation*}
	where $x^m$-th term has coefficient $\frac{N^{m}}{m!}$. 
\end{proof}
	
	\newpage
	\bibliographystyle{plain}
	\bibliography{haf}
	
\end{document}